\documentclass[11pt]{article}
\usepackage{fullpage}
\usepackage{underscore}           

\usepackage{longtable}
\usepackage{amsmath}

 \usepackage{cancel}
\usepackage[utf8]{inputenc}
\usepackage[T1]{fontenc}
\usepackage{amsfonts}
\usepackage{amssymb}
\usepackage{float}
\usepackage{capt-of}
\usepackage{url}
\usepackage{float}
 \usepackage[all]{xy}
 \usepackage{textgreek}
 \usepackage{enumitem}
 \usepackage{dsfont}
 \usepackage{stmaryrd}
\usepackage{pxfonts}
\usepackage{bbm}
\usepackage{hyperref}
\usepackage{tikz}
\usepackage[fancy]{tikz-inet}
\usetikzlibrary{automata}
\usetikzlibrary{arrows}
\usetikzlibrary{shapes}
\usetikzlibrary{decorations.pathmorphing}
\usetikzlibrary{fit}
\usepackage{ifthen}
\usepackage{mathtools}
\usepackage{stmaryrd} 
\usepackage{authblk}
\def\binomp#1#2{\left(#1\atop#2\right)}

\tikzstyle{every picture}=[
  >=stealth', 
shorten >=1pt, 
node distance=1.44cm,
auto,
bend angle=45,
initial text=,
  every state/.style={inner sep=0.75mm, minimum size=1mm},
font=\small,
]

\def\ostar{\textcircled{$\star$}}

\def\sc{\mathrm{sc}}
\def\mon{\mathrm{Mon}}
\newcommand{\IntEnt}[1]{\llbracket #1\rrbracket}

\newtheorem{definition}{Definition}
\newtheorem{lemma}{Lemma}
\newtheorem{example}{Example}
\newtheorem{proposition}{Proposition}
\newtheorem{theorem}{Theorem}

\newtheorem{remark}{Remark}
\newtheorem{property}{Property}
\newtheorem{corollary}{Corollary}

\newenvironment{proof}{\textbf{Proof:}}{\hfill$\square$\newline}

\def\Id{\mathrm{Id}}
\def\Acc{\mathrm{Acc}}
\def\ie{\emph{i.e.} }
\def\eg{\emph{e.g.} }
\def\Val{\mathrm{Val}}
\def\StBool{\mathfrak{StBool}}
\def\Sat{\mathrm{Sat}}
\def\SV{\mathrm{SV}}
\def\LSV{\mathrm{S}_{\ell}\mathrm V}

\def\Red{\mathrm{Red}}
\def\SA{\mathrm{AS}}
\def\Min{\mathrm{Min}}
\def\LST{\mathrm{LST}}
\def\row{\mathrm{row}}
\def\col{\mathrm{col}}
\def\SuppR{\mathrm{SuppR}}
\def\SuppC{\mathrm{SuppC}}
\def\MergeR{\mathrm{MergeR}}
\def\MergeC{\mathrm{MergeC}}
\def\Ind{\mathrm{Ind}}
\title{State complexity of the star of a Boolean operation}
\author[1]{Pascal Caron\thanks{\texttt{Pascal.Caron@univ-rouen.fr}}}
\author[2]{Edwin Hamel-de-le Court\thanks{\texttt{Edwin.Hamel.De.Le.Court@ulb.be}}}
\author[3]{Jean-Gabriel Luque\thanks{\texttt{Jean-Gabriel.Luque@univ-rouen.fr}}}
 
\affil[1]{LITIS, Université de Rouen\\ Avenue de l'Université\\ 76801 Saint-\'Etienne du Rouvray Cedex\\ France}
\affil[2]{Faculté des Sciences\\ Campus de la Plaine - CP 212\\ Boulevard du Triomphe, ACC.2\\ 1050 Bruxelles\\Belgium}
\affil[3]{GR$_2$IF, Université de Rouen\\ Avenue de l'Université\\ 76801 Saint-\'Etienne du Rouvray Cedex\\ France}

\begin{document}
\maketitle
\begin{abstract}
Monsters and modifiers are two concepts recently developed in the state complexity theory. 
 A monster is an automaton in which every function from states to states is represented by at least one letter. A modifier is a set of functions allowing one to transform a set of automata into one automaton. The paper describes a general strategy that can be used to compute the state complexity of many operations.  We illustrate it on the problem of the star of a Boolean operation. After applying  modifiers on  monsters, the states of the resulting automata are assimilated to combinatorial objects: the tableaux. We investigate the combinatorics of these tableaux in order to deduce the state complexity. Specifically,  we  recover the state complexity of  star of  intersection  and  star of  union, and we also  give the exact state complexity of star of  symmetrical difference. 
 We thus harmonize the search strategy for the state complexity of  star of any Boolean operations.
\end{abstract}

\section{Introduction}
The (deterministic) state complexity is a measure on regular languages defined as the size of the minimal automaton. This measure extends to regular operations as a function of the state complexity of inputs. The state complexity of an operation measures the number of states needed in the worst case to encode the resulting language in an automaton.
The classical  approach consists in computing an upper bound  and  providing a witness, that is a specific example showing that the bound is tight.

Pioneered in the $70s$, the state complexity has been investigated for
numerous unary and binary operations. See, for example, \cite{Dom02,GMRY17,JJS05,Jir05,JO08,Yu01a} for a survey. More recently, the state complexity of compositions of operations has also been studied. 
In most cases, the state complexity of a composition is strictly lower than the composition of the state complexities of the individual operations. The studies lead to interesting situation, see \eg \cite{CLP18,CLP19,CGKY11,GSY08,JO11,SSY07}.

In some cases, the classical method has to be enhanced by two independent approaches. The first one consists in describing   states  by  combinatorial objects. Thus the upper bound is computed using combinatorial tools. For instance, in \cite{CLMP15}, the states are represented by tableaux representing Boolean matrices and an upper bound for the catenation of symmetrical difference is given. 
The second one is an algebraic method consisting in 
 %
%
 building a witness for a certain class of regular operations by searching in a set of automata with as many  transition functions as possible. This method has the advantage of being applied to a large class of operations.
This approach has been described independently by Caron \textit{et al.} in \cite{CHLP20}  as the monster approach and by Davies in \cite{Dav18} as the OLPA (One Letter Per Action) approach but was implicitly present in older papers like \cite{BJLRS16,DO09}.

In the authors' formalism, the algebraic aspects are divided into two distinct notions. The first one, called modifiers (see Section \ref{sec-modif}), allows to encode an operation as a transformation on automata. The second one allows one to encapsulate the set of possible transitions in a single object, a monster (see Section \ref{sec-monstate}), which is a $k$-tuple of automata on a huge alphabet. 
We apply this strategy to compute the state complexity of the star of a Boolean operation. After applying the modifier of these operations to the monsters, 
states of the resulting automata are  encoded by combinatorial objects. In our case, the objects are tableaux. 
We then investigate the combinatorics of these tableaux in order to understand the properties of accessibility and Nerode equivalence (see Section \ref{sec-sc}).
To be more precise, in Section \ref{sec-valid}, we define the \emph{validity} property on tableaux which is a combinatorial property indicating the presence of a cross at the coordinate (0,0) under certain conditions. This property is related to the accessibility of the states in the resulting automaton.
We also define, in Section \ref{sec-localsat}, a combinatorial operation, called \emph{local saturation}, on tableaux. This operation consists in adding crosses in an array until some patterns are avoided that are defined in Section \ref{2x2case}. By using preliminary results obtained in Section \ref{section-saturation}, the local saturation operation allows us to describe the Nerode classes in a combinatorial ways.
Finally, we exhibit witnesses among  the set of the monsters in Section \ref{sec-monster}.

Through this method, we  harmonize the search
strategy for the state complexity of star of any Boolean operations. Specifically,
we recover the state complexity of star of intersection and star of union, and we also give the exact state complexity of star of symmetrical difference.
\section{Preliminaries}\label{sect-prel}

\subsection{Languages, automata and algebraic tools}\label{sec-LAAT}
The \emph{cardinality} of a finite set $E$ is denoted by $\#E$,  the \emph{set of subsets} of  $E$ is denoted by $2^E$ and the \emph{set of mappings} of $E$ into itself is denoted by $E^E$. 

Any binary Boolean function $\bullet: \mathbb{B}^2\rightarrow \mathbb{B}$ is extended to a binary operation on sets as $F{\bullet} G=\{x\mid ([x\in F]\bullet [x\in G])=1\}$ for any sets $F$ and $G$. 
For instance, the \emph{symmetric difference} of two sets $F$ and $G$  denoted by $\oplus$ and defined by $F\oplus G=(F\cup G)\setminus (F\cap G)$ is the extension of the Boolean function $\mathrm{xor}$.

For any positive integer $n$, let us denote the set $\{0,\ldots, n-1\}$ by $\llbracket n\rrbracket$. 

Let $\Sigma$ be a finite alphabet. A \emph{word} $w$ over  $\Sigma$ is a finite sequence of symbols of $\Sigma$. The \emph{length} of  $w$, denoted by $|w|$, is the number of occurrences of symbols of $\Sigma$ in $w$.
The set of all finite words over $\Sigma$ is denoted by $\Sigma ^*$. A \emph{language} is a subset of $\Sigma^*$.

A \emph{complete and deterministic finite automaton} (DFA) is a $5$-tuple $A=(\Sigma,Q,q_0,F,\delta)$ where $\Sigma$ is a finite  alphabet, $Q$ is a finite set of states, $q_0\in Q$ is the initial state, $F\subseteq Q$ is the set of final states and $\delta$ is the transition function from  $Q\times \Sigma$ to $Q$ extended in a natural way from $Q\times \Sigma^*$ to $Q$. The cardinal of $A$ is the cardinal of its set of states, \ie $\#A=\#Q$. The automaton $B=(\Sigma,Q_B,i_B,F_B,\delta_B)$ is a \emph{sub-automaton} of $A$ if $Q_B\subset Q_A$, $i_B=i_A$, $F_B\subset F_A$ and $\delta_B=\{(p,a,q)\in \delta_A\mid p,q\in Q_B\}$. 

Let $A=(\Sigma,Q,q_0,F,\delta)$ be a DFA. A word $w\in \Sigma ^*$ is \emph{recognized} by the DFA $A$ if $\delta(q_0,w)\in F$. The \emph{language recognized} by a DFA $A$ is the set $\mathrm L(A)$ of words recognized by $A$. Two DFAs are \emph{equivalent} if they recognize the same language. The set $\mathrm{Rec}(\Sigma^*)$ of all languages recognized by DFAs over $\Sigma$ is the set of \emph{recognizable languages}.
For any word $w$, we denote by $\delta^w$ the function $q\rightarrow\delta(q,w)$. 
A state $q$ is \emph{accessible} if there exists a word $w\in \Sigma ^*$ such that $q=\delta(q_0,w)$.
We denote by $\Acc(A)=(\Sigma, \Acc(Q),q_0,\Acc(F),\delta)$ the sub-automaton of $A$ containing exactly accessible states of $A$. By abuse of languages, we denote also by $\Acc(A)$ the set of accessible states.

Two states $q_1,q_2$ of $A$ are \emph{equivalent} if for any word $w$ of $\Sigma^*$, $\delta(q_1, w)\in F$ if and only if $\delta(q_2, w)\in F$. This equivalence relation is called the \emph{Nerode equivalence} and is denoted by $q_1\sim q_2$. If two states are not equivalent, then they are called \emph{distinguishable}.
 A DFA  is  \emph{minimal} if there does not exist any equivalent  DFA  with less states. It is well known that for any DFA, there exists a unique minimal equivalent DFA (\cite{HU79}). Such a minimal DFA  can be  obtained from $A$ by computing 
 \begin{equation*}
 \Min(A)=\Acc(A)_{/\sim}=(\Sigma,\Acc(Q)/\sim,[q_0],\Acc(F)/\sim,\delta_{\sim})
 \end{equation*}
 where, for any $q\in \Acc(Q)$, $[q]$ is the $\sim$-class of the state $q$ and satisfies the property  $\delta_{\sim}([q],a)=[\delta(q,a)]$, for any $a\in \Sigma$. 
 In a minimal DFA, any two distinct states are pairwise distinguishable. 

The set $\mathrm{Rat}(\Sigma^*)$ of \emph{regular languages} defined over an alphabet $\Sigma$ is the smallest set containing $\{a\}$ for each $a\in\Sigma$ and $\emptyset$ and closed by union, catenation and Kleene star. We recall that the \emph{Kleene star} of the language $L$ is defined by  $L^*=\{w=u_1\cdots u_n\mid  u_i\in L \land n\in\mathbb{N}\}$. Kleene theorem \cite{Kle56} asserts that $\mathrm{Rat}(\Sigma^*)=\mathrm{Rec}(\Sigma^*)$.

A \emph{$k$-ary regular operation} is an operation sending for each alphabet $\Sigma$ any $k$-tuples of regular languages over $\Sigma$ to a regular language over $\Sigma$.
The state complexity of a regular language $L$ denoted by $\sc(L)$ is the number of states of its minimal DFA. This notion extends to regular operations: the state complexity of a unary regular operation $\otimes$ is the function $\sc_{\otimes}$ such that, for all $n\in\mathbb{N}\setminus\{0\}$, $\sc_{\otimes}(n)$ is the maximum of all the state complexities of $\otimes(L)$ when $L$ is of state complexity $n$, \ie $\sc_{\otimes}(n)=\max\{\sc(\otimes(L)) | \sc(L) = n\}$.

This can be generalized, and the state complexity of a $k$-ary operation $\otimes$ is the $k$-ary function $\sc_\otimes$ such that, for all $(n_1,\ldots,n_k)\in (\mathbb N\setminus\{0\})^k$,
\begin{equation}\label{eq-statecomp}
  \sc_\otimes(n_1,\ldots,n_k)=\sup\{\sc(\otimes(L_1,\ldots,L_k))\mid\text{ for all }i\in\{1,\ldots,k\}, \sc(L_i)=n_i\}.
  \end{equation}
When the state complexity is finite, a witness for $\otimes$ is a a way to assign to each $(n_1,\ldots,n_k)$, assumed sufficiently big, a k-tuple of languages $(L_1,\ldots,L_k)$ with $\sc(L_i)=n_i$, for all $i\in\{1,\ldots,k\}$, satisfying $\sc_\otimes(n_1,\ldots,n_k)=\sc(\otimes(L_1,\ldots,L_k))$.
Let $\Sigma$ and $\Gamma$ be two alphabets. A morphism is a function $\phi$ from $\Sigma^*$ to $\Gamma^*$ such that, for all $v,w\in\Sigma^*$, $\phi(vw)=\phi(v)\phi(w)$. Notice that $\phi$ is completely defined by its value on letters.

Let $L$ be a regular language recognized by a DFA $A=(\Gamma,Q,q_0,F,\delta)$ and let $\phi$ be a morphism from $\Sigma^*$ to $\Gamma^*$. Then, $\phi^{-1}(L)$ is the regular language recognized by the DFA $B=(\Sigma,Q,q_0,F,\delta')$ where, for each $a\in\Sigma$ and $q\in Q$, $\delta'(q,a)=\delta(q,\phi(a))$. Therefore, notice that we have
\begin{property}\label{prop-scmorph}
  Let $L$ be a regular language and $\phi$ be a morphism. We have $\sc(\phi^{-1}(L))\leq\sc(L)$.
\end{property}
A morphism $\phi$ is \emph{$1$-uniform} if the image by $\phi$ of any letter is a letter. In other words, a $1$-uniform morphism is a (not necessarily injective) renaming of the letters and the only complexity of the mapping stems from mapping $a$ and $b$ to the same image, i.e., $\phi(a) = \phi(b)$.

A \textit{transformation} of a set $Q$ is a map from $Q$ into itself. The set of transformations endowed with the composition is a monoid where $\Id$, the identity map, is its neutral element. Any transformation $t$ of $Q$ induces a transformation of $2^{Q}$ defined by $E\cdot t= \{q\cdot t\mid q\in E\}$ for any $E\in 2^{Q}$. By extension, any transformation $t_1$ of $Q_1$ and any transformation $t_2$ of $Q_2$ induce a transformation of $2^{Q_1\times Q_2}$ defined by $E\cdot (t_1,t_2)=\{(q_1\cdot t_1,q_2\cdot t_2)\mid (q_1,q_2)\in E\}$ for any $E\in 2^{Q_1\times Q_2}$. 

We consider the special case where $Q=\IntEnt n$. If $t$ is a transformation and $i$ an element of $\IntEnt{n}$, we denote by $i\cdot t$ the image of $i$ under $t$. A transformation of $\IntEnt{n}$ can be represented by $t=[i_0, i_1, \ldots, i_{n-1}]$ which means that $i_k=k\cdot t$ for each $0\leq k\leq n-1$. A \textit{permutation} is a bijective transformation on $\IntEnt{n}$.  A \textit{cycle} of length $\ell\leq n$  is a permutation $t$, denoted   by $(i_0,i_1,\ldots, i_{\ell-1})$, on a subset $I=\{i_0,\ldots ,i_{\ell-1}\}$ of $\IntEnt{n}$  where  $i_k\cdot t=i_{k+1}$ for $0\leq k<\ell-1$, $i_{\ell-1}\cdot t=i_0$ and  for every  elements $j\in \IntEnt{n}\setminus I$ $j\cdot t=j$.  
A \textit{transposition} $t=(i,j)$ is a permutation on $\IntEnt{n}$ where $i\cdot t=j$ and $j\cdot t=i$ and for every  elements $k\in \IntEnt{n}\setminus \{i,j\}$, $k\cdot t=k$.  A \textit{contraction}  $t=\left(\begin{array}{r}i\\j\end{array}\right)$ is a transformation where  $i\cdot t=j$ and  for every  elements $k\in \IntEnt{n}\setminus \{i\}$, $k\cdot t=k$.
We use both notations $t(i)$ or $i\cdot t$ depending on the context.  
\subsection{Monsters and state complexity}\label{sec-monstate}
In \cite{Brz13}, Brzozowski gives a series of properties that would make a language $L_n$ of state complexity $n$ sufficiently complex to be a good candidate for constructing witnesses for numerous classical regular operations. One of these properties is that the size of the syntactic semigroup is $n^n$, which means that each transformation of the minimal DFA of $L_n$ can be associated to a transformation by some non-empty word. This upper bound is reached when the set of transition functions of the DFA is exactly the set of transformations from state to state. We thus consider the set of transformations of $\IntEnt{n}$ as an alphabet where each letter is simply named by the transition function it defines. This leads to the following definition :
\begin{definition}\label{def-1monster}
A $1$-monster  is an automaton
$\mon_{n}^{F}=(\Sigma,\IntEnt{n},0, F,\delta)$ defined by
\begin{itemize}
\item   the  alphabet $\Sigma=\IntEnt{n}^{\IntEnt{n}}$, 
\item the set of states  $\IntEnt{n}$,
\item the initial state $0$,
\item the set of final states $F$,
\item the transition function $\delta$  defined for any $a\in \Sigma$ by $\delta(q,a)=a(q)$.
\end{itemize}
The language recognized by a $1$-monster DFA is called a \emph{$1$-monster language}.
\end{definition}
\begin{example}\label{ex-mon}
  The $1$-monster $\mon_2^{\{1\}}$ is
  
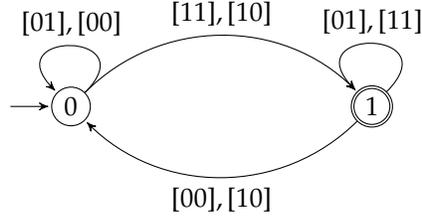
\begin{figure}[H]
  \centering
  \begin{tikzpicture}[node distance=2cm]
    \node[state,initial](p0){$0$};
    \node[state,accepting](p1) at (4,0) {$1$};
    \path[->]
    (p0)edge[loop ] node [swap]{$[01],[00]$} (p0)
    (p0)edge[bend left] node {$[11],[10]$} (p1)
    (p1)edge[loop ] node [swap]{$[01],[11]$} (p1)
    (p1)edge[bend left] node{$[00],[10]$}(p0);
  \end{tikzpicture}
  \end{figure}
  where, for all $i,j\in\{0,1\}$, the label $[ij]$ denotes the transformation sending $0$ to $i$ and $1$ to $j$, which is also a letter in the DFA above.
\end{example}

Let us notice that some families of $1$-monster languages are witnesses for the Star and Reverse operations (\cite{CHLP20}). The following claim is easy to prove and captures a universality-like property of $1$-monster languages:
\begin{property}\label{prop-res}
  Let $L$ be any regular language recognized by a DFA $A=(\Sigma,\IntEnt{n},0, F,\delta)$. The language $L$ is the preimage of $\mathrm L(\mon_n^F)$ by the $1$-uniform morphism $\phi$ such that, for all $a\in\Sigma$, $\phi(a)=\delta^a$, i.e.
  \begin{equation*}
    L=\phi^{-1}(\mathrm{L}(\mon_n^F)).
  \end{equation*}
\end{property}

This is an important and handy property that we should keep in mind. We call it the \emph{restriction-renaming} property.

We can wonder whether we can extend the notions above to provide witnesses for $k$-ary operators. In the unary case, the alphabet of a monster is the set of all possible transformations we can apply on the states. In the same mindset, a $k$-monster DFA is a $k$-tuple of DFAs, and its construction must involve the set of $k$-tuples of transformations as an alphabet. Indeed, the alphabet of a $k$-ary monster has to encode all the transformations acting on each set of states independently one from the others. This leads to the following definition :
 \begin{definition}\label{def-mon}
A $k$-monster  is a $k$-tuple of automata $\mon_{n_1,\ldots, n_k}^{F_1,\ldots, F_k}=(\mathds{M}_1,\ldots, \mathds{M}_k)$ where, for any $j\in \{1,\ldots,k\}$,  $\mathds{M}_j=(\Sigma,\IntEnt{n_j},0, F_j,\delta_j)$  is defined by
\begin{itemize}
\item the common alphabet  $\Sigma=\IntEnt{n_1}^{\IntEnt{n_1}}\times \ldots \times \IntEnt{n_k}^{\IntEnt{n_k}}$, 
\item the set of states  $\IntEnt{n_j}$,
\item the initial state  $0$,
\item the set of final states  $F_j$,
\item the transition function $\delta_j$  defined for any $(a_1,\ldots, a_k)\in \Sigma$ by $\delta_j(q,(a_1,\ldots, a_k))={a_j}(q)$.
\end{itemize}
  A $k$-tuple of languages $(L_1,\ldots,L_k)$ is called a \emph{monster $k$-language} if there exists a $k$-monster $(\mathds{M}_1,\ldots,\mathds{M}_k)$ such that $(L_1,\ldots,L_k)=(\mathrm L(\mathds{M}_1),\ldots,\mathrm L(\mathds{M}_k))$.
\end{definition}
\begin{remark}\label{remark-min}
  When $F_j$ is different from $\emptyset$ and $\IntEnt{n_j}$, $\mathds{M}_j$ is minimal.
\end{remark}
Definition \ref{def-mon} allows us to extend the restriction-renaming property in a way that is still easy to check.
\begin{property}
  Let $(L_1,\ldots,L_k)$ be a $k$-tuple of regular languages over the same alphabet $\Sigma$. We assume that each $L_j$ is recognized by the DFA $A_j=(\Sigma,\IntEnt{n_j},0,F_j,\delta_j)$. Let $\mon_{n_1,\ldots, n_k}^{F_1,\ldots, F_k}=(\mathds{M}_1,\ldots, \mathds{M}_k)$. For all $j\in\{1,\ldots,k\}$, the language $L_j$ is the preimage of $\mathrm{L}(\mathds{M}_j)$ by the $1$-uniform morphism $\phi$ such that, for all $a\in\Sigma$, $\phi(a)=(\delta_1^a,\ldots,\delta_k^a)$, i.e.
  \begin{equation*}
    (L_1,\ldots,L_k)=(\phi^{-1}(\mathrm{L}(\mathds{M}_1)),\ldots,\phi^{-1}(\mathrm{L}(\mathds{M}_k))).
  \end{equation*}
\end{property}
It has been shown that some families of $2$-monsters are witnesses for binary Boolean operations and for the catenation operation \cite{CHLP20}. Many papers concerning state complexity actually use monsters as witnesses without naming them (e.g. \cite{BJLRS16}).
Therefore, a natural question arises : can we define a simple class of regular operations for which monsters are always witnesses~?
This class should ideally encompass some classical regular operations, in particular the operations studied in the papers cited above. The objects that allow us to answer this question are $1$-uniform operations and are defined in the next section.
\subsection{Modifiers and $1$-uniform operations}\label{sec-modif}
A  regular operation $\otimes$ is  \emph{$1$-uniform} if it commutes with any $1$-uniform morphism, \ie  for any $k$-tuple of regular languages $(L_1,\ldots,L_k)$, for any $1$-uniform morphism $\phi$, $\otimes(\phi^{-1}(L_1),\ldots,\phi^{-1}(L_k))=\phi^{-1}(\otimes(L_1,\ldots,L_k))$.
For example, it is proven in \cite{Dav18} that  Kleene star and  mirror operations are $1$-uniform. 

  \begin{property}
    Let $\phi$ (resp. $\psi$) be a $j$-ary (resp. $k$-ary) $1$-uniform operation. Then, for any integer $1\leq p\leq j$, the  $(j+k-1)$-ary  operator $\phi\circ_p\psi$ such that
        \[\phi\circ_p\psi(L_1,\ldots,L_{j+k-1})=\phi(L_1,\ldots,L_{p-1},\psi(L_{p},\ldots,L_{p+k-1}),L_{p+k},\ldots,L_{j+k-1})\]
    is $1$-uniform.
  \end{property}
Monsters are relevant for the study of state complexity of $1$-uniform operations as shown in the theorem below.
\begin{theorem}\label{th-mon2}
  Any $k$-ary $1$-uniform operation admits a family of monster $k$-languages as a witness.
\end{theorem}
\begin{proof}
  Suppose now that $\otimes$ is a $k$-ary $1$-uniform operation. Then, if $(L_1,\ldots,L_k)$ is a $k$-tuple of regular languages over $\Sigma$, $(A_1,\ldots,A_k)$ the $k$-tuple of DFAs such that each $A_j=(\Sigma,Q_j,i_j,F_j,\delta_j)$ is the minimal DFA of $L_i$, and $\phi$ the $1$-uniform morphism such that, for all $a\in\Sigma$, $\phi(a)=(\delta_1^a,\ldots,\delta_k^a)$, and if $\mon_{n_1,\ldots,n_k}^{F_1,\ldots,F_k}=(\mathds{M}_1,\ldots,\mathds{M}_k)$, then $\otimes(L)=\otimes(\phi^{-1}(\mathrm{L}(\mathds{M}_1)),\ldots,\phi^{-1}(\mathrm{L}(\mathds{M}_k)))=\phi^{-1}(\otimes(\mathrm{L}(\mathds{M}_1),\ldots,\mathrm{L}(\mathds{M}_k)))$.
It follows that $\sc(\otimes(L))= \sc(\phi^{-1}(\otimes(\mathrm{L}(\mathds{M}_1),\ldots,\mathrm{L}(\mathds{M}_k))))\leq \sc(\otimes(\mathrm{L}(\mathds{M}_1),\ldots,\mathrm{L}(\mathds{M}_k)))$ by Property \ref{prop-scmorph}. In addition, each $\mathrm{L}(\mathds{M}_j)$ has the same state complexity as $L_j$.
\end{proof}

Each $1$-uniform operation corresponds to a construction over DFAs, which is handy when we need to compute the state complexity of its elements. Such a construction on DFAs has some constraints that are described in the following definitions.
\begin{definition}\label{state-conf}
A \emph{state configuration} is a $3$-tuple $(Q,i,F)$ such that $Q$ is a finite set, $i\in Q$ and $F\subseteq Q$. A \emph{transition configuration} is a $4$-tuple $(Q,i,F,\delta)$ such that $(Q,i,F)$ is a state configuration and $\delta\in Q^Q$. 
 If $A=(\Sigma,Q,i,F,\delta)$ is a DFA, the transition configuration of a letter $a\in\Sigma$ in $A$ is the $4$-tuple $(Q,i,F,\delta^a)$.
 The \emph{state configuration} of  $A=(\Sigma,Q,i,F,\delta)$ is the triplet $(Q,i,F)$.
\end{definition}

\begin{definition}\label{def-mod}
A modifier $\mathfrak{m}$ is a $k$-ary operation over transition configurations such that the state configuration of the output depends only on the state configurations of the inputs. 

It means that if $(t_1,\ldots,t_k)$ and $(t'_1,\ldots,t'_k)$ are two $k$-tuples of transition configurations such that, for any $j\in\{1,\ldots,k\}$, the state configurations of $t_j$ and $t'_j$ are equal, then the state configuration of $\mathfrak{m}(t_1,\ldots,t_k)$ is equal to the state configuration of $\mathfrak{m}(t'_1,\ldots,t'_k)$.
\end{definition}

We thus introduce the following abuse of notation. For any $k$ DFAs $A_1,\ldots, A_k$ over the alphabet $\Sigma$, we let $\mathfrak m(A_1,\ldots,A_k)$ denote the DFA $A$ over $\Sigma$ such that for any letter $a$ of $\Sigma$, the transition configuration of $a$ in $A$ is equal to \begin{equation*}\mathfrak{m}((Q_1,i_1,F_1,\delta_1^a),\ldots,(Q_k,i_k,F_k,\delta_k^a)).\end{equation*}

\begin{example}[The unary modifier $\mathfrak{Star}$]\label{ex-star}
  For any DFA $A=(\Sigma,Q,i,F,\delta)$, let us  define \begin{equation*}\mathfrak{Star}(A)=(\Sigma,2^{Q},\emptyset,\{E| E\cap F \neq \emptyset\}\cup\{\emptyset\},\delta_1),\end{equation*} where $\delta_1$ is as follows : 
  \begin{equation*}\text{for all } a\in\Sigma,\   \delta_1^a(\emptyset)=\left\{\begin{array}{ll}\{\delta^a(i)\}\text{ if }\delta^a(i)\notin F \\
  \{\delta^a(i),i\}\mbox{ otherwise }
  \end{array}\right.
  \text{for all }  E\neq\emptyset,\;
  \delta_1^a(E)=\left\{\begin{array}{ll}\delta^a(E)\text{ if }\delta^a(E)\cap F=\emptyset \\
  \delta^a(E)\cup\{i\}\mbox{ otherwise }
  \end{array}\right.\end{equation*}
\end{example}
The modifier $\mathfrak{Star}$ describes the classical construction on DFA associated to the star operation on languages, \ie for all DFA $A$, $\mathrm L(A)^*=\mathrm L(\mathfrak{Star}(A))$.
\begin{example}[The binary modifier $\mathfrak{Bool}_\bullet$]\label{ex-xor}
Consider a Boolean binary operation denoted by $\bullet$.
  For all DFAs $A=(\Sigma,Q_1,i_1,F_1,\delta_1)$ and $B=(\Sigma,Q_2,i_2,F_2,\delta_2)$, define
  \begin{equation*}\mathfrak{Bool}_\bullet(A,B)=(\Sigma,Q_1\times Q_2,(i_1,i_2),(F_1\times Q_2)\bullet (Q_1\times F_2),(\delta_1,\delta_2))\end{equation*}
  The modifier  $\mathfrak{Bool}_\bullet$ describes the classical construction associated to the operation $\bullet$ on couples of languages, \ie for all DFAs $A$ and $B$, $\mathrm L(A)\bullet \mathrm L(B)=\mathrm L( \mathfrak{Bool}_\bullet(A,B))$.
\end{example}

\begin{definition}\label{def-coherent}
  A $k$-ary modifier $\mathfrak m$ is \emph{coherent} if, for every pair of $k$-tuples of DFAs $(A_1,\ldots,A_k)$ and $(B_1,\ldots,B_k)$ such that $\mathrm{L}(A_j)=\mathrm{L}(B_j)$ for all $j\in \{1,\dots,k\}$, we have $\mathrm{L}(\mathfrak m(A_1,\ldots,A_k))=\mathrm{L}(\mathfrak m(B_1,\ldots,B_k))$. 
 \end{definition}
If $\mathfrak m$ is a coherent modifier then we denote $\otimes_{\mathfrak m}$ the operation such that for any $k$-tuple of regular languages  $(L_1,\cdots, L_k)$ we have  $\otimes_{\mathfrak m}(L_1,\dots,L_k)=\mathrm{L}(\mathfrak m(A_1,\ldots,A_k))$ for any $k$-tuple of DFAs $(A_1,\dots,A_k)$ such that each $A_i$ recognizes $L_i$.
\begin{theorem}\label{th-mod}
  A  $k$-ary regular operation $\otimes$ is $1$-uniform if and only if $\otimes=\otimes_{\mathfrak m}$ for some coherent modifier $\mathfrak m$.
\end{theorem}
\begin{proof}
Let $\otimes$ be a $1$-uniform operation. For any $k$-tuple $C=((n_1,F_1),\cdots,(n_k,F_k))$ with $n_i\in\mathbb N$ and $F_i\subset\IntEnt{n_i}$ for any $i\in\{1,\dots,k\}$, we define
\begin{equation*}
B_C=(\IntEnt{n_1}^{\IntEnt{n_1}}\times\cdots\times\IntEnt{n_k}^{\IntEnt{n_k}},Q_C,i_C,F_C,\delta_C)
\end{equation*}
as the minimal automaton that recognizes $\otimes(L(\mathbb M_1),\dots,L(\mathbb M_k))$ with $\mon^{F_1,\dots,F_k}_{n_1,\dots,n_k}=(\mathbb M_1,\cdots,\mathbb M_k)$.\\
Let $T=((Q_1,i_1,F_1,\delta_1),\dots,(Q_k,i_k,F_k,\delta_k))$ be a $k$-tuple of transition configurations. By the axiom choice, we associated to any $(Q,i)$, where $Q$ is a finite set and $i\in Q$, a bijection $\varphi_{Q,i}:Q\rightarrow\IntEnt{\#Q}$ sending $i$ to $0$. To any $f:Q\rightarrow Q$, we associate the unique map $\hat f^{i}:\IntEnt{\#Q}\rightarrow \IntEnt{\# Q}$ such that $\hat f^i(\varphi_{Q,i}(q))=\varphi_{Q,i}(f(q))$ for any $q\in Q$. 
When there is no ambiguity, we omit the subscript and denote $\hat f=\hat f^i$.\\
We associate to $T$, the automaton 
\[
A^\otimes_T=B_{((\#Q_1,\varphi_{Q_1,i_1}(F_1)),\dots,(\#Q_k,\varphi_{Q_k,i_k}(F_k)))}=
(\IntEnt{\#Q_1}^{\IntEnt{\#Q_1}}\times\cdots\times\IntEnt{\#Q_k}^{\IntEnt{\#Q_k}},
Q_T,i_T,F_T,\delta_T),
\]
and we set
\[
\mathfrak m(T)=(Q_T,i_T,F_T,\delta)
\]
where $\delta(q)=\delta_T^{(\widehat{\delta_1},\dots,\widehat{\delta_k})}(q)$.\\
Obviously $(Q_T,i_T,F_T)$ depends only on the $k$-tuple $((Q_1,i_1,F_1),\dots,(Q_k,i_k,F_k))$ and so $\mathfrak m$ is a modifier.\\
  Furthermore, by construction, for any $k$-tuple of automata $(A_1,\ldots,A_k)$ with $A_j=(\Sigma,Q_j,i_j,F_j,\delta_j)$ we have
  \[\mathrm L(\mathfrak m(A_1,\dots, A_k))=\phi^{-1}(\mathrm L(B_{((\#Q_1,\varphi_{Q_1,i_1}(F_1)),\dots,(\#Q_k,\varphi_{Q_k,i_k}(F_k)))})),\] where $\phi$ is the $1$-uniform morphism such that $\phi(a)=(\widehat{\delta_1^a},\dots,
  \widehat{\delta_k^a})$ for all $a\in\Sigma$ . Therefore, we have $\mathrm L(\mathfrak m (A_1,\dots,A_k))=\phi^{-1}(\otimes(\mathrm L(\mathbb M_1),\dots,L(\mathbb M_k)))$ where $(\mathbb M_1,\dots,\mathbb M_k)=\mon^{\varphi_{Q_1,i_1}(F_1),\dots,\varphi_{Q_k,i_k}(F_k)}_{\#Q_1,\dots,\#Q_k}$. Since $\otimes$ is $1$-uniform, we obtain \[\mathrm L(\mathfrak m (A_1,\cdots,A_k))=\otimes(\phi^{-1}(\mathrm L(\mathbb M_1)),\dots,\phi^{-1}(\mathrm L(\mathbb M_k)))=\otimes(L(A_1),\dots,L(A_k)).\]
  
Conversely, let $\otimes$ be a $k$-ary regular operation and let $\mathfrak m$ be a $k$-ary modifier such that for any $k$-tuple of regular languages $(L_1,\dots,L_k)$ and any $k$-tuple of DFA $(A_1,\dots, A_k)$ such that each $A_i$  recognizes $L_i$, $\otimes(L_1,\dots,L_k)=\mathrm L(\mathfrak m(A_1,\dots,A_k))$. Let us  prove that $\otimes$ is $1$-uniform. Let $\Gamma$ and $\Sigma$ be two alphabets. Consider a 1-uniform morphism $\phi$ from $\Gamma^*$ to $\Sigma^*$ and $(L_1,\dots,L_k)$ a $k$-tuple of languages over $\Sigma$. For any $j$, we consider an automaton $A_j=(\Sigma,Q_j,i_j,F_j,\delta_j)$ that recognizes $L_j$ and let $B_j=(\Gamma,Q_j,i_j,F_j,\tilde \delta_j)$ be the DFA such that $\tilde \delta^a_j=\delta_j^{\phi(a)}$ for any letter $a\in\Gamma$. We have $\mathrm L(B_j)=\phi^{-1}(\mathrm L(A_j))$.

Let $\mathfrak m(A_1,\dots,A_k)=(\Sigma,Q,i,F,\delta)$ and $\mathfrak m(B_1,\dots,B_k)=(\Gamma,Q',i',F',\delta')$. Since for any $j$ the state configuration of $A_j$ is the same as the state configuration of $B_j$, we have $(Q,i,F)=(Q',i',F')$. Furthermore, because the transition function of any letter $a\in\Gamma$ in $B_j$ is also the same as the transition function of $\phi(a)$ in $A_j$, we have $\delta'^a=\delta^{\phi(a)}$. Hence, $\mathrm L(\mathfrak m(B_1,\dots,B_k))=\phi^{-1}(\mathrm L(\mathfrak m(A_1,\dots,A_k)))$, which implies that $\otimes(\mathrm L(B_1),\dots,\mathrm L(B_k))=(\phi^{-1}(\otimes(A_1)),\dots,\phi^{-1}(\otimes(A_k)))$. Therefore, \[ \otimes(\phi^{-1}(L_1,\dots,L_k))=\otimes(\phi^{-1}(\mathrm L(A_1)),\dots,\phi^{-1}(\mathrm L(A_k)))=\phi^{-1}(\otimes(\mathrm L(A_1)),\dots,\mathrm L(A_k))))=\phi^{-1}(\otimes(L_1,\dots,L_k))\] as expected.
\end{proof}

\section{The combinatorics of the star of a binary Boolean operation}\label{sec-sc}

We list in the Table  \ref{op-bool} below the 16 possible Boolean functions and we write what they correspond to in terms of operation on languages.
\begin{equation*}L_1\bullet L_2=\{z\mid [z\in L_1]\bullet [z\in L_2]=1\}.\end{equation*}

 \begin{table}[H]
  \centerline{\begin{tabular}{p{0.70cm}|p{0.45cm}p{0.45cm}p{0.45cm}p{0.45cm}p{0.45cm}p{0.45cm}p{0.45cm}p{0.45cm}p{0.45cm}p{0.45cm}p{0.45cm}p{0.45cm}p{0.45cm}p{0.45cm}p{0.45cm}p{0.45cm}}
   &\rotatebox{45}{$\emptyset$}   & 
   \rotatebox{45}{$L_1\cap L_2$}  &
   \rotatebox{45}{$L_1 \cap \overline{L_2}$}  & 
   \rotatebox{45}{$L_1$}  & 
   \rotatebox{45}{$ \overline{L_1} \cap  L_2$} & 
   \rotatebox{45}{$L_2$}  & 
   \rotatebox{45}{$L_1 \oplus L_2$}  &
   \rotatebox{45}{$L_1 \cup  L_2$}  & 
   \rotatebox{45}{$\overline{L_1} \cap  \overline{L_2}$}  & 
   \rotatebox{45}{$ \overline{L_1} \oplus L_2$}  &
   \rotatebox{45}{$ \overline{L_2}$}  & 
   \rotatebox{45}{$L_1 \cup \overline{L_2}$}  & 
   \rotatebox{45}{$\overline{L_1}$}  & 
   \rotatebox{45}{$\overline{L_1} \cup  L_2$}  & 
   \rotatebox{45}{$\overline{L_1} \cup  \overline{L_2}$}  & 
   \rotatebox{45}{$\Sigma^*$}  \\\hline
  $0\bullet 0$&$0$ & $0$ & $0$ & $0$ & $0$ & $0$ & $0$ & $0$ & $1$ & $1$ & $1$ & $1$ & $1$ & $1$ & $1$ & $1$\\
   $0\bullet 1$&$0$ & $0$ & $0$ & $0$ & $1$ & $1$ & $1$ & $1$ & $0$ & $0$ & $0$ & $0$ & $1$ & $1$ & $1$ & $1$\\
  $1\bullet 0$&$0$ & $0$ & $1$ & $1$ & $0$ & $0$ & $1$ & $1$ & $0$ & $0$ & $1$ & $1$ & $0$ & $0$ & $1$ & $1$\\
 $1\bullet 1$&$0$ & $1$ & $0$ & $1$ & $0$ & $1$ & $0$ & $1$ & $0$ & $1$ & $0$ & $1$ & $0$ & $1$ & $0$ & $1$\\
  \end{tabular}}
  \caption{The 16 binary Boolean functions.}
  \label{op-bool}
  \end{table}
From Table \ref{op-bool}, we observe that there are only $10$ non degenerated operations that depend truly of the two operands $L_1$ and $L_2$. These operations are the only ones that we have to consider. Table \ref{op-bool} also shows that each of the $10$ non degenerated operations can be obtained by acting with union, intersection or xor on two languages or their complementary.

We use the symbol $\ostar$ to denote any operation defined by 
\begin{equation}\label{eq-ostar}
    L_1\ostar L_2=(L_1\bullet L_2)^*.
\end{equation}

Examples \ref{ex-star} and \ref{ex-xor} together with Theorem \ref{th-mod} show that any $\ostar$ is $1$-uniform. We define the modifier \begin{equation}\label{eq-StBoolMod}\StBool_\bullet=\mathfrak{Star}\circ \mathfrak{Bool}_\bullet.\end{equation} 
To be more precise, let us describe how the modifier $\StBool_\bullet$ acts on automata: if $A_1=(\Sigma, Q_1,i_1,F_1,\delta_1)$ and $A_2=(\Sigma,Q_2,i_2,F_2,\delta_2)$, then
\begin{equation*}\StBool_\bullet(A_1,A_2)=(\Sigma,2^{Q_1\times Q_2},\emptyset,\{E\in 2^{Q_1 \times Q_2}\mid E\cap F\neq\emptyset \}\cup\{\emptyset\},\delta)\end{equation*}
where $F=(F_1\times Q_2)\bullet(Q_1\times F_2)$ and, for all $a\in \Sigma$,
\begin{equation*}\delta^a(\emptyset)=\left\{\begin{array}{ll}\{(\delta_1^a(i_1),\delta_2^a(i_2))\}\text{ if }(\delta_1^a(i_1),\delta_2^a(i_2))\notin F \\
  \{(\delta_1^a(i_1),\delta_2^a(i_2)),(i_1,i_2)\}\mbox{ otherwise }
\end{array}\right.\end{equation*}
\begin{center}
and, for all $E\neq\emptyset$, $\delta^a(E)=\left\{\begin{array}{ll}E\cdot (\delta_1^a,\delta_2^a)\text{ if }E\cdot (\delta_1^a,\delta_2^a)\cap F=\emptyset \\
E\cdot (\delta_1^a,\delta_2^a)\cup\{(i_1,i_2)\}\mbox{ otherwise. }
\end{array}\right.$
\end{center}

Theorem \ref{th-mon2} states that any $\ostar$ admits a family of $2$-monsters as witness. For any positive integers $m,n$, let $(\mathds M_1,\mathds M_2)=\mon_{m,n}^{\{m-1\},\{0\}}$. We are going to show that, for all $(m,n)\in (\mathbb N\setminus \{0\})^2$, $(\mathrm L(\mathds M_1)),\mathrm L(\mathds M_2))$ is indeed a witness for any $\ostar$. 
For any positive integers $m,n$, any $F_1\subseteq\IntEnt{m}$, and $F_2\subseteq\IntEnt{n}$,
let us denote by $\mathrm{M}^{F_1,F_2}_\bullet$ the DFA $\StBool_\bullet(\mon_{m,n}^{F_1,F_2})$. 

We identify elements of $2^{\IntEnt{m}\times\IntEnt{n}}$ to Boolean matrices of size $m\times n$. Such a matrix is called a tableau when crosses are put in place of $1$s, and $0$s are erased.
The set of such tableaux is denoted by $\mathcal{T}_{m,n}$. Each tableau of $\mathcal{T}_{m,n}$ labels a state of $\mathrm{M}^{F_1,F_2}_\bullet$.

If $T$ is an element of $\mathcal{T}_{m,n}$, we denote by $T_{x,y}$ the value of the Boolean matrix $T$ at row $x$ and column $y$. Therefore,  the three following assertions mean the same thing : a cross is at the coordinates $(x,y)$, $T_{x,y}=1$, and $(x,y)\in T$.
We denote by $\mathrm{IsFinal}(x,y)$ the Boolean value $[i\in F_1]\bullet [j\in F_2]=[(i,j)\in F]$.
The cell $(x,y)$ of $T$ is \emph{final} if and only if $\mathrm{IsFinal}(x,y)=1$. 
The tableau $T$ is \emph{final} if there exists $(i,j)\in \IntEnt{m}\times\IntEnt{n}$ such that $T_{x,y}=1$ and $\mathrm{IsFinal}(x,y)=1$.
The \emph{final zone} of the tableau $T$ is the set of all his final cells, \textit{i.e.} $\mathrm{FinalZone}(T)=\{(i,j)\in \IntEnt{m}\times\IntEnt{n}\mid \mathrm{IsFinal}(i,j)\}$. When there is no ambiguity, we refer to final zone for the final zone of any tableau labelling a state of the  automaton we consider.

Notice that the final zone of $T$ is exactly the set $F=(F_1\times \IntEnt{n})\bullet(\IntEnt{m}\times F_2)$ of the modifier $\StBool_\bullet$.
Notice also that a tableau $T\neq \emptyset$
 is final if and only if it labels a non empty final state in the automaton $\mathrm{M}_\bullet^{F_1,F_2}$.
 
\subsection{Valid states and accessibiliy\label{sec-valid}}
 We  refer to the final zone for both the final zone of $T$ and of $T'$ because these two tableaux label a state of the same automaton.
\begin{definition}\label{definition validity}
A tableau $T$ labelling a state of $\mathrm{M}_{\bullet}^{F_1,F_2}$ is \emph{valid} if either its final zone is empty or if $T_{(0,0)}=1$.
\end{definition} 
Let $\Val^{F_1,F_2}_\bullet$ be the set of  valid tableaux labeling states of $\mathrm{M}_\bullet^{F_1,F_2}$ and     $\Val(\mathrm{M}^{F_1,F_2}_\bullet)$ be the restriction of the DFA $\mathrm{M}_{\bullet}^{F_1,F_2}$  to the states of $\Val^{F_1,F_2}_\bullet$. 
The aim of this section is to compare the set $\Val^{F_1,F_2}_\bullet$ to the set, denoted by $\Acc_{\bullet}^{F_1,F_2}$,  of accessible states of $\mathrm{M}_{\bullet}^{F_1,F_2}$.
Straightforward from the definition of the modifier $\StBool_\bullet$, the set of accessible states are valid. Nevertheless, we show that the    accessibility of the valid states  depends on the final zone.

For any tableau $T$, we define the finest equivalence relation $\Delta$  on $T$  satisfying
$(i,j)\Delta (i,j')$ and $(i,j)\Delta (i',j)$.
\begin{lemma}\label{Delta}
Let $T\neq \emptyset$ be a tableau representing a state of $\mathrm{M}_\bullet^{F_1,F_2}$.
        Let us suppose that for all $i\neq 0$ and for all $j\neq 0$, $(i,j)$ is not in the final zone of $T$. Then, we have $$T\text{ is accessible implies }\#T/_\Delta=1.$$  
\end{lemma}
\begin{proof}
 Let us prove the result by induction on the minimal words $w$ such that $\delta^w(\emptyset)=T$.\\
 If $w\in \Sigma$ then  either $\#T=1$ or $T=\{(i,0),(0,0)\}$ for some $i$ or $T=\{(0,j),(0,0)\}$ for some $j$. Obviously, in the three cases $\#T/_\Delta=1$.\\
 Suppose now the result true for any strict prefix of $w=w'a$ with $a\in\Sigma$. Let $T'=\delta^{w'}(\emptyset)$. By induction $\#T'/_{\Delta}=1$. We also have $T=\delta^a(T')$.\\
 First consider the tableau $\hat T=T'.(\delta_1^a,\delta_2^a)$. If $(i,j)\Delta (i',j')$ in $T'$ then $(\delta_1^a(i),\delta_1^a(i))\Delta(\delta_1^a(i),\delta_1^a(i))$. Indeed, this is obtained by transitivity considering that $$(i,j), (i,j')\in T'\Rightarrow (\delta^a_1(i),\delta^a_2(j)), (\delta^a_1(i),\delta^a_2(j'))\in T'$$ and 
 $$(i,j), (i',j)\in T'\Rightarrow (\delta^a_1(i),\delta^a_2(j)), (\delta^a_1(i'),\delta^a_2(j))\in T'.$$
 So by induction $\#T'/_\Delta=\#\hat T/_\Delta=1$.\\
 If $T\neq \hat T$ then $T=\hat T\cup\{(0,0)\}$ and there exists $(i,j)\in T$ in the final zone. But this implies that $i=0$ or $j=0$, so $(0,0)\Delta (i,j)$ and so $\#T/_\Delta=1$. 
\end{proof}

\begin{proposition}\label{ValnotAcc}
    When the final zone is included in $(\IntEnt{m}\times\{0\})\cup(\{0\}\times \IntEnt{n}) $, we have \[\Acc_{\bullet}^{F_1,F_2}\subsetneq \Val_{\bullet}^{F_1,F_2}.\]
\end{proposition}
\begin{proof}
Straightforward from the definition of the modifier $\StBool_\bullet$, we have $\Acc_{\bullet}^{F_1,F_2}\subset \Val_{\bullet}^{F_1,F_2}$.\\
  In the context of the proposition,  the tableau $\{(0,0),(1,1)\}$ is valid but $\Delta$ has two distinct classes. Hence, from Lemma \ref{Delta}, it is not accessible.
\end{proof}

Next definition defines an order $<$ on tableaux of $\mathcal{T}_{m,n}$ labelling the states of $\mathrm{M}_\bullet^{F_1,F_2}$.
\begin{definition}\label{def-order}
 Let $T$ and $T'$ be two  tableaux of $\mathcal{T}_{m,n}$.  The number of cross in the non-final zone of $T$ is denoted by  $\mathrm{nf}(T)$.\\
 
 We write $T<T'$
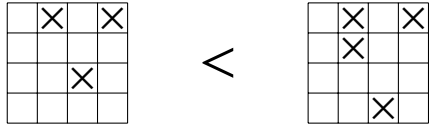
 \begin{enumerate}
     \item  if $\#T<\#T'$.\\
          \centerline{
            \begin{tikzpicture}[scale=0.4]
              \draw[step=1.0,black, thin] (0,0) grid (4,4);
              \node[scale=1.5] at (2.5,1.5) {$\times$};
              \node[scale=1.5] at (3.5,3.5) {$\times$};
              \node[scale=1.5] at (1.5,3.5) {$\times$};
               \node[scale=1.5] at (7,2)  {\Large $<$};
              \draw[step=1.0,black, thin] (10,0) grid (14,4);
              \node[scale=1.5] at (11.5,2.5) {$\times$};
              \node[scale=1.5] at (13.5,3.5) {$\times$};
              \node[scale=1.5] at (12.5,0.5) {$\times$};
              \node[scale=1.5] at (11.5,3.5) {$\times$};
          \end{tikzpicture}}
     \item  if $\#T=\#T'$ and if $\mathrm{nf}(T)<\mathrm{nf}(T')$.\\
     
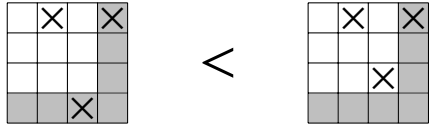
\centerline{
            \begin{tikzpicture}[scale=0.4]
                          \node[scale=1.5,fill=lightgray] at (3.5,3.5) {$ $};
              \node[scale=1.5,fill=lightgray] at (3.5,2.5) {$ $};
              \node[scale=1.5,fill=lightgray] at (3.5,1.5) {$ $};
              \node[scale=1.5,fill=lightgray] at (3.5,0.5) {$ $};
              \node[scale=1.5,fill=lightgray] at (2.5,0.5) {$ $};
              \node[scale=1.5,fill=lightgray] at (1.5,0.5) {$ $};
              \node[scale=1.5,fill=lightgray] at (0.5,0.5) {$ $};
              \node[scale=1.5,fill=lightgray] at (13.5,3.5) {$ $};
              \node[scale=1.5,fill=lightgray] at (13.5,2.5) {$ $};
              \node[scale=1.5,fill=lightgray] at (13.5,1.5) {$ $};
              \node[scale=1.5,fill=lightgray] at (13.5,0.5) {$ $};
              \node[scale=1.5,fill=lightgray] at (12.5,0.5) {$ $};
              \node[scale=1.5,fill=lightgray] at (11.5,0.5) {$ $};
              \node[scale=1.5,fill=lightgray] at (10.5,0.5) {$ $};
              \draw[step=1.0,black, thin] (0,0) grid (4,4);
              \node[scale=1.5] at (2.5,0.5) {$\times$};
              \node[scale=1.5] at (3.5,3.5) {$\times$};
              \node[scale=1.5] at (1.5,3.5) {$\times$};
               \node[scale=1.5] at (7,2)  {\Large $<$};
              \draw[step=1.0,black, thin] (10,0) grid (14,4);
              \node[scale=1.5] at (13.5,3.5) {$\times$};
              \node[scale=1.5] at (12.5,1.5) {$\times$};
              \node[scale=1.5] at (11.5,3.5) {$\times$};
          \end{tikzpicture}}
     \item if $\#T=\#T'$,  $\mathrm{nf}(T)=\mathrm{nf}(T')$ and $T_{0,0}=1$ while $T'_{0,0}=0$.\\
     
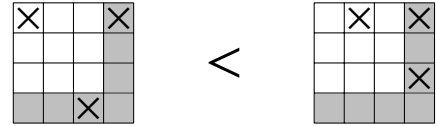
\centerline{
            \begin{tikzpicture}[scale=0.4]
                          \node[scale=1.5,fill=lightgray] at (3.5,3.5) {$ $};
              \node[scale=1.5,fill=lightgray] at (3.5,2.5) {$ $};
              \node[scale=1.5,fill=lightgray] at (3.5,1.5) {$ $};
              \node[scale=1.5,fill=lightgray] at (3.5,0.5) {$ $};
              \node[scale=1.5,fill=lightgray] at (2.5,0.5) {$ $};
              \node[scale=1.5,fill=lightgray] at (1.5,0.5) {$ $};
              \node[scale=1.5,fill=lightgray] at (0.5,0.5) {$ $};
              \node[scale=1.5,fill=lightgray] at (13.5,3.5) {$ $};
              \node[scale=1.5,fill=lightgray] at (13.5,2.5) {$ $};
              \node[scale=1.5,fill=lightgray] at (13.5,1.5) {$ $};
              \node[scale=1.5,fill=lightgray] at (13.5,0.5) {$ $};
              \node[scale=1.5,fill=lightgray] at (12.5,0.5) {$ $};
              \node[scale=1.5,fill=lightgray] at (11.5,0.5) {$ $};
              \node[scale=1.5,fill=lightgray] at (10.5,0.5) {$ $};
              \draw[step=1.0,black, thin] (0,0) grid (4,4);
              \node[scale=1.5] at (2.5,0.5) {$\times$};
              \node[scale=1.5] at (3.5,3.5) {$\times$};
              \node[scale=1.5] at (0.5,3.5) {$\times$};
               \node[scale=1.5] at (7,2)  {\Large $<$};
              \draw[step=1.0,black, thin] (10,0) grid (14,4);
              \node[scale=1.5] at (13.5,3.5) {$\times$};
              \node[scale=1.5] at (13.5,1.5) {$\times$};
              \node[scale=1.5] at (11.5,3.5) {$\times$};
          \end{tikzpicture}}
 \end{enumerate}
\end{definition}

 \begin{proposition}\label{prop-acc}
  If
  the final zone of $\mathrm{M}^{F_1,F_2}_\bullet$ is not a subset of  $(\{0\}\times \IntEnt{n})\cup(\IntEnt{m}\times \{0\})$ then  
  any of its non empty state 
  is accessible if and only if it is a valid tableau.
\end{proposition}
\begin{proof}
From the definition of the transition function of $\StBool_\bullet$, every  non empty  tableau $T$  which is not valid, \ie with $T_{0,0}=0$ and a cross in the final zone, is not accessible.
  
  Now we prove the converse by induction on the order $<$ defined above. \\
  The initial case is obvious since the tableau $\{(0,0)\}$ is accessible from $\emptyset$ by reading the letter  $(\Id,\Id)$. 
  Let $T'>\{(0,0)\}$ be a valid tableau. We exhibit a letter $(f,g)$  together with a valid tableau $T<T'$ such that $\delta^{(f,g)}(T)=T'$. Assuming, by induction, that $T$ is accessible, this shows the accessibility of $T$.\\ Suppose first that $T'$ is not final. We have to consider two cases:
  \begin{enumerate}
      \item  If $T'_{0,0}=0$ then we consider a cross $(i,j)\in T'$ and we set $T=T'\cdot ((0\ i),(0\ j))$. We notice that $T$ is a valid tableau that has the same number of crosses as $T'$, at most as many non-final crosses as $T'$ because all the crosses in $T'$ are non-final, and $T_{0,0}=1$. Hence, $T<T'$ and by induction $T$ is accessible. Since $T\cdot ((0\ i),(0\ j))=T'$ is not final, we have  $\delta^{\left((0\ i),(0\ j)\right)}(T)=T'$. We deduce that $T'$ is accessible.

          
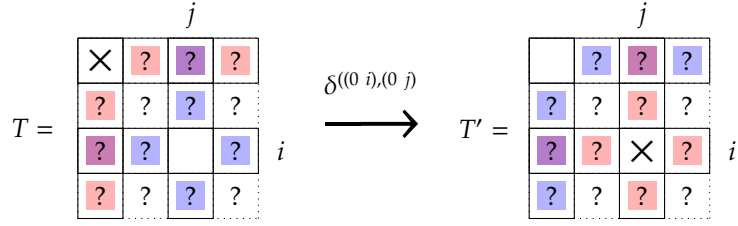
\centerline{
         \begin{tikzpicture}[scale=0.6]
           \foreach \i in {0,...,3} {
           \node[scale=1.5,fill=red,opacity=0.3] at (\i+0.5,3.5) {$ $};
           \node[scale=1.5,fill=blue,opacity=0.3] at (\i+0.5,1.5) {$ $};
           \node[scale=1.5,fill=red,opacity=0.3] at (0.5,\i+0.5) {$ $};
           \node[scale=1.5,fill=blue,opacity=0.3] at (2.5,\i+0.5) {$ $};
            \node[scale=1.5,fill=blue,opacity=0.3] at (\i+10.5,3.5) {$ $};
           \node[scale=1.5,fill=red,opacity=0.3] at (\i+10.5,1.5) {$ $};
           \node[scale=1.5,fill=blue,opacity=0.3] at (10.5,\i+0.5) {$ $};
           \node[scale=1.5,fill=red,opacity=0.3] at (12.5,\i+0.5) {$ $};
           }
           \foreach \i in {0,...,3} {
           \foreach \j in {0,...,3} {
           \node at (\i+0.5,\j+0.5) {$?$};
           \node at (\i+10.5,\j+0.5) {$?$};}}
           \node[scale=1.5,fill=white] at (0.5,3.5) {$\times$};
           \node[scale=1.5,fill=white] at (12.5,1.5) {$\times$};
           \node[scale=1.5,fill=white] at (2.5,1.5) {$ $};
           \node[scale=1.5,fill=white] at (10.5,3.5) {$ $};
           \draw[step=1,black,dotted] (0,4) grid (4,0);
           \draw[step=1,black,thin] (0,4) grid (1,0);
           \draw[step=1,black,thin] (0,4) grid (4,3);
           \draw[step=1,black,thin] (2,4) grid (3,0);
           \draw[step=1,black,thin] (0,2) grid (4,1);
            \draw[step=1,black,dotted] (10,4) grid (14,0);
           \draw[step=1,black,thin] (10,4) grid (11,0);
           \draw[step=1,black,thin] (10,4) grid (14,3);
           \draw[step=1,black,thin] (12,4) grid (13,0);
           \draw[step=1,black,thin] (10,2) grid (14,1);
           \node at (2.5,4.5) {$j$};
           \node at (4.5,1.5) {$i$};
           \node at (12.5,4.5) {$j$};
           \node at (14.5,1.5) {$i$};
           \node at(-1,2) {$T=$};
           \node at(9,2) {$T'=$};
                         \node at (6.5,3) {$\delta^{((0\ i),(0\ j)}$};
              \node[scale=2.5] at (6.5,2) {$\longrightarrow$};
         \end{tikzpicture}
         } 
         
      \item If $T'_{0,0}=1$ then, since $T'$ is not final, the cell $(0,0)$ is not final. Furthermore  $\{(0,0)\}<T'$ implies there exists $(i,j)\neq (0,0)$ such that  $T'_{i,j}=1$. Let $(i',j')$ be a cell in the final zone with $i'\neq 0$ and $j'\neq 0$. We have to consider the following two sub-cases:
      \begin{enumerate}
          \item Suppose that we can choose $(i,j)$ such that  both $i$ and $j$ are not zero. We set $T=T'\cdot ((i\ i'),(j\ j'))$. The tableau $T$ is valid and  has the same number of crosses as $T'$ but has less non final crosses than $T'$ because $T$ contains at least one final cross while all the crosses of $T'$ are non-final. So we have $T<T'$ and  $\delta^{\left((i\ i'),(j\ j')\right)}(T)=T\cdot ((i\ i'),(j\ j'))=T'$ because $T'$ is not final. By induction, we deduce that $T'$ is accessible. 
          \item  If the tableau $T'$ contains crosses only in its first row and in its first column then we suppose that $i=0$ (the other case is symmetric) which means that there exists $j\neq 0$ such that $T'_{0,j}=1$.  For each of the following cases we check that $T<T'$ and $\delta^{(f,g)}(T)=T\cdot (f,g)=T'$:  
          \begin{enumerate}
          \item If $T'_{i',0}=T'_{0,j'}=0$ then we set $T=T'\cup\{(i',j')\}\setminus{(0,j)}$ and $(f,g)=\left(\binomp {i'} 0,\binomp{j'}j\right)$.
          
            
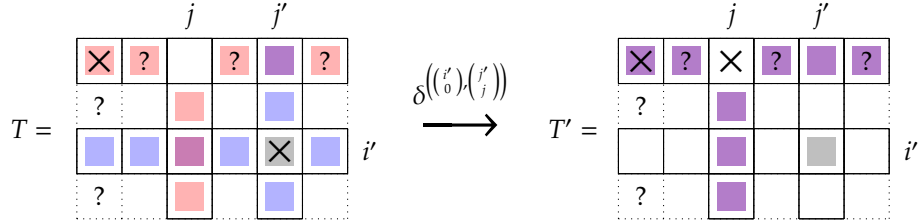
\centerline{
          \begin{tikzpicture}[scale=0.6]
          \foreach \i in {0,...,5} 
          { \node[scale=1.5,fill=red,opacity=0.3] at (\i+0.5,3.5){$ $};
           \node[scale=1.5,fill=blue,opacity=0.3] at (\i+0.5,1.5){$ $};
          \node[scale=1.5,fill=red,opacity=0.3] at (\i+12.5,3.5){$ $};
          \node[scale=1.5,fill=blue,opacity=0.3] at (\i+12.5,3.5){$ $};
          }
          \foreach \j in{0,...,3}
          {
          \node[scale=1.5,fill=red,opacity=0.3] at (2.5,\j+0.5){$ $};
          \node[scale=1.5,fill=blue,opacity=0.3] at (4.5,\j+0.5){$ $};
           \node[scale=1.5,fill=red,opacity=0.3] at (14.5,\j+0.5){$ $};
          \node[scale=1.5,fill=blue,opacity=0.3] at (14.5,\j+0.5){$ $};
          }
          \node[scale=1.5,fill=white] at (2.5,3.5){$ $};
          \node[scale=1.5,fill=lightgray] at (4.5,1.5){$ $};
          \node[scale=1.5,fill=white] at (14.5,3.5){$ $};
          \node[scale=1.5,fill=lightgray] at (16.5,1.5){$ $};
          \draw[step=1,black,dotted] (0,4) grid (6,0);
          \draw[step=1,black,dotted] (12,4) grid (18,0);
          \draw[step=1,black,thin] (0,4) grid (6,3);
          \draw[step=1,black,thin] (0,2) grid (6,1);
          \draw[step=1,black,thin] (2,4) grid (3,0);
          \draw[step=1,black,thin] (4,4) grid (5,0);
          \draw[step=1,black,thin] (0,4) grid (6,3);
          \draw[step=1,black,thin] (0,2) grid (6,1);
          \draw[step=1,black,thin] (2,4) grid (3,0);
          \draw[step=1,black,thin] (4,4) grid (5,0);
          \draw[step=1,black,thin] (12,4) grid (18,3);
          \draw[step=1,black,thin] (12,2) grid (18,1);
          \draw[step=1,black,thin] (14,4) grid (15,0);
          \draw[step=1,black,thin] (16,4) grid (17,0);
          \draw[step=1,black,thin] (12,4) grid (18,3);
          \draw[step=1,black,thin] (12,2) grid (18,1);
          \draw[step=1,black,thin] (14,4) grid (15,0);
          \draw[step=1,black,thin] (16,4) grid (17,0);
          \foreach \i in {1,3,5} { \node at (\i+0.5,3.5) {$?$};\node at (\i+12.5,3.5) {$?$};}
          \node[scale=1.5] at (0.5,3.5) {$\times$};
          \node[scale=1.5] at (4.5,1.5) {$\times$};
          \node[scale=1.5] at (14.5,3.5) {$\times$};
          \node[scale=1.5] at (12.5,3.5) {$\times$};
           \node at (0.5,2.5) {$?$};
           \node at (0.5,0.5) {$?$};
            \node at (12.5,2.5) {$?$};
           \node at (12.5,0.5) {$?$};
           \node at(2.5,4.5) {$j$};\node at(14.5,4.5) {$j$};
           \node at(4.5,4.5) {$j'$};\node at(16.5,4.5) {$j'$};
           \node at(6.5,1.5) {$i'$};\node at(18.5,1.5) {$i'$};
           \node at (-1,2) {$T=$};
           \node at (11,2) {$T'=$};
            \node at (8.5,3) {$\delta^{\left(\binomp{i'}0,\binomp{j'}j\right)}$};
              \node[scale=2] at (8.5,2) {$\longrightarrow$};
          \end{tikzpicture}
          }
          
          \item If $T'_{i',0}=0$ and $T'_{0,j'}=1$ then we set $T=T'\cup\{(i',j')\}\setminus\{(0,j')\}$ and $(f,g)=(\binomp{i'}0,\Id)$.

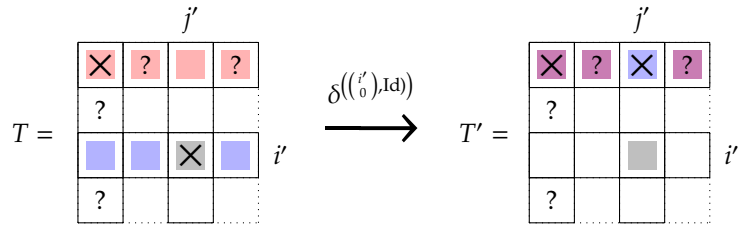
           \centerline{
          \begin{tikzpicture}[scale=0.6]
          \node[scale=1.5,fill=red,opacity=0.3] at (0.5,3.5) {$ $};
          \node[scale=1.5,fill=red,opacity=0.3] at (1.5,3.5) {$ $};
           \node[scale=1.5,fill=red,opacity=0.3] at (2.5,3.5) {$ $};
           \node[scale=1.5,fill=red,opacity=0.3] at (3.5,3.5) {$ $};
             \node[scale=1.5,fill=blue,opacity=0.3] at (0.5,1.5) {$ $};
           \node[scale=1.5,fill=blue,opacity=0.3] at (1.5,1.5) {$ $};
           \node[scale=1.5,fill=blue,opacity=0.3] at (3.5,1.5) {$ $};
           \node[scale=1.5,fill=lightgray] at (2.5,1.5) {$ $};
           \draw[step=1.0,black, dotted] (0,4) grid (4,0);
            \draw[step=1.0,black, thin] (0,4) grid (4,3);
            \draw[step=1.0,black, thin] (0,2) grid (4,1);
           \draw[step=1.0,black, thin] (0,4) grid (1,0);
            \draw[step=1.0,black, thin] (2,4) grid (3,0);
            \node[scale=1.5] at (0.5,3.5) {$\times$};
            \node at (1.5,3.5) {$?$};
            \node [scale=1.5] at (2.5,1.5) {$\times$};
            \node at (3.5,3.5) {$?$};
            \node at (0.5,2.5) {$?$};
            \node at(0.5,0.5) {$?$};
              \node[scale=1.5,fill=blue,opacity=0.3] at (10.5,3.5) {$ $};
          \node[scale=1.5,fill=blue,opacity=0.3] at (11.5,3.5) {$ $};
           \node[scale=1.5,fill=blue,opacity=0.3] at (12.5,3.5) {$ $};
           \node[scale=1.5,fill=blue,opacity=0.3] at (13.5,3.5) {$ $};
             \node[scale=1.5,fill=red,opacity=0.3] at (10.5,3.5) {$ $};
           \node[scale=1.5,fill=red,opacity=0.3] at (11.5,3.5) {$ $};
           \node[scale=1.5,fill=red,opacity=0.3] at (13.5,3.5) {$ $};
           \node[scale=1.5,fill=lightgray] at (12.5,1.5) {$ $};
           \draw[step=1.0,black, dotted] (10,4) grid (14,0);
            \draw[step=1.0,black, thin] (10,4) grid (14,3);
            \draw[step=1.0,black, thin] (10,2) grid (14,1);
           \draw[step=1.0,black, thin] (10,4) grid (11,0);
            \draw[step=1.0,black, thin] (12,4) grid (13,0);
            \node[scale=1.5] at (10.5,3.5) {$\times$};
            \node at (11.5,3.5) {$?$};
            \node[scale=1.5] at (12.5,3.5) {$\times$};
            \node at (13.5,3.5) {$?$};
            \node at (10.5,2.5) {$?$};
              \node at (10.5,0.5) {$?$};
             \node at (6.5,3) {$\delta^{\left(\binomp{i'}0,\Id)\right)}$};
              \node[scale=2.5] at (6.5,2) {$\longrightarrow$};
              \node at (4.5,1.5) {$i'$};
              \node at (14.5,1.5) {$i'$};
              \node at (2.5,4.5) {$j'$};
              \node at (12.5,4.5) {$j'$};
              \node at(-1,2) {$T=$};
              \node at(9,2) {$T'=$};
          \end{tikzpicture}
          }

          \item If $T'_{i',0}=1$ and $T'_{0,j'}=0$ then, symmetrically to the previous case, we set $T=T'\cup\{(i',j')\}\setminus\{(i',0)\}$ and $(f,g)=(\Id,\binomp{j'}0)$.
          
          \item If $T'_{i',0}=T'_{0,j'}=1$ then we set $T=T'\cdot (\Id,(0\ j'))$ and $(f,g)=\left(\mathrm Id,(0\ j')\right)$

           \centerline{
          \begin{tikzpicture}[scale=0.6]
          \node[scale=1.5,fill=red,opacity=0.3] at (0.5,3.5) {$ $};
          \node[scale=1.5,fill=red,opacity=0.3] at (0.5,2.5) {$ $};
           \node[scale=1.5,fill=red,opacity=0.3] at (0.5,1.5) {$ $};
           \node[scale=1.5,fill=red,opacity=0.3] at (0.5,0.5) {$ $};
             \node[scale=1.5,fill=blue,opacity=0.3] at (2.5,0.5) {$ $};
           \node[scale=1.5,fill=blue,opacity=0.3] at (2.5,2.5) {$ $};
           \node[scale=1.5,fill=blue,opacity=0.3] at (2.5,3.5) {$ $};
           \node[scale=1.5,fill=lightgray] at (2.5,1.5) {$ $};
           \draw[step=1.0,black, dotted] (0,4) grid (4,0);
            \draw[step=1.0,black, thin] (0,4) grid (4,3);
            \draw[step=1.0,black, thin] (0,2) grid (4,1);
           \draw[step=1.0,black, thin] (0,4) grid (1,0);
            \draw[step=1.0,black, thin] (2,4) grid (3,0);
            \node[scale=1.5] at (0.5,3.5) {$\times$};
            \node at (1.5,3.5) {$?$};
            \node [scale=1.5] at (2.5,3.5) {$\times$};
            \node [scale=1.5] at (2.5,1.5) {$\times$};
            \node at (3.5,3.5) {$?$};
            \node at (2.5,2.5) {$?$};
            \node at(2.5,0.5) {$?$};
             \node[scale=1.5,fill=blue,opacity=0.3] at (10.5,3.5) {$ $};
          \node[scale=1.5,fill=blue,opacity=0.3] at (10.5,2.5) {$ $};
           \node[scale=1.5,fill=blue,opacity=0.3] at (10.5,1.5) {$ $};
           \node[scale=1.5,fill=blue,opacity=0.3] at (10.5,0.5) {$ $};
             \node[scale=1.5,fill=red,opacity=0.3] at (12.5,0.5) {$ $};
           \node[scale=1.5,fill=red,opacity=0.3] at (12.5,2.5) {$ $};
           \node[scale=1.5,fill=red,opacity=0.3] at (12.5,3.5) {$ $};
           \node[scale=1.5,fill=lightgray] at (12.5,1.5) {$ $};
           \draw[step=1.0,black, dotted] (10,4) grid (14,0);
            \draw[step=1.0,black, thin] (10,4) grid (14,3);
            \draw[step=1.0,black, thin] (10,2) grid (14,1);
           \draw[step=1.0,black, thin] (10,4) grid (11,0);
            \draw[step=1.0,black, thin] (12,4) grid (13,0);
            \node[scale=1.5] at (10.5,3.5) {$\times$};
            \node[scale=1.5] at (10.5,1.5) {$\times$};
            \node at (11.5,3.5) {$?$};
            \node[scale=1.5] at (12.5,3.5) {$\times$};
            \node at (13.5,3.5) {$?$};
            \node at (10.5,2.5) {$?$};
              \node at (10.5,0.5) {$?$};
             \node at (6.5,3) {$\delta^{(\Id,(0\ j')}$};
              \node[scale=2.5] at (6.5,2) {$\longrightarrow$};
              \node at (4.5,1.5) {$i'$};
              \node at (14.5,1.5) {$i'$};
              \node at (2.5,4.5) {$j'$};
              \node at (12.5,4.5) {$j'$};
              \node at(-1,2) {$T=$};
              \node at(9,2) {$T'=$};
          \end{tikzpicture}
          }
          
          \end{enumerate}
      \end{enumerate}
  \end{enumerate}
  Now suppose that $T'$ is final. This implies that $T'_{0,0}=1$ and there exists a final cell $(i,j)$ such that $T'_{i,j}=1$. We have  two cases to consider.
  \begin{enumerate}
      \item If $(i,j)\neq (0,0)$ then we set $T=T'\cdot ((0\ i),(0\ j)) \setminus\{(i,j)\}$. Since $T$ has one cross less than $T'$, we have $T<T'$. Furthermore, $
      \delta^{((0\ i),(0\ j)}(T)=T\cdot ((0\ i),(0\ j))\cup \{(0,0)\}=T'.
      $
      By induction, the state $T'$ is accessible.
      
          
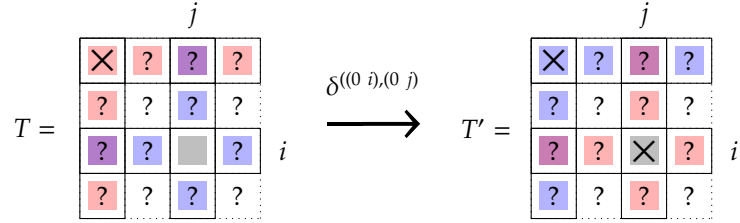
\centerline{
          \begin{tikzpicture}[scale=0.6]
          \node[scale=1.5,fill=red,opacity=0.3] at (0.5,3.5) {$ $};
           \node[scale=1.5,fill=red,opacity=0.3] at (1.5,3.5) {$ $};
           \node[scale=1.5,fill=red,opacity=0.3] at (2.5,3.5) {$ $};
           \node[scale=1.5,fill=red,opacity=0.3] at (3.5,3.5) {$ $};
           \node[scale=1.5,fill=red,opacity=0.3] at (0.5,2.5) {$ $};
           \node[scale=1.5,fill=red,opacity=0.3] at (0.5,1.5) {$ $};
           \node[scale=1.5,fill=red,opacity=0.3] at (0.5,0.5) {$ $};
             \node[scale=1.5,fill=blue,opacity=0.3] at (0.5,1.5) {$ $};
           \node[scale=1.5,fill=blue,opacity=0.3] at (1.5,1.5) {$ $};
%
           \node[scale=1.5,fill=blue,opacity=0.3] at (3.5,1.5) {$ $};
            \node[scale=1.5,fill=blue,opacity=0.3] at (2.5,2.5) {$ $};
            \node[scale=1.5,fill=blue,opacity=0.3] at (2.5,3.5) {$ $};
           \node[scale=1.5,fill=lightgray] at (2.5,1.5) {$ $};
           \node[scale=1.5,fill=blue,opacity=0.3] at (2.5,0.5) {$ $};
           \draw[step=1.0,black, dotted] (0,4) grid (4,0);
            \draw[step=1.0,black, thin] (0,4) grid (4,3);
            \draw[step=1.0,black, thin] (0,2) grid (4,1);
           \draw[step=1.0,black, thin] (0,4) grid (1,0);
            \draw[step=1.0,black, thin] (2,4) grid (3,0);
            \node[scale=1.5] at (0.5,3.5) {$\times$};
            \node at (1.5,3.5) {$?$};
            \node at (2.5,3.5) {$?$};
            \node at (3.5,3.5) {$?$};
            \node at (0.5,2.5) {$?$};
            \node at (1.5,2.5) {$?$};
            \node at (2.5,2.5) {$?$};
            \node at (3.5,2.5) {$?$};
            \node at (0.5,1.5) {$?$};
            \node at (1.5,1.5) {$?$};
            \node at (3.5,1.5) {$?$};
              \node at (0.5,0.5) {$?$};
            \node at (1.5,0.5) {$?$};
            \node at(2.5,0.5) {$?$};
            \node at (3.5,0.5) {$?$};
             \node[scale=1.5,fill=blue,opacity=0.3] at (10.5,3.5) {$ $};
                   \node[scale=1.5,fill=blue,opacity=0.3] at (11.5,3.5) {$ $};
           \node[scale=1.5,fill=blue,opacity=0.3] at (12.5,3.5) {$ $};
           \node[scale=1.5,fill=blue,opacity=0.3] at (13.5,3.5) {$ $};
           \node[scale=1.5,fill=blue,opacity=0.3] at (10.5,2.5) {$ $};
           \node[scale=1.5,fill=blue,opacity=0.3] at (10.5,1.5) {$ $};
           \node[scale=1.5,fill=blue,opacity=0.3] at (10.5,0.5) {$ $};
             \node[scale=1.5,fill=red,opacity=0.3] at (10.5,1.5) {$ $};
           \node[scale=1.5,fill=red,opacity=0.3] at (11.5,1.5) {$ $};
            \node[scale=1.5,fill=red,opacity=0.3] at (13.5,1.5) {$ $};
            \node[scale=1.5,fill=red,opacity=0.3] at (12.5,2.5) {$ $};
            \node[scale=1.5,fill=red,opacity=0.3] at (12.5,3.5) {$ $};
           \node[scale=1.5,fill=lightgray] at (12.5,1.5) {$ $};
           \node[scale=1.5,fill=red,opacity=0.3] at (12.5,0.5) {$ $};
           \draw[step=1.0,black, dotted] (10,4) grid (14,0);
            \draw[step=1.0,black, thin] (10,4) grid (14,3);
            \draw[step=1.0,black, thin] (10,2) grid (14,1);
           \draw[step=1.0,black, thin] (10,4) grid (11,0);
            \draw[step=1.0,black, thin] (12,4) grid (13,0);
            \node[scale=1.5] at (10.5,3.5) {$\times$};
            \node[scale=1.5] at (12.5,1.5) {$\times$};
            \node at (11.5,3.5) {$?$};
            \node at (12.5,3.5) {$?$};
            \node at (13.5,3.5) {$?$};
            \node at (10.5,2.5) {$?$};
            \node at (11.5,2.5) {$?$};
            \node at (12.5,2.5) {$?$};
            \node at (13.5,2.5) {$?$};
            \node at (10.5,1.5) {$?$};
            \node at (11.5,1.5) {$?$};
            \node at (13.5,1.5) {$?$};
             \node at (10.5,0.5) {$?$};
            \node at (11.5,0.5) {$?$};
            \node at(12.5,0.5) {$?$};
            \node at (13.5,0.5) {$?$};
                         \node at (6.5,3) {$\delta^{((0\ i),(0\ j)}$};
              \node[scale=2.5] at (6.5,2) {$\longrightarrow$};
              \node at (4.5,1.5) {$i$};
              \node at (14.5,1.5) {$i$};
              \node at (2.5,4.5) {$j$};
              \node at (12.5,4.5) {$j$};
              \node at(-1,2) {$T=$};
              \node at(9,2) {$T'=$};
          \end{tikzpicture}
          }
          
      \item If $(0,0)$ is the only final cross in $T'$ then we consider $i\neq 0$ and $j\neq 0$ be such that $(i,j)$ is a final cell. 
       Since $\{(0,0)\}<T'$, there exists $(i',j')\in T'\setminus\{(0,0)\}$. We consider the following two cases:
      \begin{enumerate}
      \item Suppose that we can choose $i'\neq 0$ and $j'\neq 0$. We set $T=T'\cdot ((i\ i'),(j\ j'))$. Since  $T_{0,0}=1$, the tableau $T$ is final and valid. Furthermore, $T$ has the same number of crosses as $T'$ but $T$ has more final crosses than $T'$ because $T_{i,j}=1$ and $T'_{i,j}=0$. This implies $T<T'$. We check that
      $\delta^{((i\ i'),(j\ j'))}(T)=T\cdot ((i\ i'),(j\ j'))=T'$. Hence, by induction $T'$ is accessible.
    
            
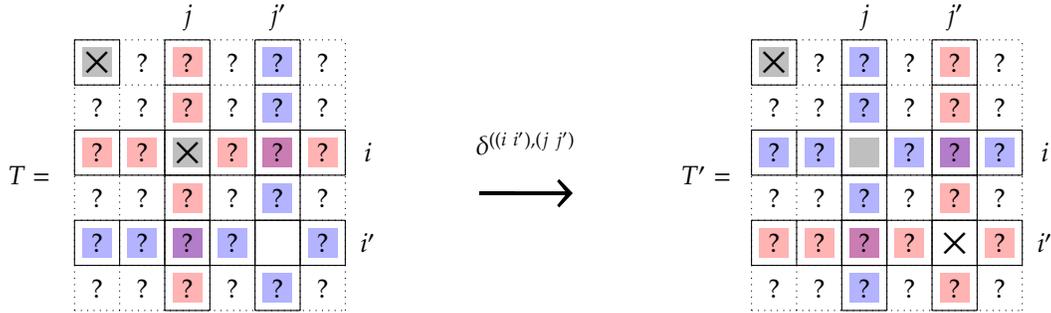
\centerline{
           \begin{tikzpicture}[scale=0.6]
            \node[scale=1.5,fill=lightgray] at (0.5,7.5) {$ $};
            \node[scale=1.5,fill=red,opacity=0.3] at (2.5,7.5) {$ $};
            \node[scale=1.5,fill=blue,opacity=0.3] at (4.5,7.5) {$ $};
            \node[scale=1.5,fill=red,opacity=0.3] at (2.5,6.5) {$ $};
            \node[scale=1.5,fill=blue,opacity=0.3] at (4.5,6.5) {$ $};
            \node[scale=1.5,fill=blue,opacity=0.3] at (4.5,5.5) {$ $};
            \node[scale=1.5,fill=red,opacity=0.3] at (0.5,5.5) {$ $};
            \node[scale=1.5,fill=red,opacity=0.3] at (1.5,5.5) {$ $};
            \node[scale=1.5,fill=red,opacity=0.3] at (2.5,5.5) {$ $};
            \node[scale=1.5,fill=red,opacity=0.3] at (3.5,5.5) {$ $};
            \node[scale=1.5,fill=red,opacity=0.3] at (4.5,5.5) {$ $};
            \node[scale=1.5,fill=red,opacity=0.3] at (5.5,5.5) {$ $};
            \node[scale=1.5,fill=lightgray] at (2.5,5.5) {$ $};
            \node[scale=1.5,fill=blue,opacity=0.3] at (4.5,4.5) {$ $};
            \node[scale=1.5,fill=red,opacity=0.3] at (2.5,4.5) {$ $};
            \node[scale=1.5,fill=red,opacity=0.3] at (2.5,3.5) {$ $};
            \node[scale=1.5,fill=blue,opacity=0.3] at (0.5,3.5) {$ $};
            \node[scale=1.5,fill=blue,opacity=0.3] at (1.5,3.5) {$ $};
            \node[scale=1.5,fill=blue,opacity=0.3] at (2.5,3.5) {$ $};
            \node[scale=1.5,fill=blue,opacity=0.3] at (3.5,3.5) {$ $};
            \node[scale=1.5,fill=blue,opacity=0.3] at (5.5,3.5) {$ $};
            \node[scale=1.5,fill=red,opacity=0.3] at (2.5,2.5) {$ $};
            \node[scale=1.5,fill=blue,opacity=0.3] at (4.5,2.5) {$ $};           
             \draw[step=1.0,black, dotted] (0,8) grid (6,2);
             \draw[step=1.0,black, thin] (0,8) grid (1,7);
              \draw[step=1.0,black, thin] (2,8) grid (3,2);
              \draw[step=1.0,black, thin] (4,8) grid (5,2);
              \draw[step=1.0,black, thin] (0,6) grid (6,5);
              \draw[step=1.0,black, thin] (0,4) grid (6,3);
             \node[scale=1.5] at(0.5,7.5) {$\times$};
             \node[scale=1.5] at(2.5,5.5) {$\times$};
             \node at(1.5,7.5) {$?$};
             \node at(2.5,7.5) {$?$};
             \node at(3.5,7.5) {$?$};
             \node at(4.5,7.5) {$?$};
             \node at(5.5,7.5) {$?$};
             \node at(0.5,6.5) {$?$};
              \node at(1.5,6.5) {$?$};
             \node at(2.5,6.5) {$?$};
             \node at(3.5,6.5) {$?$};
             \node at(4.5,6.5) {$?$};
             \node at(5.5,6.5) {$?$};
              \node at(0.5,5.5) {$?$};
              \node at(1.5,5.5) {$?$};
             \node at(3.5,5.5) {$?$};
             \node at(4.5,5.5) {$?$};
             \node at(5.5,5.5) {$?$};
            \node at(0.5,4.5) {$?$};
              \node at(1.5,4.5) {$?$};
             \node at(2.5,4.5) {$?$};
             \node at(3.5,4.5) {$?$};
             \node at(4.5,4.5) {$?$};
             \node at(5.5,4.5) {$?$};
            \node at(0.5,3.5) {$?$};
              \node at(1.5,3.5) {$?$};
             \node at(2.5,3.5) {$?$};
             \node at(3.5,3.5) {$?$};
             \node at(5.5,3.5) {$?$};
            \node at(0.5,2.5) {$?$};
              \node at(1.5,2.5) {$?$};
             \node at(2.5,2.5) {$?$};
             \node at(3.5,2.5) {$?$};
             \node at(4.5,2.5) {$?$};
             \node at(5.5,2.5) {$?$};
             \node[scale=1.5,fill=lightgray] at (15.5,7.5) {$ $};
                \node[scale=1.5,fill=blue,opacity=0.3] at (17.5,7.5) {$ $};
            \node[scale=1.5,fill=red,opacity=0.3] at (19.5,7.5) {$ $};
            \node[scale=1.5,fill=blue,opacity=0.3] at (17.5,6.5) {$ $};
            \node[scale=1.5,fill=red,opacity=0.3] at (19.5,6.5) {$ $};
            \node[scale=1.5,fill=red,opacity=0.3] at (19.5,5.5) {$ $};
            \node[scale=1.5,fill=blue,opacity=0.3] at (15.5,5.5) {$ $};
            \node[scale=1.5,fill=blue,opacity=0.3] at (16.5,5.5) {$ $};
            \node[scale=1.5,fill=blue,opacity=0.3] at (17.5,5.5) {$ $};
            \node[scale=1.5,fill=blue,opacity=0.3] at (18.5,5.5) {$ $};
            \node[scale=1.5,fill=blue,opacity=0.3] at (19.5,5.5) {$ $};
            \node[scale=1.5,fill=blue,opacity=0.3] at (20.5,5.5) {$ $};
            \node[scale=1.5,fill=lightgray] at (17.5,5.5) {$ $};
            \node[scale=1.5,fill=red,opacity=0.3] at (19.5,4.5) {$ $};
            \node[scale=1.5,fill=blue,opacity=0.3] at (17.5,4.5) {$ $};
            \node[scale=1.5,fill=blue,opacity=0.3] at (17.5,3.5) {$ $};
            \node[scale=1.5,fill=red,opacity=0.3] at (15.5,3.5) {$ $};
            \node[scale=1.5,fill=red,opacity=0.3] at (16.5,3.5) {$ $};
            \node[scale=1.5,fill=red,opacity=0.3] at (17.5,3.5) {$ $};
            \node[scale=1.5,fill=red,opacity=0.3] at (18.5,3.5) {$ $};
            \node[scale=1.5,fill=red,opacity=0.3] at (20.5,3.5) {$ $};
            \node[scale=1.5,fill=blue,opacity=0.3] at (17.5,2.5) {$ $};
            \node[scale=1.5,fill=red,opacity=0.3] at (19.5,2.5) {$ $};            
             \draw[step=1.0,black, dotted] (15,8) grid (21,2);
             \draw[step=1.0,black, thin] (15,8) grid (16,7);
              \draw[step=1.0,black, thin] (17,8) grid (18,2);
              \draw[step=1.0,black, thin] (19,8) grid (20,2);
              \draw[step=1.0,black, thin] (15,6) grid (21,5);
              \draw[step=1.0,black, thin] (15,4) grid (21,3);
             \node[scale=1.5] at(15.5,7.5) {$\times$};
             \node[scale=1.5] at(19.5,3.5) {$\times$};
             \node at(16.5,7.5) {$?$};
             \node at(17.5,7.5) {$?$};
             \node at(18.5,7.5) {$?$};
             \node at(19.5,7.5) {$?$};
             \node at(20.5,7.5) {$?$};
             \node at(15.5,6.5) {$?$};
              \node at(16.5,6.5) {$?$};
             \node at(17.5,6.5) {$?$};
             \node at(18.5,6.5) {$?$};
             \node at(19.5,6.5) {$?$};
             \node at(20.5,6.5) {$?$};
              \node at(15.5,5.5) {$?$};
              \node at(16.5,5.5) {$?$};
             \node at(18.5,5.5) {$?$};
             \node at(19.5,5.5) {$?$};
             \node at(20.5,5.5) {$?$};
            \node at(15.5,4.5) {$?$};
              \node at(16.5,4.5) {$?$};
             \node at(17.5,4.5) {$?$};
             \node at(18.5,4.5) {$?$};
             \node at(19.5,4.5) {$?$};
             \node at(20.5,4.5) {$?$};
            \node at(15.5,3.5) {$?$};
              \node at(16.5,3.5) {$?$};
             \node at(17.5,3.5) {$?$};
             \node at(18.5,3.5) {$?$};
             \node at(20.5,3.5) {$?$};
            \node at(15.5,2.5) {$?$};
              \node at(16.5,2.5) {$?$};
             \node at(17.5,2.5) {$?$};
             \node at(18.5,2.5) {$?$};
             \node at(19.5,2.5) {$?$};
             \node at(20.5,2.5) {$?$};
             \node at(2.5,8.5) {$j$};
             \node at(4.5,8.5) {$j'$};
             \node at(6.5,5.5){$i$};
             \node at(6.5,3.5){$i'$};
             \node at(17.5,8.5) {$j$};
             \node at(19.5,8.5) {$j'$};
             \node at(21.5,5.5){$i$};
             \node at(21.5,3.5){$i'$};
               \node at (10,5.7) {$\delta^{((i\ i'),(j\ j')}$};
              \node[scale=2.5] at (10,4.5) {$\longrightarrow$};
              \node at (-1,5) {$T=$};
                \node at (14,5) {$T'=$};
          \end{tikzpicture}}

       \item Suppose that $T'_{k,\ell}=1$ implies $k=0$ or $\ell=0$. We assume that $i'=0$ and $j'\neq 0$ (the other case is obtained symmetrically).  Let us assume first that there exists $i''\neq 0$ such that $T'_{i'',0}=1$. In that case we define $T^{(2)}= T'\cdot ((i\ i''),(j\ j'))$ and $T= T^{(2)}\cdot ((i\ 0),\Id))$. We observe that $\delta^{(i\ 0),\Id)}(T)=T\cdot ((i\ 0),\Id))=T^{(2)}$.

       \centerline{
       \begin{tikzpicture}[scale=0.6]
           \foreach \i in {0,...,3} {
         \node[scale=1.5,fill=red,opacity=0.3] at (\i+0.5,3.5){$ $};
         \node[scale=1.5,fill=blue,opacity=0.3] at (\i+0.5,1.5){$ $};
                  \node[scale=1.5,fill=blue,opacity=0.3] at (\i+10.5,3.5){$ $};
         \node[scale=1.5,fill=red,opacity=0.3] at (\i+10.5,1.5){$ $};
                  }
        \foreach \i in{1,...,3} {
        \node at(\i+0.5,1.5){$?$};
        \node at(\i+10.5,3.5){$?$};
        }
        \node[scale=1.5,fill=white] at (2.5,3.5) {$ $};
        \node[scale=1.5,fill=white] at (12.5,3.5) {$ $};
        \node[scale=1.5,fill=red,opacity=0.3] at (2.5,3.5) {$ $};
        \node[scale=1.5,fill=blue,opacity=0.3] at (12.5,3.5) {$ $};
        \node[scale=1.5,fill=lightgray] at (0.5,3.5) {$ $};
        \node[scale=1.5,fill=lightgray] at (2.5,1.5) {$ $};
        \node[scale=1.5,fill=lightgray] at (10.5,3.5) {$ $};
        \node[scale=1.5,fill=lightgray] at (12.5,1.5) {$ $};
        \node[scale=1.5] at (0.5,3.5) {$\times$};
        \node[scale=1.5] at (10.5,3.5) {$\times$};
         \draw[step=1.0,black, dotted] (0,4) grid (4,0);
         \draw[step=1.0,black, dotted] (10,4) grid (14,0);
         \node[scale=1.5] at (0.5,1.5) {$\times$};
         \node[scale=1.5] at (2.5,1.5) {$\times$};
         \node[scale=1.5] at (12.5,3.5) {$\times$};
         \node[scale=1.5] at (10.5,1.5) {$\times$}; 
         \draw[step=1.0,black,thin] (0,4) grid (4,3);
         \draw[step=1.0,black,thin] (0,1) grid (4,2);
         \node at(2.5,4.5) {$j$};
         \node at(4.5,1.5) {$i$};
         \node at(-1,2) {$T=$};
         \node at(12.5,4.5) {$j$};
         \node at(14.5,1.5) {$i$};
         \node at(9,2) {$T^{(2)}=$};
        \node at (6.5,3) {$\delta^{((i\ 0),\Id)}$};
        \node[scale=2] at (6.5,2) {$\longrightarrow$};
               \draw[step=1.0,black,thin] (0,4) grid (4,3);
         \draw[step=1.0,black,thin] (0,1) grid (4,2);
       \end{tikzpicture}
       }
       Futhermore  $\delta^{((i\ i''),(j\ j'))}(T^{(2)})= T^{(2)}\cdot ((i\ i''),(j\ j'))=T'$.
       
       \centerline{
       \begin{tikzpicture}[scale=0.6]
         \draw[step=1.0,black, dotted] (0,6) grid (6,0);
         \draw[step=1.0,black, dotted] (15,6) grid (21,0);
         \foreach \i in {0,...,5} {
         \node[scale=1.5,fill=red,opacity=0.3] at (\i+0.5,3.5){$ $};
         \node[scale=1.5,fill=red,opacity=0.3] at (2.5,\i+0.5){$ $};
         \node[scale=1.5,fill=blue,opacity=0.3] at (\i+0.5,1.5){$ $};
         \node[scale=1.5,fill=blue,opacity=0.3] at (4.5,\i+0.5){$ $};
         \node at (0.5,\i+0.5) {$?$};
         \node at (\i+0.5,5.5) {$?$};
         \node[scale=1.5,fill=blue,opacity=0.3] at (\i+15.5,3.5){$ $};
         \node[scale=1.5,fill=blue,opacity=0.3] at (17.5,\i+0.5){$ $};
         \node[scale=1.5,fill=red,opacity=0.3] at (\i+15.5,1.5){$ $};
         \node[scale=1.5,fill=red,opacity=0.3] at (19.5,\i+0.5){$ $};
         \node at (15.5,\i+0.5) {$?$};
         \node at (\i+15.5,5.5) {$?$};
         }
         \node[scale=1.5,fill=lightgray] at (0.5,5.5) {$ $};
         \node[scale=1.5,fill=lightgray] at (15.5,5.5) {$ $};
         \node[scale=1.5,fill=white] at (0.5,3.5) {$ $};
         \node[scale=1.5,fill=white] at (15.5,1.5) {$ $};
         \node[scale=1.5,fill=red,opacity=0.3] at (0.5,3.5) {$ $};
         \node[scale=1.5,fill=red,opacity=0.3] at (15.5,1.5) {$ $};
       \node[scale=1.5] at (0.5,5.5) {$\times$};
       \node[scale=1.5] at (15.5,5.5) {$\times$};
       \node[scale=1.5,fill=white] at (2.5,5.5) {$ $};
       \node[scale=1.5,fill=white] at (19.5,5.5) {$ $};
       \node[scale=1.5,fill=red,opacity=0.3] at (2.5,5.5) {$ $};
       \node[scale=1.5,fill=red,opacity=0.3] at (19.5,5.5) {$ $};
         \node[scale=1.5] at (2.5,5.5) {$\times$};
         \node[scale=1.5] at (19.5,5.5) {$\times$};
         \node[scale=1.5] at (0.5,3.5) {$\times$};
         \node[scale=1.5] at (15.5,1.5) {$\times$};
         \node[scale=1.5,fill=lightgray] at (2.5,3.5) {$ $};
         \node[scale=1.5,fill=lightgray] at (17.5,3.5) {$ $};
         \node[scale=1.5,fill=white] at (4.5,1.5) {$ $};
         \node[scale=1.5,fill=white] at (19.5,1.5) {$ $};
         \node[scale=1.5,fill=blue,opacity=0.3] at (4.5,1.5) {$ $};
         \node[scale=1.5,fill=red,opacity=0.3] at (19.5,1.5) {$ $};
    \draw[step=1.0,black, thin] (0,1) grid (6,2);   
    \draw[step=1.0,black, thin] (0,6) grid (1,5);
    \draw[step=1.0,black, thin] (0,3) grid (6,4);
    \draw[step=1.0,black, thin] (2,0) grid (3,6);
    \draw[step=1.0,black, thin] (4,0) grid (5,6);
    \node at (2.5,6.5) {$j$};
    \node at (4.5,6.5) {$j'$};
    \node at (6.5,1.5) {$i''$};
    \node at (6.5,3.5) {$i$};
    \node at (-1,3) {$T^{(2)}=$};
    \draw[step=1.0,black, thin] (15,1) grid (21,2);   
    \draw[step=1.0,black, thin] (15,6) grid (16,5);
    \draw[step=1.0,black, thin] (15,3) grid (21,4);
    \draw[step=1.0,black, thin] (17,0) grid (18,6);
    \draw[step=1.0,black, thin] (19,0) grid (20,6);
    \node at (17.5,6.5) {$j$};
    \node at (19.5,6.5) {$j'$};
    \node at (21.5,1.5) {$i''$};
    \node at (21.5,3.5) {$i$};
    \node at (14,3) {$T'=$};
          \node at (10,3.7) {$\delta^{((i\ i''),(j\ j')}$};
              \node[scale=2.5] at (10,2.5) {$\longrightarrow$};
         \end{tikzpicture}
       }
       
       We check that $T<T'$. Hence, by induction $T'$ is accessible.

       Now consider the case where, for every $i''\neq 0$, $T'_{i'',0}=0$. The only crosses in $T'$ are in the line $0$, otherwise we apply the previous case. We construct $T^{(2)}= T'\cdot (\Id,(j\ j'))$ and $T=T^{(2)}\cup{(i,j)}\setminus \{(0,j)\}$.
       We have $T^{(2)}=\delta^{(\binomp i{0},\Id)}(T)$.
       
       \centerline{
       \begin{tikzpicture}[scale=0.6]
           \foreach \i in {0,...,3} {
         \node[scale=1.5,fill=red,opacity=0.3] at (\i+0.5,3.5){$ $};
         \node[scale=1.5,fill=blue,opacity=0.3] at (\i+0.5,1.5){$ $};
                  \node[scale=1.5,fill=blue,opacity=0.3] at (\i+10.5,3.5){$ $};
         \node[scale=1.5,fill=red,opacity=0.3] at (\i+10.5,3.5){$ $};
                  }
          \node[scale=1.5,fill=lightgray] at (0.5,3.5) {$ $};
        \node[scale=1.5,fill=lightgray] at (2.5,1.5) {$ $};
        \node[scale=1.5,fill=lightgray] at (10.5,3.5) {$ $};
        \node[scale=1.5,fill=lightgray] at (12.5,1.5) {$ $};
        \node[scale=1.5] at (0.5,3.5) {$\times$};\node[scale=1.5] at (10.5,3.5) {$\times$};
        \node[scale=1.5] at (2.5,1.5) {$\times$};\node[scale=1.5] at (12.5,3.5) {$\times$};
        \node at (1.5,3.5) {$?$};\node at (11.5,3.5) {$?$};
        \node at (3.5,3.5) {$?$};\node at (13.5,3.5) {$?$};
        \node at(2.5,4.5) {$j$};
         \node at(4.5,1.5) {$i$};
         \node at(-1,2) {$T=$};
         \node at(12.5,4.5) {$j$};
         \node at(14.5,1.5) {$i$};
         \node at(9,2) {$T^{(2)}=$};
        \node at (6.5,3) {$\delta^{(\binomp i0,\Id)}$};
        \node[scale=2] at (6.5,2) {$\longrightarrow$};
                \draw[step=1.0,black, dotted] (0,4) grid (4,0);
         \draw[step=1.0,black, dotted] (10,4) grid (14,0);
                \draw[step=1.0,black,thin] (0,4) grid (4,3);
         \draw[step=1.0,black,thin] (0,1) grid (4,2);
                \draw[step=1.0,black,thin] (10,4) grid (14,3);
         \draw[step=1.0,black,thin] (10,1) grid (14,2);
        \end{tikzpicture}
       }
       
        Furthermore $\delta^{=(\Id,(j\ j'))}(T^{(2)})=  T^{(2)}\cdot ((\Id,(j\ j'))=T'$.

        \centerline{
        \begin{tikzpicture}[scale=0.6]
        \foreach \i in{0,...,3} {
        \node[scale=1.5,fill=red,opacity=0.3] at (2.5,\i+0.5) {$ $};
        \node[scale=1.5,fill=blue,opacity=0.3] at (14.5,\i+0.5) {$ $};
        \node[scale=1.5,fill=blue,opacity=0.3] at (4.5,\i+0.5) {$ $};
        \node[scale=1.5,fill=red,opacity=0.3] at (16.5,\i+0.5) {$ $};
        }
        \node[scale=1.5,fill=lightgray] at (0.5,3.5) {$ $};
        \node[scale=1.5,fill=lightgray] at (2.5,1.5) {$ $};
        \node[scale=1.5,fill=lightgray] at (12.5,3.5) {$ $};
        \node[scale=1.5,fill=lightgray] at (14.5,1.5) {$ $};
        \node[scale=1.5] at (0.5,3.5) {$\times$};
        \node[scale=1.5] at (2.5,3.5) {$\times$};
        \node[scale=1.5] at (12.5,3.5) {$\times$};
        \node[scale=1.5] at (16.5,3.5) {$\times$};
        \foreach \i in{1,3,5} {
        \node at (0.5+\i,3.5) {$?$};
        \node at (12.5+\i,3.5) {$?$};
        }
        \node at(2.5,4.5) {$j$};\node at(4.5,4.5) {$j'$};
         \node at(6.5,1.5) {$i$};
        \node at(14.5,4.5) {$j$};\node at(16.5,4.5) {$j'$};
         \node at(18.5,1.5) {$i$};
           \node at(-1,2) {$T^{(2)}=$};
        \node at(11,2) {$T'=$};
          \draw[step=1.0,black, dotted] (0,4) grid (6,0);
         \draw[step=1.0,black, dotted] (12,4) grid (18,0);
                \draw[step=1.0,black,thin] (0,4) grid (6,3);
         \draw[step=1.0,black,thin] (0,1) grid (6,2);
                \draw[step=1.0,black,thin] (12,4) grid (18,3);
         \draw[step=1.0,black,thin] (12,1) grid (18,2);
          \node at (8.5,3) {$\delta^{(\Id,(j\ j'))}$};
        \node[scale=2] at (8.5,2) {$\longrightarrow$};
        \end{tikzpicture}
        }
        
        We check that $T<T'$ and by induction $T'$ is accessible.
      \end{enumerate}
  \end{enumerate}
    
\end{proof}

The results of the section are summarized in the following theorem.

\begin{theorem}\label{remark-acc}
 The automaton $\Acc(\mathrm{M}_{\bullet}^{F_1,F_2})$ is a sub-automaton of $\Val({\mathrm M}_{\bullet}^{F_1,F_2})$. Moreover, the two following assertions are equivalent
 \begin{enumerate}
     \item The final zone is not empty and not included in $(\IntEnt{m}\times \{0\})\cup (\{0\}\times \IntEnt{n})$.
     \item We have 
 $\Acc(\mathrm{M}_{\bullet}^{F_1,F_2})=\Val(\mathrm{M}^{F_1,F_2}_\bullet).$
 \end{enumerate}
\end{theorem}

\subsection{Saturation and Nerode equivalence\label{section-saturation}}
   
Let $\widehat{\mathrm{M}}^{F_1,F_2}_\bullet$ be the automaton with the same alphabet, the same states, the same initial state and the same final states as $\mathrm M^{F_1,F_2}_\bullet$ but with the transition function $d$ defined by $d^{(f,g)}(T)=T\cdot (f,g)$ .

The $\widehat{\mathrm M}_{\bullet}$ automata are simpler than the $\mathrm M_{\bullet}$ ones as the transition function is only defined by the composition of functions. Therefore we will show that two states are equivalent in  $\widehat{\mathrm{M}}^{F_1,F_2}_{\bullet}$ if and only if they are equivalent in  $\mathrm{M}^{F_1,F_2}_{\bullet}$. Thus, we compute the Nerode equivalence on  $\widehat{\mathrm{M}}^{F_1,F_2}_{\bullet}$  which is easier than to compute it on $\mathrm{M}^{F_1,F_2}_{\bullet}$. 

Let us denote by $E_{i,j}=\{(i,j)\}$ the tableau with only one cross at position $(i,j)$
\begin{lemma}\label{lem-tw}
Let $T$ be a non empty tableau and $w \in (\IntEnt{m}^{\IntEnt{m}}\times\IntEnt{n}^{\IntEnt{n}})^*$.  The tableau $\delta^w(T)$ is final in $\mathrm M^{F_1,F_2}_\bullet$ if and only if  at least one of the following assertion is true:
\begin{enumerate}
\item There exists $(i,j) \in T$ such that $d^w(E_{i,j})$ is final.\\
   
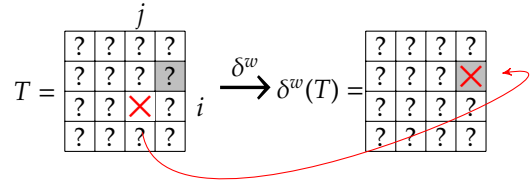
\centerline{
            \begin{tikzpicture}[scale=0.4]
          ;
              \node[scale=1.5,fill=lightgray] at (3.5,2.5) {$ $};
              \draw[step=1.0,black, thin] (0,0) grid (4,4);
 \foreach \j in {0,...,3} { \foreach \i in{0,...,3} {
        \node at (\j+0.5,\i+0.5) {$?$};
        \node at (\j+10.5,\i+0.5) {$?$};
        }
        }
             \node[scale=1.5,fill=white] at (2.5,1.5) {$ $};
              \node[scale=1.5,color=red](A) at (2.5,1.5) {$\times$};
               \node[scale=1.5] at (6,2)  {\Large $\rightarrow$};
               \node[scale=1] at (6,2.8)  {$\delta^w$};
               \node[scale=1] at (2.5,4.5) {$j$};
               \node[scale=1] at (4.5,1.5) {$i$};
              %
%
              \node at(-1,2) {$T=$};
              \node at(8.5,2) {$\delta^w(T)=$};
              \node[scale=1.5,fill=lightgray] at (13.5,2.5) {$ $};
               \node[scale=1.5,color=red](B) at (13.5,2.5) {$\times$};
                             \draw[step=1.0,black, thin] (10,0) grid (14,4);
               \draw[->,color=red] (A) to [bend right,out=270,in=0] (B);
          \end{tikzpicture}}
            
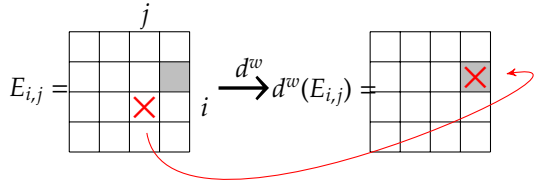
\centerline{
            \begin{tikzpicture}[scale=0.4]
          ;
 \
             \node[scale=1.5,fill=white] at (2.5,1.5) {$ $};
              \node[scale=1.5,color=red](A) at (2.5,1.5) {$\times$};
               \node[scale=1.5] at (5.8,2)  {\Large $\rightarrow$};
               \node[scale=1] at (6,2.8)  {$d^w$};
               \node[scale=1] at (2.5,4.5) {$j$};
               \node[scale=1] at (4.5,1.5) {$i$};
              %
              \node at(-1,2) {$E_{i,j}=$};
              \node at(8.5,2) {$d^w(E_{i,j})=$};
               \node[scale=1.5,fill=lightgray] at (3.5,2.5) {$ $};
               \draw[step=1.0,black, thin] (0,0) grid (4,4);
               \node[scale=1.5,fill=lightgray] at (13.5,2.5) {$ $};
               \node[scale=1.5,color=red](B) at (13.5,2.5) {$\times$};
               \draw[step=1.0,black, thin] (10,0) grid (14,4);
               \draw[->,color=red] (A) to [bend right,out=270,in=0] (B);
          \end{tikzpicture}}
\item There exists $w=ps\in (\IntEnt{m}^{\IntEnt{m}}\times\IntEnt{n}^{\IntEnt{n}})^*$, with $s\neq \varepsilon$ such that $\delta^p(T)$ is final and $d^s(E_{0,0})$ is final.
 
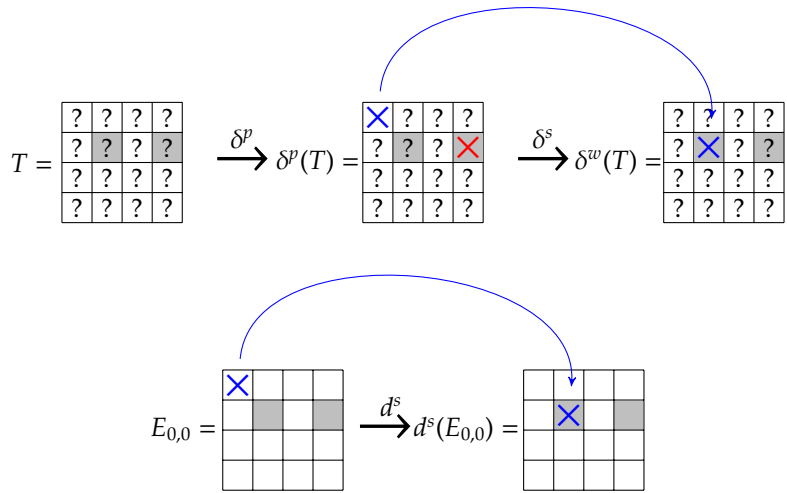
\centerline{
            \begin{tikzpicture}[scale=0.4]
          ;
              \node[scale=1.5,fill=lightgray] at (3.5,2.5) {$ $};
              \node[scale=1.5,fill=lightgray] at (1.5,2.5) {$ $};
              \node[scale=1.5,fill=lightgray] at (11.5,2.5) {$ $};
               \node[scale=1.5,fill=lightgray] at (23.5,2.5) {$ $};
              \draw[step=1.0,black, thin] (0,0) grid (4,4);
 \foreach \j in {0,...,3} { \foreach \i in{0,...,3} {
        \node at (\j+0.5,\i+0.5) {$?$};
        \node at (\j+10.5,\i+0.5) {$?$};
        \node at (\j+20.5,\i+0.5) {$?$};
        }
        }
             \node[scale=1.5,fill=white] at (10.5,3.5) {$ $};
              \node[scale=1.5,color=blue](A) at (10.5,3.5) {$\times$};
               \node[scale=1.5] at (6,2)  {\Large $\rightarrow$};
               \node[scale=1] at (6,2.8)  {$\delta^p$};
               \node[scale=1.5] at (16,2)  {\Large $\rightarrow$};
               \node[scale=1] at (16,2.8)  {$\delta^s$};
              \draw[step=1.0,black, thin] (10,0) grid (14,4);
\draw[step=1.0,black, thin] (20,0) grid (24,4);
              \node at(-1,2) {$T=$};
              \node at(8.5,2) {$\delta^p(T)=$};
              \node at(18.5,2) {$\delta^w(T)=$};
              \node[scale=1.5,fill=lightgray] at (13.5,2.5) {$ $};
               \node[scale=1.5,color=red] at (13.5,2.5) {$\times$};
               \node[scale=1.5,fill=lightgray] at (21.5,2.5) {$ $};
               \node[scale=1.5,color=blue](B) at (21.5,2.5) {$\times$};  
                \draw[->,color=blue] (A) to [bend right,out=90,in=90] (B);
          \end{tikzpicture}}
          \centerline{
            \begin{tikzpicture}[scale=0.4]
          ;
 \
             \node[scale=1.5,fill=white] at (2.5,1.5) {$ $};
              \node[scale=1.5,color=blue](A) at (0.5,3.5) {$\times$};
             \node[scale=1.5,fill=lightgray] at (3.5,2.5) {$ $};
            \node[scale=1.5,fill=lightgray] at (1.5,2.5) {$ $};
            \draw[step=1.0,black, thin] (0,0) grid (4,4);             %
               \node[scale=1.5] at (5.4,2)  {\Large $\rightarrow$};
               \node[scale=1] at (5.6,2.8)  {$d^s$};
              %
%
              \node at(-1.3,2) {$E_{0,0}=$};
              \node at(8.1,2) {$d^s(E_{0,0})=$};
                  \node[scale=1.5,fill=lightgray] at (13.5,2.5) {$ $};
              \node[scale=1.5,fill=lightgray] at (11.5,2.5) {$ $};
               \node[scale=1.5,color=blue](B) at (11.5,2.5) {$\times$};
              \draw[step=1.0,black, thin] (10,0) grid (14,4);
                \draw[->,color=blue] (A) to [bend right,out=90,in=90] (B);
          \end{tikzpicture}}
\end{enumerate}
\end{lemma}
\begin{proof}
Notice that for any $w=ps$, if $(i,j)\in \delta^p(T)$ and $d^s(E_{i,j})$ is final then the tableau $\delta^w(T)$ is final 
The two assertions are  special cases of this property. This shows the if part of the lemma.\\
Now we prove the converse. Let $w\in \left(\IntEnt{m}^{\IntEnt{m}}\times \IntEnt{n}^{\IntEnt{n}}\right)^*$ such that $\delta^w(T)$ is final. We prove the result by induction on $w$. If $w=\varepsilon$ then the result is obvious. If $|w|=1$ then by construction $d^w(T)$ is final and this means that there exists $(i,j)\in T$ such that $d^w(E_{i,j})$ is final. Now we assume that the length of $w$ is at least $2$. We have 
$w=a\cdot w'$ for some $a\in \IntEnt{m}^{\IntEnt{m}}\times \IntEnt{n}^{\IntEnt{n}}$ and we set $T'=\delta^a(T)$. 
By construction,
we have 
$T'=d^a(T)$ or $T'=d^a(T) \cup\{(0,0)\}$ 
depending on the finality of $d^a(T)$.
By induction one of the following property is true:
\begin{enumerate}
\item There exists $(i',j') \in T'$ such that $d^{w'}(E_{i',j'})$ is final. If $(i',j')=(0,0)$ and $T'$ is final then the second assertion is immediately true for $T$. If $(i',j')\neq(0,0)$ or if $T'$ is non final then there exists $(i,j)\in T$ such that 
$d^a(E_{i,j})=E_{i',j'}$ and $d^w(E_{i,j})=d^{w'}(E_{i',j'})$ 
is final; so, the first assertion is true for $T$.
\item  There exists $w'=p's$ with $s\neq\varepsilon$ such that $\delta^s(T')$ is final and $d^s(E_{0,0})$ is  final. Hence, $\delta^{a\cdot p'}(T)$ is final, and the second assertion is true for $T$.
\end{enumerate}
\end{proof}
\begin{lemma}\label{lem-dw2dz}
        For any $w\in (\IntEnt{m}^{m}\times \IntEnt{n}^{n})^*$, there exists $z\in \IntEnt{m}^{m}\times \IntEnt{n}^{n}$ such that for any non empty tableau  $T$ we have $d^w(T)=d^z(T)$. 
\end{lemma}
\begin{proof}
 The result comes   directly from the fact that any composition of functions of $\IntEnt{m}^{m}\times \IntEnt{n}^{n}$ is a function of $\IntEnt{m}^{m}\times \IntEnt{n}^{n}$ \textit{i.e.}
 $d^{(f_1,g_1)\cdots (f_k,g_k)}=d^{(f_k\circ \cdots \circ f_1,g_k\circ \cdots \circ g_1)}.$
\end{proof}

The following result characterizes Nerode classes of $\mathrm M^{F_1,F_2}_\bullet$.
\begin{proposition}
  Any two states represented by non-empty tableaux are Nerode equivalent in $\mathrm M^{F_1,F_2}_\bullet$ if and only if they are Nerode equivalent in $\widehat{\mathrm{M}}^{F_1,F_2}_\bullet$. Furthermore,   in $\mathrm M^{F_1,F_2}_\bullet$, the empty tableau is Nerode equivalent  to the state $\{(0,0)\}$ if and only if $(0,0)$ is in the final zone.
\end{proposition}
\begin{proof}
Let $T$ and $T'$ be two non-empty tableaux.\\
Suppose first that $T$ and $T'$ are not Nerode equivalent in $\mathrm M_\bullet^{F_1,F_2}$. Without loss of generalities, we assume there exists $w\in (\IntEnt{m}^{\IntEnt{m}}\times \IntEnt{n}^{\IntEnt{n}})^*$ such that $\delta^w(T)$ is final while $\delta^w(T')$ is not final. We assume that $w$ is minimal in the sense that for any prefixe $p\neq w$ of $w$, $\delta^p(T)$ and $\delta^p(T')$ have the same finality. From Lemma \ref{lem-tw}, one has to consider two cases:
\begin{enumerate}
    \item There exists $(i,j) \in T$ such that $d^w(E_{i,j})=E_{i',j'}$ is  final. We have $(i,j)\not\in T'$ because, in that case, the tableau $T'$ would be final. Hence, $d^{w}(T)$ is final because $(i',j')\in d^{w}(T)$ while  $d^{w}(T')$ is not final because $d^{w}(T')\subset \delta^w(T')$. So the tableaux $T$ and $T'$ are not Nerode equivalent in  $\widehat M_\bullet^{F_1,F_2}$.
    \item There exists $w=ps$ with $s\neq\varepsilon$ such that $\delta^p(T)$ is final and $d^s(E_{0,0})$ is final. In that case, since $\delta^w(T')$ is not final, we have $(0,0)\not\in\delta^p(T')$ which is also non final. This contradicts the  minimality of $w$.
\end{enumerate}
Conversely, from Lemma \ref{lem-dw2dz} it suffices to prove that if $a\in \IntEnt{m}^{m}\times\IntEnt{n}^{n}$ is such that $d^a(T)$ is final while $d^a(T')$ is not final then $\delta^a(T)$ is final while $\delta^a(T')$ is not final. This property comes directly from the definition of $\delta$ since the addition of $(0,0)$ does not change the finality of a final state. 

The case of the empty tableau is straightforward from the construction.
\end{proof}

This proposition means that to describe  Nerode equivalence on non-empty tableaux in $\mathrm M^{F_1,F_2}_\bullet$ it suffices to describe it in the simpler automaton $\widehat{\mathrm{M}}^{F_1,F_2}_\bullet$.

From now we only consider non-empty tableaux. Recall that we denote Nerode equivalence by $\sim$.

\begin{proposition}\label{prop-lattice}
        For any two states $T_1$ and $T_2$ of  $\widehat{\mathrm{M}}^{F_1,F_2}_\bullet$, we have $T_1\sim T_2$ implies $T_1\sim T_1\cup T_2$.\\ Equivalently each  Nerode class  in  $\widehat{\mathrm{M}}^{F_1,F_2}_\bullet$ is  a join lattice for the inclusion order.
\end{proposition}
\begin{proof}
Since for any $w\in \left(\IntEnt m^{\IntEnt m}\times \IntEnt n^{\IntEnt n}\right)^*$ there exists $a\in \IntEnt m^{\IntEnt m}\times \IntEnt n^{\IntEnt n}$ such that $d^w=d^a$,
we have just prove  that $d^a(T_1)$ is final if and only if $d^a(T_1\cup T_2)$ is final for any  $a\in \IntEnt m^{\IntEnt m}\times \IntEnt n^{\IntEnt n}$.

If $d^a(T_1)$ is final then there exists $(i,j)\in T_1$ such that $d^a(E_{i,j})$ is  final and so, since $(i,j)\in T_1\cup T_2$ the tableau $d^a(T_1\cup T_2)$ is also final.

 Conversely, if $d^a(T_1\cup T_2)$ is final then there exists $(i,j)\in (T_1\cup T_2)$ such that $d^a(E_{i,j})$ is  final. If $(i,j)\in T_1$ then $d^a(T_1)$ is obviously final. If $(i,j)\in T_2$ then $d^a(T_2)$ is final and, since $T_1$ and $T_2$ are Nerode equivalent, the tableau $d^a(T_1)$  is  also final.

 \end{proof}
From this proposition, in each Nerode class there exists a unique state represented by a tableau with a maximal number of crosses. This allows us to give the following definition. 
\begin{definition}\label{def-sat}  A tableau is   \emph{saturated} if it is the unique tableau having the maximal number of crosses in its Nerode class, \ie $T$ is saturated if and only if $T\sim T'$ implies $T'\subset T$.
\end{definition}
The set of  saturated tableaux is  representative of Nerode classes. We denote by $\Sat(T)$ the unique saturated tableau in the Nerode class of $T$ and by $\Sat(A)$ the automaton isomorphic to $A/_{\sim}$, the states of which are labelled by saturated tableaux, for any sub-automaton $A$ of $\mathrm{M}_\bullet^{F_1,F_2}$.

\begin{corollary}
  For any tableau $T$, the two following assertions are equivalent
  \begin{enumerate}
      \item There exists a valid tableau in the class of $T$.
      \item The tableau $\Sat(T)$ is valid. 
  \end{enumerate}
\end{corollary}
\begin{proof}
 We have only to prove that $T$ is valid implies $\Sat(T)$ is valid. The tableaux $T$ and $\Sat(T)$ have the same finality. If $T$ and $\Sat(T)$ are both final then $T_{0,0}=1$ because $T$ is valid. From Proposition \ref{prop-lattice}, this implies that  $\Sat(T)_{0,0}=1$ and so $\Sat(T)$ is valid.
 If $T$ and $\Sat(T)$ are both non final then $\Sat(T)$ have no cross in the final zone, so it is valid.
\end{proof}
\begin{corollary}\label{cor-perm}
  If $T$ and $T'$ are Nerode equivalent in $\mathrm M^{F_1,F_2}_\bullet$  then for any pair $(f,g)\in\IntEnt{m}^{\IntEnt{m}}\times\IntEnt{n}^{\IntEnt{n}}$ the tableaux $d^{(f,g)}(T)$ and $d^{(f,g)}(T')$ are Nerode equivalent.\\
  Furthermore, when $f$ and $g$ are two permutations, the converse is true and $d^{(f,g)}(\Sat(T))=\Sat(d^{(f,g)}(T))$.
\end{corollary}
We denote by \begin{equation}\label{eq-SV}\SV^{F_1,F_2}_\bullet=\{\Sat(T)\mid T\in\Val_\bullet^{F_1,F_2}\}\end{equation}  the set of  valid saturated tableaux. 
\begin{theorem}\label{th-SV}
  For any $F_1, F_2$ and any operation $\bullet$, we have $\#\mathrm{\Min}(\mathrm M_\bullet^{F_1,F_2})\leq \#\SV^{F_1,F_2}_\bullet$.\\
  Futhermore, the equality holds when there exist $i\neq 0$ and $j\neq 0$ such that $(i,j)$ is in the final zone.
\end{theorem}
\begin{proof}
 The automaton $\Sat(\Val(\mathrm{M}_\bullet^{F_1,F_2}))$ is the Nerode quotient of a sub-automaton of $\mathrm{M}_\bullet^{F_1,F_2}$ containing at least all its accessible states. This implies  $\#\mathrm{\Min}(\mathrm M_\bullet^{F_1,F_2})\leq \#\SV^{F_1,F_2}_\bullet.$\\
 Furthermore Proposition \ref{prop-acc} means  $\Val(\mathrm{M}_\bullet^{F_1,F_2})=\Acc(\mathrm{M}_\bullet^{F_1,F_2})$ when there exist $i\neq 0$ and $j\neq 0$ such that $(i,j)$ is in the final zone. So by construction, the equality holds.
\end{proof}
As a consequence, the state complexity of $\ostar$ satisfies the inequality  below:
\begin{equation*}\sc_{\ostar}(m,n)\leq\max\left\{\#\SV^{F_1,F_2}_\bullet\mid F_1\subset\IntEnt{m}, F_2\subset\IntEnt{n} \right\}.
\end{equation*}
The main implication of  Theorem \ref{th-SV} is that in order to compute the state complexity, we need to study the combinatorics of saturated tableaux. This is what we do in the following sections. First, we study the (simplest) $2\times 2$ case and then we generalize it for bigger tableaux, via the notion of local saturation.
\subsection{Computing the Nerode equivalence: the $2\times 2$ case \label{2x2case}}
Let us set  $m=n=2$ and investigate in more details  Nerode equivalence of the automaton $\widehat{M_\bullet}^{F_1,F_2}$ in that case. We assume that $\#F_1=\#F_2=1$ and that the operation $\bullet$ is not degenerated, otherwise at most one of the automata of $\mon^{F_1,F_2}_{m,n}$ is not minimal.   Nerode classes depend only on the set $(F_1\times\IntEnt 2)\bullet (\IntEnt 2\times F_2)$ which can be either a singleton $\{(\alpha,\beta)\}$, or the complementary of a singleton $\left(\IntEnt 2\times \IntEnt 2\right)\setminus \{(\alpha,\beta)\}$, or a diagonal set  $\{(\alpha,\beta),(1-\alpha,1-\beta)\}$. 
In each case, acting on languages or complements does not change the global form of $(F_1\times\IntEnt 2)\bullet (\IntEnt 2\times F_2)$.
According, from Table \ref{op-bool}, we have to consider only three cases
\begin{enumerate}
    \item[A]\label{case-A} If $\bullet=\cap$ then the final zone is $\{(\alpha,\beta)\}$ for some $\alpha,\beta\in\{0,1\}$. We check that all  Nerode classes are singletons.
    \item[X]\label{case-X} If $\bullet=\oplus$  then the final zone is 
    $\{(\alpha,\beta),(1-\alpha,1-\beta)\}$ for some $\alpha,\beta\in\{0,1\}$. We check that the only Nerode class containing at least two tableaux is the class \[\left\{\begin{array}{|c|c|}\hline \times&\times\\\hline \times&\times\\\hline  \end{array},
    \begin{array}{|c|c|}\hline \times&\times\\\hline \times&\\\hline  \end{array},\begin{array}{|c|c|}\hline \times&\times\\\hline &\times\\\hline  \end{array},
    \begin{array}{|c|c|}\hline &\times\\\hline \times&\times\\\hline  \end{array},\begin{array}{|c|c|}\hline \times&\\\hline \times&\times\\\hline  \end{array}
    \right\}.\]
     \item[O]\label{case-O} If $\bullet=\cup$  then the final zone is 
     $\left(\IntEnt 2\times \IntEnt 2\right)\setminus \{(\alpha,\beta)\}$ for some $\alpha,\beta\in\{0,1\}$. We check that the only Nerode class containing at least two tableaux is the class \[\left\{\begin{array}{|c|c|}\hline \times&\times\\\hline \times&\times\\\hline  \end{array},
    \begin{array}{|c|c|}\hline \times&\times\\\hline \times&\\\hline  \end{array},\begin{array}{|c|c|}\hline \times&\times\\\hline &\times\\\hline  \end{array},
    \begin{array}{|c|c|}\hline &\times\\\hline \times&\times\\\hline  \end{array},\begin{array}{|c|c|}\hline \times&\\\hline \times&\times\\\hline  \end{array}, \begin{array}{|c|c|}\hline \times&\\\hline &\times\\\hline\end{array},
    \begin{array}{|c|c|}\hline &\times\\\hline \times&\\\hline  \end{array}
    \right\}.\]
\end{enumerate}
As a consequence of the classification above, for any $2\times 2$-tableau $T$, we have $\Sat(T)\in
\left\{\ \begin{array}{|c|c|}\hline \times&\times\\\hline \times&\times\\\hline  \end{array} ,T\right\}$. 
In Case A, we have $\Sat(T)=T$ for any $T$.\\
In Case X, we have $\Sat(T)=\begin{array}{|c|c|}\hline \times&\times\\\hline \times&\times\\\hline  \end{array}$ if and only if $\#T>2$.\\
In Case O,  we have $\Sat(T)=\begin{array}{|c|c|}\hline \times&\times\\\hline \times&\times\\\hline  \end{array}$ if and only if $T$ contains a diagonal, \ie there exist $\alpha,\beta\in\{0,1\}$ such that $\{(\alpha,\beta),(1-\alpha,1-\beta)\}\subseteq T$.
\subsection{Computing the Nerode equivalence: from the $2\times 2$ case to the general case\label{sec-localsat}}
We assume that $m,n\geq 2$, $\#F_1\not\in\{0,m\}$, $\#F_2\not\in\{0,n\}$ and that the law $\bullet$ is not degenerated. According to Table \ref{op-bool}, without loss of generality, we assume that $\bullet\in\{\cup,\cap,\oplus\}$. 
\begin{remark}\label{rem-nondeg}
The above condition implies  that there exists at least one row (resp. column) containing a final cell and a non final cell.
\end{remark}
\begin{definition}\label{def-restriction}
          Let  $T\in 2^{\IntEnt{m}\times\IntEnt{n}}$. Let $i_0, i_1\in\IntEnt{m}$,  
          $j_0, j_1\in\IntEnt{n}$ and $I=\{i_0,i_1\}\times\{j_0,j_1\}$. We denote by $T|_I=T\cap I$ 
the \emph{restriction} of $T$ to the cells $\{i_0,i_1\}\times \{j_0,j_1\}$. A tableau $T$  is called \emph{simple} if $T=T|_I$ for some $I=\{i_0,i_1\}\times\{j_0,j_1\}$.\\
Let $T$ be a simple tableau and $I=\{i_0,i_1\}\times\{j_0,j_1\}$ with $i_0\leq i_1$ and $j_0\leq j_1$ such that $T=T|_I$. We define the $I$-\emph{reduced} (or simply \emph{reduced} when the context does not induce ambiguity) of $T$ as the $2\times 2$ tableau $\Red_I(T)=\{(\alpha,\beta)\mid \alpha,\beta\in\IntEnt {2}\text{ and }(i_\alpha,j_\beta)\in T\}$. When there is no ambiguity we omit the subscript $I$.\\
Symmetrically, for any $2\times2$-tableau $X$, $m,n\geq 2$, and $I=\{i_0,i_1\}\times\{j_0,j_1\}$ with $0\leq i_0<i_1\leq m-1$ and $0\leq j_0<j_1\leq n-1$, we define its \emph{$I-$induced} as the tableau $\Ind_{I}(X)$ which is the unique $m\times n$-simple tableau the $I$-reduced of which is $X$, \ie $\Red_I(\Ind_I(X))=X$. 
\end{definition}
  Let symbol $F$ refer to the final zone of any tableau $T$. More precisely, we have
$F=(F_1\times\IntEnt n)\bullet (\IntEnt m\times F_2)$. We consider that the final zone of any reduced tableau is $F_R=(\{1\}\times\IntEnt 2)\bullet (\IntEnt 2\times\{1\})$. By this way, the notion of saturated reduced tableau makes sense. 
If $I=\{i_0,i_1\}\times \{j_0,j_1\}\subset \IntEnt m\times\IntEnt n$ then the set $F_I=I\cap F$ refers to the \emph{restricted final zone} associated to $I$.
\begin{remark}\label{CasesFI}
The restricted final zone satisfies one of the following identities
\begin{enumerate}
    \item[(a)] $F_I=\emptyset$,
    \item[(b)] $j_0\neq j_1$ and $F_I=\{(i_a,j_0),(i_a,j_1)\}$ for some $a\in\IntEnt2$,
    \item[(c)] $i_0\neq i_1$ and $F_I=\{(i_0,j_b),(i_1,j_b)\}$ for some $b\in\IntEnt2$,
    \item[(d)] $i_0\neq i_1$, $j_0\neq j_1$, and $F_I=(\{i_a\}\times \{j_0,j_1\})\bullet (\{i_0,i_1\}\times \{j_b\}$ for some $a,b\in\IntEnt2$,
    \item[(e)] $F_I=I$.
\end{enumerate}
\end{remark}

The following lemma shows that every simple tableau  can be projected to its reduced tableau  preserving the property of finality of the cells.
\begin{lemma}\label{T2Red}
Let $I=\{i_0,i_1\}\times \{j_0,j_1\}\subset \IntEnt m\times\IntEnt n$. There exist two maps $f_I,g_I\in\IntEnt2^{\IntEnt2}$ such that for any $\alpha,\beta\in\IntEnt2$, $(i_\alpha,j_\beta)\in F_I$ if and only if $(f_I(\alpha),g_I(\beta))\in F_R$.
\end{lemma}
\begin{proof}
The construction is summarized in Table \ref{figi} according to the type A, X or O of $\bullet$ (see Section \ref{2x2case}) and the cases listed in Remark \ref{CasesFI}.
\begin{table}[H]\[
\begin{array}{|cc||c|c|c|c|c|}
\hline
\ _{Types (\bullet)}\backslash^{Cases (I)}&&(a)&(b)&(c)&(d)&(e)\\\hline
A&\begin{array}{l} f_I\\g_I\end{array}&\begin{array}{c}0\\0\end{array}&\begin{array}{l}\alpha\rightarrow [\alpha=a]\\
1\end{array}& \begin{array}{l}1\\\beta\rightarrow [\beta=b] \end{array}&
\begin{array}{l}\alpha \rightarrow [\alpha=a]\\
\beta\rightarrow [\beta=b]\end{array}&\begin{array}{c}1\\1\end{array}\\\hline
X&\begin{array}{l} f_I\\g_I\end{array}&\begin{array}{l}0\\0\end{array}&
\begin{array}{l}\alpha\rightarrow [\alpha=a]\\
0\end{array}
& \begin{array}{l}0\\ \beta\rightarrow [\beta=b] \end{array}&
\begin{array}{l}\alpha\rightarrow [\alpha= a]\\
\beta\rightarrow [\beta=b]\end{array}&\begin{array}{c}1\\0\end{array}\\\hline

O&\begin{array}{l} f_I\\g_I\end{array}&\begin{array}{l}0\\0 \end{array}&\begin{array}{l}\alpha\rightarrow [\alpha=a]\\
0\end{array}& \begin{array}{l}0\\ \beta\rightarrow [\beta=b] \end{array}&
\begin{array}{l}\alpha\rightarrow [\alpha=a]\\
\beta\rightarrow[\beta=b]\end{array}&\begin{array}{l}1\\1 \end{array}\\\hline
\end{array}
\]
\caption{The maps $f_I$ and $g_I$. \label{figi}}
\end{table}
\end{proof}

\begin{example} The following picture illustrates the various items of Remark \ref{CasesFI} for $\bullet=\oplus$. The red cells correspond to the final zone $F$ of the tableau.
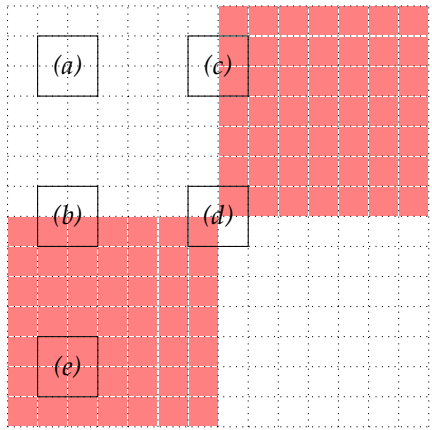
\begin{center}
    \begin{tikzpicture}[scale=0.4]
      \foreach \j in {0,...,6} { \foreach \i in{0,...,6} {
        \node[scale=1.5,fill=red,,opacity=0.5] at (\j+0.5,\i+0.5) {$ $};
        \node[scale=1.5,fill=red,,opacity=0.5] at (\j+7.5,\i+7.5) {$ $};
        }
        }
    \draw[step=1.0,black, dotted] (0,0) grid (14,14);
    \draw[black] (1,1) rectangle (3,3); \node at (2,2) {(e)};
    \draw[black] (6,6) rectangle (8,8); \node at (7,7) {(d)};
    \draw[black] (1,6) rectangle (3,8); \node at (2,7) {(b)};
    \draw[black] (1,11) rectangle (3,13); \node at (2,12) {(a)};
    \draw[black] (6,11) rectangle (8,13); \node at (7,12) {(c)};
\end{tikzpicture}    
\end{center}
The pair of maps $(f_I,g_I)$ of Table \ref{figi} is illustrated below with respect to each final zone $F_I$, $(a)$ to $(e)$ drawn in the previous figure.
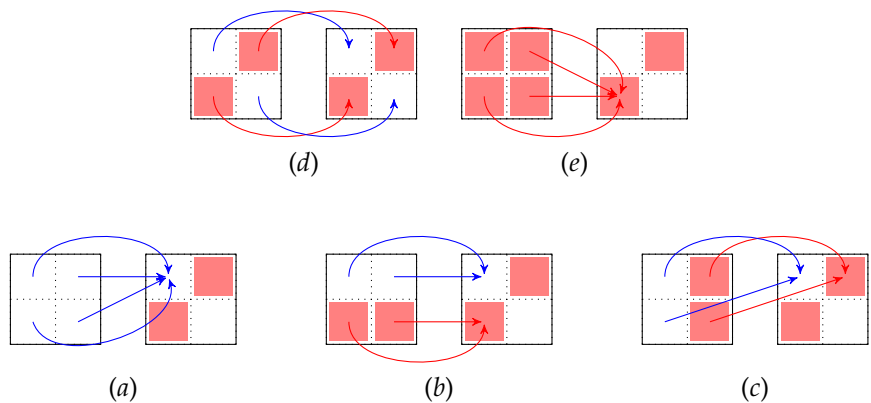
\begin{center}
    \begin{tikzpicture}[scale=0.6]
    \node[scale=2,fill=red,,opacity=0.5] at (3.5,0.5) {$ $};
     \node[scale=2,fill=red,,opacity=0.5] at (4.5,1.5) {$ $};
    \draw[step=1.0,black, dotted] (0,0) grid (2,2);
    \draw[step=1.0,black, dotted] (3,0) grid (5,2);
    \draw[black] (0,0) rectangle (2,2);
    \draw[black] (3,0) rectangle (5,2);
     \draw[->,color=blue] (0.5,0.5) to [bend left,out=-90,in=-90] (3.5,1.5);
     \draw[->,color=blue] (1.5,0.5) to  (3.5,1.5);
     \draw[->,color=blue] (1.5,1.5) to  (3.5,1.5);
     \draw[->,color=blue] (0.5,1.5) to [bend right,out=90,in=90] (3.5,1.5);
     
     \node at (2.5,-1) {$(a)$};
     \node[scale=2,fill=red,,opacity=0.5] at (10.5,0.5) {$ $};
     \node[scale=2,fill=red,,opacity=0.5] at (11.5,1.5) {$ $};
     \node[scale=2,fill=red,,opacity=0.5] at (7.5,0.5) {$ $};
     \node[scale=2,fill=red,,opacity=0.5] at (8.5,0.5) {$ $};
     
    \draw[step=1.0,black, dotted] (7,0) grid (9,2);
    \draw[step=1.0,black, dotted] (10,0) grid (12,2);
    \draw[black] (7,0) rectangle (9,2);
    \draw[black] (10,0) rectangle (12,2);
     \draw[->,color=blue] (8.5,1.5) to [bend right,out=0,in=180] (10.5,1.5);
     \draw[->,color=red] (8.5,0.5) to [bend right,out=0,in=180] (10.5,0.5);
     \draw[->,color=red] (7.5,0.5) to [bend left,out=-90,in=-90] (10.5,0.5);
     \draw[->,color=blue] (7.5,1.5) to [bend right,out=90,in=90] (10.5,1.5);
        \node at (9.5,-1) {$(b)$};

 \node[scale=2,fill=red,,opacity=0.5] at (17.5,0.5) {$ $};
     \node[scale=2,fill=red,,opacity=0.5] at (18.5,1.5) {$ $};
     \node[scale=2,fill=red,,opacity=0.5] at (15.5,1.5) {$ $};
     \node[scale=2,fill=red,,opacity=0.5] at (15.5,0.5) {$ $};
     
    \draw[step=1.0,black, dotted] (14,0) grid (16,2);
    \draw[step=1.0,black, dotted] (17,0) grid (19,2);
    \draw[black] (14,0) rectangle (16,2);
    \draw[black] (17,0) rectangle (19,2);
     \draw[->,color=blue] (14.5,0.5) to [bend right,out=0,in=180] (17.5,1.5);
     \draw[->,color=red] (15.5,0.5) to [bend right,out=0,in=180] (18.5,1.5);
     \draw[->,color=red] (15.5,1.5) to [bend right,out=90,in=90] (18.5,1.5);
     \draw[->,color=blue] (14.5,1.5) to [bend right,out=90,in=90] (17.5,1.5);
        \node at (16.5,-1) {$(c)$};

     \node[scale=2,fill=red,opacity=0.5] at (4.5,5.5) {$ $};
     \node[scale=2,fill=red,,opacity=0.5] at (5.5,6.5) {$ $};
     \node[scale=2,fill=red,,opacity=0.5] at (7.5,5.5) {$ $};
     \node[scale=2,fill=red,,opacity=0.5] at (8.5,6.5) {$ $};
     
    \draw[step=1.0,black, dotted] (4,5) grid (6,7);
    \draw[step=1.0,black, dotted] (7,5) grid (9,7);
    \draw[black] (4,5) rectangle (6,7);
    \draw[black] (7,5) rectangle (9,7);
    \draw[->,color=red] (4.5,5.5) to [bend right,out=-90,in=-90] (7.5,5.5);
    \draw[->,color=blue] (5.5,5.5) to [bend right,out=-90,in=-90] (8.5,5.5);
    \draw[->,color=blue] (4.5,6.5) to [bend left,out=90,in=90] (7.5,6.5);
    \draw[->,color=red] (5.5,6.5) to [bend left,out=90,in=90] (8.5,6.5);

    \node at (6.5,4) {$(d)$};  
    
       \node[scale=2,fill=red,opacity=0.5] at (10.5,5.5) {$ $};
     \node[scale=2,fill=red,,opacity=0.5] at (11.5,6.5) {$ $};
     \node[scale=2,fill=red,opacity=0.5] at (10.5,6.5) {$ $};
     \node[scale=2,fill=red,,opacity=0.5] at (11.5,5.5) {$ $};
     \node[scale=2,fill=red,,opacity=0.5] at (13.5,5.5) {$ $};
     \node[scale=2,fill=red,,opacity=0.5] at (14.5,6.5) {$ $};
     \draw[step=1.0,black, dotted] (13,5) grid (15,7);
    \draw[step=1.0,black, dotted] (10,5) grid (12,7);
    \draw[black] (13,5) rectangle (15,7);
    \draw[black] (10,5) rectangle (12,7);   
    \draw[->,color=red] (10.5,5.5) to [bend right,out=-90,in=-90] (13.5,5.5);
    \draw[->,color=red] (10.5,6.5) to [bend left,out=90,in=90] (13.5,5.5);
    \draw[->,color=red] (11.5,6.5) to (13.5,5.5);
    \draw[->,color=red] (11.5,5.5) to (13.5,5.5);
    \node at (12.5,4) {$(e)$};  
\end{tikzpicture}    
\end{center}
\end{example}

Next result states that the operations of saturation and reduction commute for simple  tableaux. This is the first step to compute the saturation of any tableau.
\begin{proposition}\label{Prop:simple}
Let $I=\{i_0,i_1\}\times \{j_0,j_1\}$ with $i_0<i_1$ and $j_0<j_1$. Let $T,T'$ be two simple tableaux such that $T|_I=T$ and $T'|_I=T'$ and let $X$ and $X'$ be two $2\times 2$-tableaux.  We have the following properties
\begin{enumerate}
    \item\label{enum-simple1} The tableau $\Sat(T)$ is simple and $\Sat(T)=\Sat(T)|_I$.
    \item\label{enum-simple2} If $T$ and $T'$ are Nerode equivalent then $\Red_I(T)$ and $\Red_I(T')$ are Nerode equivalent.
    \item\label{enum-simple3} If $X$ and $X'$ are Nerode equivalent then $\Ind_I(X)$ and $\Ind_I(X')$ are Nerode equivalent.
    \item\label{enum-simple4} We have $\Sat(\Red_I(T))=\Red_I(\Sat(T))$.
\end{enumerate}
\end{proposition}
\begin{proof}
\begin{enumerate}
    \item  We have to prove that $\Sat(T)\subset I$.\\
 Let $(i,j)\in\Sat(T)$. If $(i,j)\not\in I$ then $(i,j)$ matches with one of the following three cases
 \begin{enumerate}
     \item If $i\not\in \{i_0,i_1\}$ and $j\not\in\{j_0,j_1\}$ then we consider two crosses: a first one, $(\alpha,\beta)$, belonging to the final zone, and a second one, $(\alpha',\beta')$  not belonging to the final zone. Let $f$ and $g$ be such that $f(i_0)=f(i_1)=\alpha'$, $g(j_0)=g(j_1)=\beta'$, $f(i)=\alpha$, and $g(j)=\beta$. We have $d^{(f,g)}(T)=\{(\alpha',\beta')\}$ which is not final, while $ d^{(f,g)}(\Sat(T))$ is final because $(\alpha,\beta)\in d^{(f,g)}(\Sat(T))$. This contradicts the fact that $T$ and  $\Sat(T)$ are in the same Nerode class.
     \item If $i\in\{i_0,i_1\}$ and $j\not\in \{j_0,j_1\}$ then, according to Remark \ref{rem-nondeg}, there exist $\alpha,\beta,\beta'$ such that $(\alpha,\beta)$ is in the final zone while $(\alpha,\beta')$ is not. Let $f$ and $g$ be such that $f(i_0)=f(i_1)=\alpha$, $g(j_0)=g(j_1)=\beta$ and $g(j)=\beta'$. The tableau  $d^{(f,g)}(T)=\{(\alpha,\beta)\}$ is  not final but, since  $(\alpha',\beta')\in d^{(f,g)}(\Sat(T))$, the tableau $ d^{(f,g)}(\Sat(T))$ is final. This contradicts the Nerode equivalence of $T$ and $\Sat(T)$.
     \item The remaining case ($i\not \in\{i_0,i_1\}$ and $j\in \{j_0,j_1\}$) is obtained symmetrically to the prevrious one.
 \end{enumerate}
 Accordingly to the previous enumeration, we have necessarily $(i,j)\in I$. So $\Sat(T)\subset I$, \ie $\Sat(T)$ is simple. 
 \item  Since we are not in a degenerated case, there exist $J=\{k_0,k_1\}\times\{\ell_0,\ell_1\}$ with $k_0\neq k_1$ and $\ell_0\neq \ell_1$ such that for any $(\alpha,\beta)\in\IntEnt2\times\IntEnt2$ we have $(k_\alpha,\ell_\beta)\in F_J$ if and only if $(\alpha,\beta)\in F_R$. Suppose that $\Red(T)\not\sim\Red(T')$ then, without loss of generality, there exists $(f,g)\in\IntEnt2^{\IntEnt 2}\times 2^{\IntEnt 2}$ such that $d^{(f,g)}(\Red(T))$ is final while $d^{(f,g)}(\Red(T'))$ is not final. This implies thet there exists $(a,b)\in\Red(T)$ such that $(f(a),g(b))\in F_R$ while for any $(\alpha,\beta)\in \Red(T')$ we have $(f(\alpha),g(\beta))\not \in F_R$. Consider the two following map
 \[
 f^I(i)=\left\{\begin{array}{ll} k_{f(\alpha)}&\mbox{ if }i=i_\alpha\mbox{ for some }\alpha\in\IntEnt2\\
 i&\mbox{ otherwise}\end{array}\right.\mbox{ and }
 g^I(j)=\left\{\begin{array}{ll} \ell_{g(\beta)}&\mbox{ if }j=j_\beta\mbox{ for some }\beta\in\IntEnt2\\
 j&\mbox{ otherwise}\end{array}\right.
 \]
 We check that $d^{(f^I,g^I)}(T)$ is final while $d^{(f^I,g^I)}(T')$ is non final. This shows that $T\not\sim T'$ and proves the result.
 \item Suppose that $\Ind_I(X)\not\sim \Ind_I(X')$. This means that there exists $(f,g)\in\IntEnt m^{\IntEnt m}\times \IntEnt n^{\IntEnt n}$ such that $d^{(f,g)}(\Ind_I(X))$ and $d^{(f,g)}(\Ind_I(X'))$ have not the same finality. Let $J=\{f(i_0),f(i_1)\}\times\{g(j_0),g(j_1)\}=\{k_0,k_1\}\times\{\ell_0,\ell_1\}$ with $k_0\leq k_1$ and $\ell_0\leq \ell_1$. Lemma \ref{T2Red} implies that there exists a pair of maps $(f_J,g_J)$ such that for any $\alpha,\beta\in\IntEnt2$, $(f_J(k_\alpha),g_J(\ell_\beta))\in F_J$ if and only if $(\alpha,\beta)\in F_R$. We set
\[
\widehat f=\left\{\begin{array}{ll}f_J&\mbox{if }k_0=f(i_0)\\
\alpha\rightarrow 1-f_J(\alpha)&\mbox{if }k_0\neq f(i_0)\end{array}\right. \mbox{ and }
\widehat g=\left\{\begin{array}{ll}g_J&\mbox{if }\ell_0=g(j_0)\\
\beta\rightarrow 1-g_J(\beta)&\mbox{if }\ell_0\neq f(j_0)\end{array}\right.
\]
We check that $(\widehat f(\alpha),\widehat g(\beta))\in F_R$  if and only if $(f(i_\alpha),g(i_\beta))\in F_J$.  Hence, $d^{(\widehat f,\widehat g)}(X)$ and $d^{(\widehat f,\widehat g)}(X')$ have not the same finality and so $X\not\sim X'$. This proves the result.
 \item
Let $R=\Ind_I(\Sat(\Red(T)))$. Recall that we have
\begin{equation}
    \Red(R)=\Sat(\Red(T)),\label{RedT'2SatT}
\end{equation}
see Figure \ref{fig-constT'} for an illustration. 
\begin{figure}[H]
  \centerline{
            \begin{tikzpicture}[scale=0.4]
           \draw[step=1.0,black, thin] (0,0) grid (5,5);
           \node at (1.5,1.5) {$3$};
           \node at (3.5,1.5) {$4$};
           \node at (1.5,3.5) {$1$};
           \node at (3.5,3.5) {$2$};
            \node at(2.5,-1) {$T$};
            \draw[->] (5.5,2.5) to(7.5,2.5);
            \node at (6.5,3.5) {$\Red$};
           \draw[step=1.0,black, thin] (8,2) grid (10,4);
          \node at (8.5,2.5) {$3$};
           \node at (9.5,2.5) {$4$};
           \node at (8.5,3.5) {$1$};
           \node at (9.5,3.5) {$2$};
            \node at(9,1) {$\Red(T)$};
            \draw[->] (10.5,2.5) to(12.5,2.5);
             \node at (11.5,3.5) {$\Sat$};
        \draw[step=1.0,black, thin] (13,2) grid (15,4);
          \node at (13.5,2.5) {$3'$};
           \node at (14.5,2.5) {$4'$};
           \node at (13.5,3.5) {$1'$};
           \node at (14.5,3.5) {$2'$};
           \node at(14,0) {$\begin{array}{c}\Sat(\Red(T))\\=\\\Red(R)\end{array}$};
            \draw[<-] (15.5,2.5) to(17.5,2.5);
            \node at(16.5,3.5) {$\Red$};
        \draw[step=1.0,black, thin] (18,0) grid (23,5);  
                 \node at (19.5,1.5) {$3'$};
           \node at (21.5,1.5) {$4'$};
           \node at (19.5,3.5) {$1'$};
           \node at (21.5,3.5) {$2'$};
           \node at (20.5,-1) {$R$};
           \draw[step=1.0,black, thin] (9,12) grid (14,17);
             \node at (10.5,13.5) {$3''$};
           \node at (12.5,13.5) {$4''$};
           \node at (10.5,15.5) {$1''$};
           \node at (12.5,15.5) {$2''$};
           \node at (11.5,18.5) {$R'$};
           \draw[step=1.0,black, thin] (11,7) grid (13,9);
          \node at (11.5,7.5) {$3''$};
          \node at (11.5,8.5) {$1''$};
          \node at (12.5,7.5) {$4''$};
          \node at (12.5,8.5) {$2''$};
           \node at (12,6) {$\Red(\Sat(T))$};
           \node at (4,10.5) {$\Sat$};
           \node at (18.5,10.5) {$\Sat$};
           \node at (13,10.5) {$\Red$};
           \draw[->] (2.5,5.5) to (8.5,14.5);
           \draw[->,dotted] (20,5.5) to (14.5,14.5);
           \draw[->] (11.5,11.5) to (12,9.5);
          \end{tikzpicture}}
          
          \caption{Illustration of the construction of $R$ and $R'$ for $I=\{1,3\}\times\{1,3\}$. \label{fig-constT'}}
\end{figure}
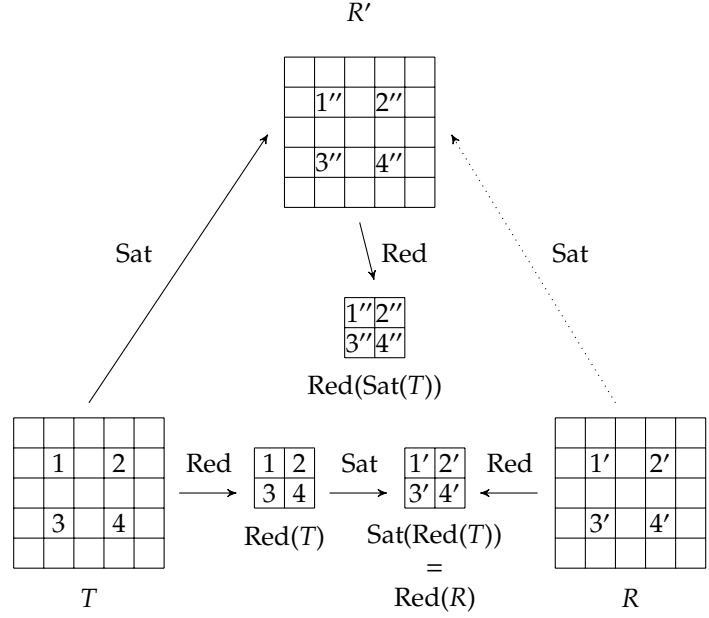
From Property \ref{enum-simple3} of the statement, we have $T\sim R$ because \[T=\Ind_I(\Red(T)),\ R=\Ind_I(\Sat(\Red(T)))\mbox{ and }\Red(T)\sim\Sat(\Red(T)).\]
Let $R'=\Sat(T)$. From Property \ref{enum-simple1}, the tableau $R'$ is also simple for $I$. Since $T\sim R$ we have also $\Sat(R)=R'$ and then $\#R\leq\#R'$ which implies $\#\Red(T')\leq\#\Red(R')$. But from Property \ref{enum-simple2}, the tableaux $\Red(R)=\Sat(\Red(T))$ and $\Red(R')$ are Nerode equivalent. Hence $\Sat(\Red(R'))=\Red(R)$ and so $\#\Red(R')\leq\#\Red(R)$. It follows that $\#\Red(R')=\#\Red(R)$ and, as a consequence of  Proposition \ref{prop-lattice}, the we have $\Sat(\Red(T))=\Red(R)=\Red(R')=\Red(\Sat(T))$.
\end{enumerate}
\end{proof}
The meaning of this result is that, in the case of simple tableaux, the saturation reduces to the $2\times 2$ case. 
We now investigate the general case and prove that the saturated of any tableau can be obtained by iterated saturation on  $2\times 2$ sub-tableaux.
\begin{definition}\label{def-arrow}
          Let $T\in 2^{\IntEnt{m}\times\IntEnt{n}}$. We write $T\ \rightarrow\  T'$ if $T'=T\cup \Sat(T|_I)$ for any $I=\{i_0,i_1\}\times \{j_0,j_1\}\in\IntEnt{m}\times \IntEnt{n}$. 
          We denote by ${\displaystyle \mathop{\rightarrow}^*}$ the transitive closure of $\rightarrow$.
           \end{definition}
We say that a tableau $T$ is \emph{locally saturated} when $T\displaystyle\mathop\rightarrow^* T'$ implies $T=T'$. We denote by $\LSV_\bullet^{F_1,F_2}$ the set of the locally saturated valid tableaux.

\begin{lemma}\label{lem-arrow2sim}
      If $T\rightarrow T'$  then $T\sim T'$.
\end{lemma}
\begin{proof} Let $I=\{i_0,i_1\}\times \{j_0,j_1\}\in\IntEnt{m}\times \IntEnt{n}$ such that 
        $\Sat(T|_{I})= T'|_{I}$ and $T\setminus I=T'\setminus I$.
Let $(f,g)\in\IntEnt{m}^{\IntEnt{m}}\times\IntEnt{n}^{\IntEnt{n}}$. Suppose first that $d^{(f,g)}(T)$ is final. Since $T\subset T'$, the tableau $d^{(f,g)}(T')$ is also final.
Conversely suppose that $d^{(f,g)}(T)$ is not final and assume that $d^{(f,g)}(T')$ is final. This implies that there exists $(i,j)\in \Sat(T|_I)\setminus T$ such that $(f(i),g(j))$ is in the final zone. This implies that  $d^{(f,g)}(\Sat(T|_I))$ is final. But since $T|_I\subset T$, the tableau $d^{(f,g)}(T|_I)$ is necessarily not final. The two previous results contradict the fact that $T|_I\sim \Sat(T|_I)$. Hence, the tableau  $d^{(f,g)}T'$ is not final and $T\sim T'$.
\end{proof}

As a consequence of Lemma \ref{lem-arrow2sim}, if a tableau is saturated then it is also locally saturated. Now, let us prove the converse.
\begin{lemma}\label{lem-interval}
        Let $T\sim T''$ with $T\subset T''$. For any $T\subset T'\subset T''$, we have $T\sim T'$. In other words, each Nerode class is a union of intervals for the inclusion order.
\end{lemma}
\begin{proof}
 First, suppose that $d^{(f,g)}(T)$ is final. Since $T\subset T'$, we have  $d^{(f,g)}(T)\subset d^{(f,g)}(T')$ and the tableau $d^{(f,g)}(T')$ is also final.\\
 Conversely, suppose that $d^{(f,g)}(T')$ is final then, since $T'\subset T''$, the tableau $d^{(f,g)}(T'')$ is also final. But $T\sim T''$ implies that $d^{(f,g)}(T)$ is final.
 We deduce that $T\sim T'$.
\end{proof}

According to the Proposition \ref{Prop:simple}, 
each step of local saturation works as in the $2\times 2$ case. So we have to investigate the different possibilities with respect to the type of the operation $\bullet$ as described in Section \ref{2x2case}.
We first examine the case where $\bullet$ has the type (A).
\begin{lemma}\label{lem-T2T+c}
        For any type (A) operation $\bullet$,  any tableau $T$ and any $(i_c,j_c)\not\in T$, the tableaux $T$ and $T\cup\{(i_c,j_c)\}$ are not Nerode equivalent. 
\end{lemma}
\begin{proof}
 First we remark that, since $\bullet$ is a non degenerated type A operation there exist $\alpha,\beta\in\IntEnt m$ and $\gamma,\delta\in\IntEnt n$ such that $(\alpha,\gamma)\in F$ and $\{(\alpha,\delta), (\beta,\gamma),(\beta,\delta)\}\cap F=\emptyset$.
 We consider the two maps
 \[
 f(i)=\left\{\begin{array}{ll} \alpha&\mbox{if }i=i_c,\\ \beta &\mbox{otherwise}\end{array}\right.\mbox{ and }
 g(j)=\left\{\begin{array}{ll} \gamma&\mbox{if }j=j_c,\\ \delta &\mbox{otherwise}\end{array}\right..
 \]
 Obviously $d^{(f,g)}(T)\subset\{(\alpha,\delta),(\beta,\gamma),(\beta,\delta)\}$ is non final while $d^{(f,g)}(T\cup\{(i_c,j_c)\})$ is final since it contains $(\alpha,\gamma)$. It follows that the states $T$ and $T\cup\{(i_c,j_c)\}$ are not Nerode equivalent.
\end{proof}
\begin{example}\rm
Let us illustrate the proof of Lemma \ref{lem-T2T+c}. Consider the two following tableaux\\
\centerline{
  \begin{tikzpicture}[scale=0.4]
          ;
               \foreach \j in {0,...,2} { \foreach \i in{3,...,5} {
        \node[scale=1.5,fill=red,,opacity=0.5] at (\j+0.5,\i+0.5) {$ $};
        \node[scale=1.5,fill=red,,opacity=0.5] at (\j+12.5,\i+0.5) {$ $};
        }
        }
              \draw[step=1.0,black, thin] (0,0) grid (6,6);
              \draw[step=1.0,black, thin] (12,0) grid (18,6);
 \foreach \j in {0,...,5} { \foreach \i in{0,...,5} {
        \node at (\j+0.5,\i+0.5) {$?$};
        \node at (\j+12.5,\i+0.5) {$?$};
        }
        }
      \node[scale=1.5,fill=white] at (2.5,1.5) {$ $};
      \node[scale=1.5,fill=white] at (14.5,1.5) {$ $};
      \node[scale=1.5] at (14.5,1.5) {$\times$};
   \node at (3,-1){$T$} ;      
   \node at (15,-1){$T'=T\cup\{(i_c,j_c)\}$} ;
          \end{tikzpicture}}
          In these pictures, the red zones illustrate the common final zone of $T$ and $T'$, and the question marks denote the same values at the same cells in both tableaux. Notice that, as specified in the statement of Lemma \ref{lem-T2T+c}, the tableaux $T$ and $T'$ only differ at the cell $(i_c,j_c)$.\\
We choose a $2\times 2$ zone that behaves as a non degenerate $2\times 2$ tableau. This zone is drawn in green in the following picture:\\
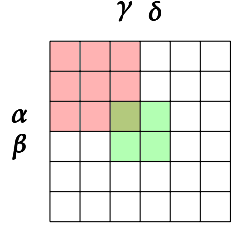
\centerline{
  \begin{tikzpicture}[scale=0.4]
          ;
               \foreach \j in {0,...,2} { \foreach \i in{3,...,5} {
        \node[scale=1.5,fill=red,,opacity=0.3] at (\j+0.5,\i+0.5) {$ $};
        \node at (-1,3.5) {$\alpha$};
        \node at (-1,2.5) {$\beta$};
        \node at (2.5,7) {$\gamma$};
        \node at (3.5,7) {$\delta$};
        }
        }
          \foreach \j in {2,...,3} { \foreach \i in{2,...,3} {
        \node[scale=1.5,fill=green,,opacity=0.3] at (\j+0.5,\i+0.5) {$ $};
        }
        }
              \draw[step=1.0,black, thin] (0,0) grid (6,6);
  %
          \end{tikzpicture}}
Below, we illustrate the image by $(f,g)$ of each cell. The cells of the green zone are denoted by $A$, $B$, $C$ and $D$ in the left tableau. The cells in the right tableau are labelled by $A$, $B$, $C$ or $D$ according to their images by $(f,g)$.  \\
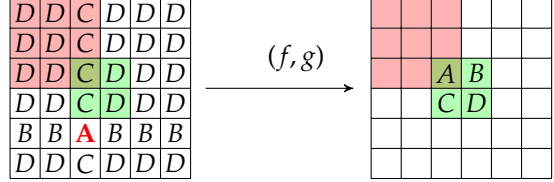
\centerline{
  \begin{tikzpicture}[scale=0.4]
               \foreach \j in {0,...,2} { \foreach \i in{3,...,5} {
        \node[scale=1.5,fill=red,,opacity=0.3] at (\j+0.5,\i+0.5) {$ $};
        \node[scale=1.5,fill=red,,opacity=0.3] at (\j+12.5,\i+0.5) {$ $};
        }
        } \foreach \j in {2,...,3} { \foreach \i in{2,...,3} {
        \node[scale=1.5,fill=green,,opacity=0.3] at (\j+0.5,\i+0.5) {$ $};
                \node[scale=1.5,fill=green,,opacity=0.3] at (\j+12.5,\i+0.5) {$ $};
        }}
              \draw[step=1.0,black, thin] (0,0) grid (6,6);
              \draw[step=1.0,black, thin] (12,0) grid (18,6);
\node at (14.5,3.5) {$A$};
\node at (15.5,3.5) {$B$};
\node at (14.5,2.5) {$C$};
\node at (15.5,2.5) {$D$};
\foreach \j in {0,1,3,4,5} {
\node at (\j+0.5,1.5) {$B$};
\foreach \i in {0,2,3,4,5} {
\node at (\j+0.5,\i+0.5) {$D$};
}
}
\foreach \i in {0,2,3,4,5} {
\node at (2.5,\i+0.5) {$C$};
}
\draw[->] (6.5,3) to (11.5,3);
\node at(9.5,4) {$(f,g)$};
\node[color=red] at(2.5,1.5) {$\mathbf A$};
          \end{tikzpicture}}
From the previous computation, the images of $T$ and $T'$ are respectively

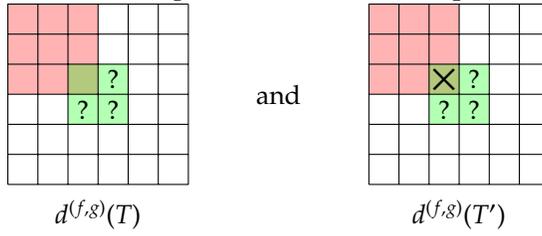
\centerline{
  \begin{tikzpicture}[scale=0.4]
          ;
               \foreach \j in {0,...,2} { \foreach \i in{3,...,5} {
        \node[scale=1.5,fill=red,,opacity=0.3] at (\j+0.5,\i+0.5) {$ $};
        \node[scale=1.5,fill=red,,opacity=0.3] at (\j+12.5,\i+0.5) {$ $};
        }
        } \foreach \j in {2,...,3} { \foreach \i in{2,...,3} {
        \node[scale=1.5,fill=green,,opacity=0.3] at (\j+0.5,\i+0.5) {$ $};
                \node[scale=1.5,fill=green,,opacity=0.3] at (\j+12.5,\i+0.5) {$ $};
        }}
              \draw[step=1.0,black, thin] (0,0) grid (6,6);
              \draw[step=1.0,black, thin] (12,0) grid (18,6);
\node[scale=1.5] at (14.5,3.5) {$\times$};
\node at (14.5,2.5) {$?$};
\node at (15.5,2.5) {$?$};
\node at (15.5,3.5) {$?$};
\node at (2.5,2.5) {$?$};
\node at (3.5,2.5) {$?$};
\node at (3.5,3.5) {$?$};
  \node at (3,-1){$d^{(f,g)}(T)$} ;      
  \node at (15,-1){$d^{(f,g)}(T')$} ;
  \node at (9,3) {and};
         \end{tikzpicture}}
\end{example}

\begin{proposition}\label{prop-derodA}
  If $\bullet$ is a non degenerated type A operation then any Nerode class of ${\widehat M}_\bullet^{F_1,F_2}$ is a singleton.
\end{proposition}
\begin{proof}
 Suppose that $T\sim T'$ and $T\neq T'$. From Proposition \ref{prop-lattice}, we have $T\sim T\cup T'$. Without loss of generality, we assume that there exists $(i,j)\in T'$ such that $(i,j)\not\in T$. From Lemma \ref{lem-interval}, $T$ is equivalent to any tableau $T''$ satisfying $T\subset T''\subset T'$. In particular we have $T\equiv T\cup\{(i,j)\}$. But this contradicts Lemma \ref{lem-T2T+c}. So we deduce that $T\not\sim T'$ or $T=T'$.
\end{proof}

Lemma \ref{lem-T2T+c} generalizes to other types as follows:
\begin{lemma}\label{lem-T2Tgen}
        For any non degenerated operation $\bullet$ we have $T\not\sim T\cup \{(i_c,j_c)\}$ for any locally saturated tableau $T$ and any $(i_c,j_c)\not\in T$.
\end{lemma}
\begin{proof}
 Proposition \ref{prop-derodA} implies the assertion for the type A.
Let us now consider two sets $R=(\IntEnt{m}\times \{j_c\})\cap T$ and $C=(\{i_c\}\times \IntEnt{n})\cap T$. 
 \begin{enumerate}
     \item If $\bullet$ has the type O then either $R=\emptyset$ or $C=\emptyset$. Indeed, if there exist $j$ such that $(i_c,j)\in C$ and $i$ such that $(i,j_c)\in R$ then the restriction $T|_{\{i,i_c\}\times\{j,j_c\}}$ is not saturated because, from Proposition \ref{Prop:simple}, $\Red(T|_{\{i,i_c\}\times\{j,j_c\}})$ is not saturated, and so $T{\displaystyle \mathop\rightarrow^*} T\cup \{(i_c,j_c),(i,j)\}\neq T$ and this contradicts the fact that $T$ is locally saturated. Without loss of generality, we assume that $R=\emptyset$, the other case being obtained symmetrically. We recall that there exist $I,J$ such that  $F=(I\times\IntEnt n)\cup (\IntEnt m\times J)$. Since $F\neq \emptyset$ and $F\neq \IntEnt{m}\times \IntEnt{n}$, there exist $\alpha,\gamma,\delta$ such that $(\alpha,\gamma)\in F$ and $(\alpha,\delta)\not \in F$. Indeed, it suffices to choose $(\alpha,\gamma)\in (\IntEnt m\setminus I)\times J$ and $\delta\in \IntEnt n\setminus J$. Let $f:\IntEnt m\rightarrow \IntEnt m$ and $g:\IntEnt n\rightarrow \IntEnt n$ such that
     \[
     f(i)=\alpha\mbox{ and }
     g(j)=\left\{\begin{array}{ll}\gamma&\mbox{if }j=j_c\\\delta&\mbox{otherwise}\end{array}\right.
     \]
 Since $R=\emptyset$, for any $(i,j)\in T$ we have $(f(i),g(j))=(\alpha,\delta)\not\in F$. This implies that $d^{(f,g)}(T)$ is not final. But since $(f(i_c),g(j_c))=(\alpha,\gamma)\in F$, we obtain also that $d^{(f,g)}\left(T\cup\{(i_c,j_c)\}\right)$ is final. So $T$ and $T\cup\{(i_c,j_c)\}$ are not Nerode equivalent.
 
 \item Suppose that $\bullet$ has the type X.  If $R=\emptyset$ or $C=\emptyset$ then we prove our result by applying the same strategy as is case 1. So we suppose that $R\neq\emptyset$ and $C\neq\emptyset$. Let us first prove that the result reduces to the case where $\#R=\#C=1$. We  remark that for any $i,j$ such that $(i,j_c)\in C$ and $(i_c,j)\in R$  we have $(i,j)\not\in T$, otherwise, as $T$ is saturated, $(i_c,j_c)$ is in $T$. Let $i',j'$ such that $(i',j_c)\in C$ and $(i_c,j')\in R$.
 Let $f_1:\IntEnt m\rightarrow\IntEnt m$ and $g_1:\IntEnt n\rightarrow\IntEnt n$ such that 
 \[
 f_1(i)=\left\{\begin{array}{ll}i'&\mbox{ if }(i,j_c)\in C\\
 i&\mbox{otherwise} \end{array}\right.\mbox{ and }
 g_1(j)=\left\{\begin{array}{ll}j'&\mbox{ if }(i_c,j)\in R\\
 j&\mbox{otherwise} \end{array}\right.
 \]
 We set $T_2=d^{(f_1,g_1)}(T)$, $T'_2=d^{(f_1,g_1)}(T\cup\{(i_c,j_c)\})$, $C_2=(\IntEnt m\times j_c)\cap T_2$, and $R_2=(i_c\times \IntEnt n)\cap T_2$.  The tableau $T_2$ is obtained by removing each row identical to the $j_c$ one and each column identical to the $i_c$ one. As a consequence, $T_2$ is locally saturated and $(i_c,j_c)\not\in T_2$. We check also that  $T'_2=T_2\cup\{(i_c,j_c)\}$, and $\#R_2=\#C_2=1$.

Consider now $g_2:\IntEnt n \rightarrow \IntEnt n$ such that 
$$g_2(j)=\left\{\begin{array}{ll}
j_c&\text{if }(i',j)\in T_2\\
j'&\text{otherwise}
\end{array}
\right.$$
and set $T_3=d^{(\Id,g_2)}(T_2)$ and $T'_3=d^{(\Id,g_2)}(T'_2)$.
The tableau $T_3$ contains crosses  only in the two columns $j'$ and $j_c$.

More precisely we have
\begin{equation}\label{eq-T3c}
T_3\cap\left(\IntEnt m\times\{j_c\}\right)=\{
(i,j_c)\mid\exists j\in\IntEnt n \mbox{ such that }(i,j),(i',j)\in T_2
\}
\end{equation}
and 
\begin{equation}\label{eq-T3cp}
T_3\cap\left(\IntEnt m\times\{j'\}\right)=\{
(i,j')\mid\exists j\in\IntEnt n \mbox{ such that }(i,j)\in T_2\mbox{ and }(i',j)\not\in T_2
\}
\end{equation}
Suppose that $(i_c,j_c)\in T_3$ then, by Formula (\ref{eq-T3c}), there exists $j\in\IntEnt n$ such that $(i_c,j)\in T_2$ and $(i',j)\in T_2$. Since $(i',j_c)\in T_2$ and $T_2$ is locally saturated, this implies $(i_c,j_c)\in T_2$. This contradicts the definition of $T_2$. Hence, $(i_c,j_c)\not\in T_3$.
Furthermore, straightforwardly from (\ref{eq-T3cp}), $(i',j')\not\in T_3$.\\
We have also
\[
T'_3\cap(\IntEnt m\times\{j_c\})=\{(i,j_c)\mid \exists j\in\IntEnt n\mbox{ such that } (i,j)\in T'_2\mbox{ and }(i',j)\in T_2\},
\]
and
\[
T'_3\cap(\IntEnt m\times\{j'\})=\{(i,j')\mid \exists j\in\IntEnt n\mbox{ such that } (i,j)\in T'_2\mbox{ and }(i',j)\not\in T_2\}.
\]
Hence we check that $\{(i_c,j'), (i',j_c),(i_c,j_c)\}\subset T'_3$.
%
Consider now the map $f_2:\IntEnt m\rightarrow \IntEnt m$ such that
$$f_2(i)=\left\{\begin{array}{ll}
i_c&\text{if }(i,j')\in T_3\\
i'&\text{otherwise}
\end{array}
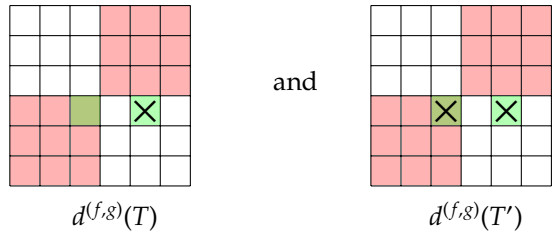
\right.$$
and set $T_4=d^{(f_2,\Id)}(T_3)$ and $T'_4=d^{(f_2,\Id)}(T'_3)$.
By construction, $T_4, T'_4\subset\{i',i_c\}\times\{j',j_c\}$. From the definition of $f_2$, $(i',j')\not\in T_4$. If we suppose that $(i_c,j_c)\in T_4$ then there exists $i\in\IntEnt m$ such that $(i,j_c)\in T_3$ and $(i,j')\in T_3$. Since $(i,j_c)\in T_3$, Formula (\ref{eq-T3c}) implies that there exists $j_1\in\IntEnt n$ such that $(i,j_1)\in T_2$ and $(i',j_1)\in T_2$. Since $(i,j')\in T_3$, Formula (\ref{eq-T3cp}) implies that there exists $j_2\in \IntEnt m$ such that $(i,j_2)\in T_2$ and $(i',j_2)\not\in T_2$. To summarize we have $\{(i,j_1),(i',j_1),(i,j_2)\}\subset T_2$ but $(i',j_2)\not\in T_2$. This contradicts the fact that $T_2$ is saturated. Hence, $(i_c,j_c)\not\in T_4$. On the other hand, we have $\{(i_c,j'),(i',j_c),(i_c,j_c)\}\subset T'_4$ because $\{(i_c,j'),(i',j_c),(i_c,j_c)\}\subset T'_3$, $f_2(i_c)=i_c$, and $f_2(i')=i'$.

From Proposition \ref{Prop:simple} and Section \ref{2x2case}, since $T_4\subset\{(i_c,j'),(i',j_c)\}$, the two tableaux $T_4$ and $T'_4$ are clearly non-equivalent.
 \end{enumerate}
\end{proof}
\begin{example}\rm
We illustrate first the case of the type (O). Let us consider the two following tableaux
\centerline{
  \begin{tikzpicture}[scale=0.4]
          ;
              \foreach \j in {0,...,5} { \foreach \i in{0,...,5} {
        \node at (\j+0.5,\i+0.5) {$?$};
        \node at (\j+12.5,\i+0.5) {$?$};
        }
        }
       \foreach\i in{0,...,5}{
       \node[scale=1.5,fill=white] at (3.5,\i+0.5) {$ $};
       \node[scale=1.5,fill=white] at (15.5,\i+0.5) {$ $};
      \node[scale=1.5,fill=white] at (\i+12.5,1.5) {$ $};
      \node[scale=1.5,fill=white] at (\i+0.5,1.5) {$ $};
      }
      \node[scale=1.5,color=red] at (15.5,1.5) {$\mathbf\times$};
      \node[scale=1.5,color=blue] at (17.5,1.5) {$\mathbf\times$};
      \node[scale=1.5,color=blue] at (13.5,1.5) {$\mathbf\times$};
      \node[scale=1.5,color=blue] at (1.5,1.5) {$\mathbf\times$};
      \node[scale=1.5,color=blue] at (5.5,1.5) {$\mathbf\times$};
               \foreach \j in {0,...,2} { \foreach \i in{0,...,5} {
        \node[scale=1.5,fill=red,,opacity=0.5] at (\j+0.5,\i+0.5) {$ $};
        \node[scale=1.5,fill=red,,opacity=0.5] at (\j+12.5,\i+0.5) {$ $};
        }
        }
           \foreach \j in {3,...,5} { \foreach \i in{3,...,5} {
        \node[scale=1.5,fill=red,,opacity=0.5] at (\j+0.5,\i+0.5) {$ $};
        \node[scale=1.5,fill=red,,opacity=0.5] at (\j+12.5,\i+0.5) {$ $};
        }
        }
              \draw[step=1.0,black, thin] (0,0) grid (6,6);
              \draw[step=1.0,black, thin] (12,0) grid (18,6);
   \node at (3,-1){$T$} ;      
   \node at (15,-1){$T'=T\cup\{(i_c,j_c)\}$} ;
          \end{tikzpicture}}
          In this figure, the red zone corresponds to the final zone of the tableaux, the red cross indicates the cell $(i_c,j_c)$ and the set of the blue crosses is the set $R$. Since $R\neq \emptyset$ we have $C=\emptyset$ as pictured above.\\
          From the definition of $(f,g)$, the value of any cell which is on the same column as the red cross is sent to the cell $A$ while the other values are sent to the cell $B$. Notice that the image of $(i_c,j_c)$ is the letter colored in red.\\
         \centerline{
  \begin{tikzpicture}[scale=0.4]
            \foreach \j in {0,...,2} { \foreach \i in{0,...,5} {
        \node[scale=1.5,fill=red,,opacity=0.3] at (\j+12.5,\i+0.5) {$ $};
        \node[scale=1.5,fill=red,,opacity=0.3] at (\j+0.5,\i+0.5) {$ $};
        }
        }
           \foreach \j in {3,...,5} { \foreach \i in{3,...,5} {
        \node[scale=1.5,fill=red,,opacity=0.3] at (\j+12.5,\i+0.5) {$ $};
       \node[scale=1.5,fill=red,,opacity=0.3] at (\j+0.5,\i+0.5) {$ $};
        }
        }
              \draw[step=1.0,black, thin] (12,0) grid (18,6);
              \draw[step=1.0,black, thin] (0,0) grid (6,6);
\node[fill=green,scale=1.5,opacity=0.3] at (14.5,2.5){$ $};
\node[fill=green,scale=1.5,opacity=0.3] at (16.5,2.5){$ $};
\node at (14.5,2.5) {$A$};
\node at (16.5,2.5) {$B$};
\foreach \j in {0,1,2,4,5} { \foreach\i in {0,...,5}{
\node at (\j+0.5,\i+0.5) {$B$};
}}
\foreach \i in {0,...,5} {
\node at (3.5,\i+0.5) {$A$};
}
\node at (19,2.5) {$\alpha$};
\node at (16.5,6.5) {$\delta$};
\node at (14.5,6.5) {$\gamma$};
\node at (3.5,6.5) {$j_c$};
\node at (-1,1.5) {$i_c$};
\node[fill=white,scale=1.5] at (3.5,1.5) {$ $};
\node[color=red] at (3.5,1.5) {$\mathbf A$};
\draw[->] (6.5,3) to (11.5,3);
\node at(9.5,4) {$(f,g)$};
          \end{tikzpicture}}
Hence the image of the tableaux $T$ and $T'$ are respectively\\
\centerline{
  \begin{tikzpicture}[scale=0.4]
            \foreach \j in {0,...,2} { \foreach \i in{0,...,5} {
        \node[scale=1.5,fill=red,,opacity=0.3] at (\j+0.5,\i+0.5) {$ $};
        \node[scale=1.5,fill=red,,opacity=0.3] at (\j+12.5,\i+0.5) {$ $};
        }
        }
           \foreach \j in {3,...,5} { \foreach \i in{3,...,5} {
        \node[scale=1.5,fill=red,,opacity=0.3] at (\j+0.5,\i+0.5) {$ $};
       \node[scale=1.5,fill=red,,opacity=0.3] at (\j+12.5,\i+0.5) {$ $};
        }
        }
              \draw[step=1.0,black, thin] (0,0) grid (6,6);
              \draw[step=1.0,black, thin] (12,0) grid (18,6);
\node[fill=green,scale=1.5,opacity=0.3] at (2.5,2.5){$ $};
\node[fill=green,scale=1.5,opacity=0.3] at (4.5,2.5){$ $};
\node[fill=green,scale=1.5,opacity=0.3] at (14.5,2.5){$ $};
\node[fill=green,scale=1.5,opacity=0.3] at (16.5,2.5){$ $};
\node[scale=1.5] at(4.5,2.5){$\times$};
\node[scale=1.5] at(14.5,2.5){$\times$};
\node[scale=1.5] at(16.5,2.5){$\times$};
\node at (9.5,3.5) {and};
\node at (3.5,-1) {$d^{(f,g)}(T)$};
\node at (15.5,-1) {$d^{(f,g)}(T')$};
          \end{tikzpicture}}
The left tableau being non final while the right one is final shows that $T$ and $T'$ are not Nerode equivalent.\\
Notice that the same strategy works for the same pair of tableaux even in the case of the type (X), as shown in the figure below\\
\centerline{
  \begin{tikzpicture}[scale=0.4]
            \foreach \j in {0,...,2} { \foreach \i in{0,...,2} {
        \node[scale=1.5,fill=red,,opacity=0.3] at (\j+0.5,\i+0.5) {$ $};
        \node[scale=1.5,fill=red,,opacity=0.3] at (\j+12.5,\i+0.5) {$ $};
        }
        }
           \foreach \j in {3,...,5} { \foreach \i in{3,...,5} {
        \node[scale=1.5,fill=red,,opacity=0.3] at (\j+0.5,\i+0.5) {$ $};
       \node[scale=1.5,fill=red,,opacity=0.3] at (\j+12.5,\i+0.5) {$ $};
        }
        }
              \draw[step=1.0,black, thin] (0,0) grid (6,6);
              \draw[step=1.0,black, thin] (12,0) grid (18,6);
\node[fill=green,scale=1.5,opacity=0.3] at (2.5,2.5){$ $};
\node[fill=green,scale=1.5,opacity=0.3] at (4.5,2.5){$ $};
\node[fill=green,scale=1.5,opacity=0.3] at (14.5,2.5){$ $};
\node[fill=green,scale=1.5,opacity=0.3] at (16.5,2.5){$ $};
\node[scale=1.5] at(4.5,2.5){$\times$};
\node[scale=1.5] at(14.5,2.5){$\times$};
\node[scale=1.5] at(16.5,2.5){$\times$};
\node at (9.5,3.5) {and};
\node at (3.5,-1) {$d^{(f,g)}(T)$};
\node at (15.5,-1) {$d^{(f,g)}(T')$};
          \end{tikzpicture}}
\end{example}

\begin{example}\rm
Now let us illustrate  the case of the type (X) and consider the two following tableaux with $(i_c,j_c)$ denoted by the red cross on the right tableau\\
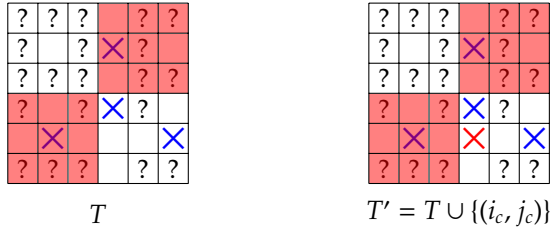
\centerline{
  \begin{tikzpicture}[scale=0.4]
          ;
              \foreach \j in {0,...,5} { \foreach \i in{0,...,5} {
        \node at (\j+0.5,\i+0.5) {$?$};
        \node at (\j+12.5,\i+0.5) {$?$};
        }
        }
       \foreach\i in{0,...,5}{
       \node[scale=1.5,fill=white] at (3.5,\i+0.5) {$ $};
       \node[scale=1.5,fill=white] at (15.5,\i+0.5) {$ $};
      \node[scale=1.5,fill=white] at (\i+12.5,1.5) {$ $};
      \node[scale=1.5,fill=white] at (\i+0.5,1.5) {$ $};
      }
      \node[scale=1.5,color=red] at (15.5,1.5) {$\mathbf\times$};
      \node[scale=1.5,color=blue] at (17.5,1.5) {$\mathbf\times$};
      \node[scale=1.5,color=blue] at (15.5,2.5) {$\mathbf\times$};
      \node[scale=1.5,color=blue] at (15.5,4.5) {$\mathbf\times$};
     \node[scale=1.5,color=blue] at (3.5,2.5) {$\mathbf\times$};
      \node[scale=1.5,color=blue] at (3.5,4.5) {$\mathbf\times$};
      \node[scale=1.5,color=blue] at (13.5,1.5) {$\mathbf\times$};
      \node[scale=1.5,color=blue] at (1.5,1.5) {$\mathbf\times$};
      \node[scale=1.5,color=blue] at (5.5,1.5) {$\mathbf\times$};
       \node[scale=1.5,fill=white] at (1.5,2.5) {$ $};
       \node[scale=1.5,fill=white] at (5.5,2.5) {$ $};
       \node[scale=1.5,fill=white] at (1.5,4.5) {$ $};
       \node[scale=1.5,fill=white] at (5.5,4.5) {$ $};
              \node[scale=1.5,fill=white] at (13.5,2.5) {$ $};
       \node[scale=1.5,fill=white] at (17.5,2.5) {$ $};
       \node[scale=1.5,fill=white] at (13.5,4.5) {$ $};
       \node[scale=1.5,fill=white] at (17.5,4.5) {$ $};
               \foreach \j in {0,...,2} { \foreach \i in{0,...,2} {
        \node[scale=1.5,fill=red,,opacity=0.5] at (\j+0.5,\i+0.5) {$ $};
        \node[scale=1.5,fill=red,,opacity=0.5] at (\j+12.5,\i+0.5) {$ $};
        }
        }
           \foreach \j in {3,...,5} { \foreach \i in{3,...,5} {
        \node[scale=1.5,fill=red,,opacity=0.5] at (\j+0.5,\i+0.5) {$ $};
        \node[scale=1.5,fill=red,,opacity=0.5] at (\j+12.5,\i+0.5) {$ $};
        }
        }
              \draw[step=1.0,black, thin] (0,0) grid (6,6);
              \draw[step=1.0,black, thin] (12,0) grid (18,6);
   \node at (3,-1){$T$} ;      
   \node at (15,-1){$T'=T\cup\{(i_c,j_c)\}$} ;
          \end{tikzpicture}}

 We first choose the column $i'$ and the row $j'$  (in green in the following figure)\\
 
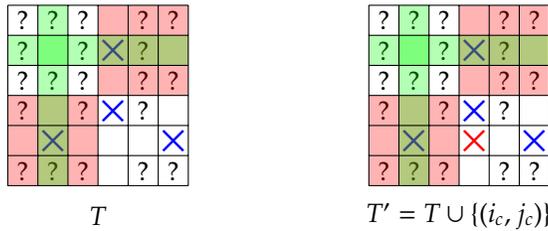
\centerline{
  \begin{tikzpicture}[scale=0.4]
          ;
              \foreach \j in {0,...,5} { \foreach \i in{0,...,5} {
        \node at (\j+0.5,\i+0.5) {$?$};
        \node at (\j+12.5,\i+0.5) {$?$};
        }
        }
       \foreach\i in{0,...,5}{
       \node[scale=1.5,fill=white] at (3.5,\i+0.5) {$ $};
       \node[scale=1.5,fill=white] at (15.5,\i+0.5) {$ $};
      \node[scale=1.5,fill=white] at (\i+12.5,1.5) {$ $};
      \node[scale=1.5,fill=white] at (\i+0.5,1.5) {$ $};
      }
      \node[scale=1.5,color=red] at (15.5,1.5) {$\mathbf\times$};
      \node[scale=1.5,color=blue] at (17.5,1.5) {$\mathbf\times$};
      \node[scale=1.5,color=blue] at (15.5,2.5) {$\mathbf\times$};
      \node[scale=1.5,color=blue] at (15.5,4.5) {$\mathbf\times$};
     \node[scale=1.5,color=blue] at (3.5,2.5) {$\mathbf\times$};
      \node[scale=1.5,color=blue] at (3.5,4.5) {$\mathbf\times$};
      \node[scale=1.5,color=blue] at (13.5,1.5) {$\mathbf\times$};
      \node[scale=1.5,color=blue] at (1.5,1.5) {$\mathbf\times$};
      \node[scale=1.5,color=blue] at (5.5,1.5) {$\mathbf\times$};
       \node[scale=1.5,fill=white] at (1.5,2.5) {$ $};
       \node[scale=1.5,fill=white] at (5.5,2.5) {$ $};
       \node[scale=1.5,fill=white] at (1.5,4.5) {$ $};
       \node[scale=1.5,fill=white] at (5.5,4.5) {$ $};
              \node[scale=1.5,fill=white] at (13.5,2.5) {$ $};
       \node[scale=1.5,fill=white] at (17.5,2.5) {$ $};
       \node[scale=1.5,fill=white] at (13.5,4.5) {$ $};
       \node[scale=1.5,fill=white] at (17.5,4.5) {$ $};
               \foreach \j in {0,...,2} { \foreach \i in{0,...,2} {
        \node[scale=1.5,fill=red,,opacity=0.3] at (\j+0.5,\i+0.5) {$ $};
        \node[scale=1.5,fill=red,,opacity=0.3] at (\j+12.5,\i+0.5) {$ $};
        }
        }
           \foreach \j in {3,...,5} { \foreach \i in{3,...,5} {
        \node[scale=1.5,fill=red,,opacity=0.3] at (\j+0.5,\i+0.5) {$ $};
        \node[scale=1.5,fill=red,,opacity=0.3] at (\j+12.5,\i+0.5) {$ $};
        }
        }
        \foreach \j in {0,...,5} {
             \node[scale=1.5,fill=green,,opacity=0.3] at (1.5,\j+0.5) {$ $};
             \node[scale=1.5,fill=green,,opacity=0.3] at (13.5,\j+0.5) {$ $};
        }
        \foreach \i in {0,...,5} {
             \node[scale=1.5,fill=green,,opacity=0.3] at (\i+0.5,4.5) {$ $};
             \node[scale=1.5,fill=green,,opacity=0.3] at (\i+12.5,4.5) {$ $};
        }
              \draw[step=1.0,black, thin] (0,0) grid (6,6);
              \draw[step=1.0,black, thin] (12,0) grid (18,6);
   \node at (3,-1){$T$} ;      
   \node at (15,-1){$T'=T\cup\{(i_c,j_c)\}$} ;
          \end{tikzpicture}}
 We apply the transformation $(f_1,g_1)$ in order to obtain the tableaux $T_2$ and $T'_2$ 
 
 
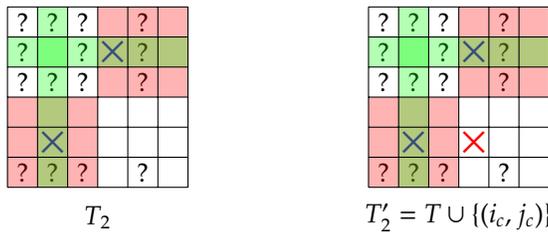
\centerline{
  \begin{tikzpicture}[scale=0.4]
          ;
              \foreach \j in {0,...,5} { \foreach \i in{0,...,5} {
        \node at (\j+0.5,\i+0.5) {$?$};
        \node at (\j+12.5,\i+0.5) {$?$};
        }
        }
        \foreach \j in {0,...,5} {
        \node[scale=1.5,fill=white] at (5.5,\j+0.5) {$ $};
        \node[scale=1.5,fill=white] at (17.5,\j+0.5) {$ $};
        }
        \foreach \i in {0,...,5} {
        \node[scale=1.5,fill=white] at (\i+0.5,2.5) {$ $};
        \node[scale=1.5,fill=white] at (\i+12.5,2.5) {$ $};
        }
       \foreach\i in{0,...,5}{
       \node[scale=1.5,fill=white] at (3.5,\i+0.5) {$ $};
       \node[scale=1.5,fill=white] at (15.5,\i+0.5) {$ $};
      \node[scale=1.5,fill=white] at (\i+12.5,1.5) {$ $};
      \node[scale=1.5,fill=white] at (\i+0.5,1.5) {$ $};
      }
      \node[scale=1.5,color=red] at (15.5,1.5) {$\mathbf\times$};
      \node[scale=1.5,color=blue] at (15.5,4.5) {$\mathbf\times$};
      \node[scale=1.5,color=blue] at (3.5,4.5) {$\mathbf\times$};
      \node[scale=1.5,color=blue] at (13.5,1.5) {$\mathbf\times$};
      \node[scale=1.5,color=blue] at (1.5,1.5) {$\mathbf\times$};
       \node[scale=1.5,fill=white] at (1.5,2.5) {$ $};
       \node[scale=1.5,fill=white] at (5.5,2.5) {$ $};
       \node[scale=1.5,fill=white] at (1.5,4.5) {$ $};
       \node[scale=1.5,fill=white] at (5.5,4.5) {$ $};
              \node[scale=1.5,fill=white] at (13.5,2.5) {$ $};
       \node[scale=1.5,fill=white] at (17.5,2.5) {$ $};
       \node[scale=1.5,fill=white] at (13.5,4.5) {$ $};
       \node[scale=1.5,fill=white] at (17.5,4.5) {$ $};
               \foreach \j in {0,...,2} { \foreach \i in{0,...,2} {
        \node[scale=1.5,fill=red,,opacity=0.3] at (\j+0.5,\i+0.5) {$ $};
        \node[scale=1.5,fill=red,,opacity=0.3] at (\j+12.5,\i+0.5) {$ $};
        }
        }
           \foreach \j in {3,...,5} { \foreach \i in{3,...,5} {
        \node[scale=1.5,fill=red,,opacity=0.3] at (\j+0.5,\i+0.5) {$ $};
        \node[scale=1.5,fill=red,,opacity=0.3] at (\j+12.5,\i+0.5) {$ $};
        }
        }
        \foreach \j in {0,...,5} {
             \node[scale=1.5,fill=green,,opacity=0.3] at (1.5,\j+0.5) {$ $};
             \node[scale=1.5,fill=green,,opacity=0.3] at (13.5,\j+0.5) {$ $};
        }
        \foreach \i in {0,...,5} {
             \node[scale=1.5,fill=green,,opacity=0.3] at (\i+0.5,4.5) {$ $};
             \node[scale=1.5,fill=green,,opacity=0.3] at (\i+12.5,4.5) {$ $};
        }
              \draw[step=1.0,black, thin] (0,0) grid (6,6);
              \draw[step=1.0,black, thin] (12,0) grid (18,6);
   \node at (3,-1){$T_2$} ;      
   \node at (15,-1){$T'_2=T\cup\{(i_c,j_c)\}$} ;
          \end{tikzpicture}}
 
 Applying $(\Id,g_2)$, we obtain
 
 
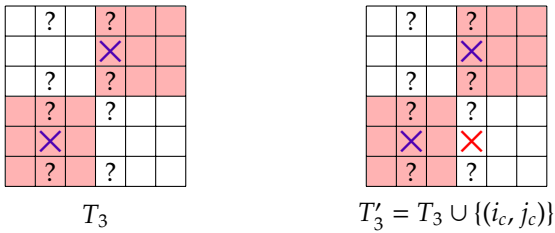
\centerline{
  \begin{tikzpicture}[scale=0.4]
          ;
      \node[scale=1.5,color=red] at (15.5,1.5) {$\mathbf\times$};
      \node[scale=1.5,color=blue] at (15.5,4.5) {$\mathbf\times$};
      \node[scale=1.5,color=blue] at (3.5,4.5) {$\mathbf\times$};
      \node[scale=1.5,color=blue] at (13.5,1.5) {$\mathbf\times$};
      \node[scale=1.5,color=blue] at (1.5,1.5) {$\mathbf\times$};
       \node[scale=1.5,fill=white] at (1.5,2.5) {$ $};
       \node[scale=1.5,fill=white] at (5.5,2.5) {$ $};
       \node[scale=1.5,fill=white] at (1.5,4.5) {$ $};
       \node[scale=1.5,fill=white] at (5.5,4.5) {$ $};
              \node[scale=1.5,fill=white] at (13.5,2.5) {$ $};
       \node[scale=1.5,fill=white] at (17.5,2.5) {$ $};
       \node[scale=1.5,fill=white] at (13.5,4.5) {$ $};
       \node[scale=1.5,fill=white] at (17.5,4.5) {$ $};
               \foreach \j in {0,...,2} { \foreach \i in{0,...,2} {
        \node[scale=1.5,fill=red,,opacity=0.3] at (\j+0.5,\i+0.5) {$ $};
        \node[scale=1.5,fill=red,,opacity=0.3] at (\j+12.5,\i+0.5) {$ $};
        }
        }
           \foreach \j in {3,...,5} { \foreach \i in{3,...,5} {
        \node[scale=1.5,fill=red,,opacity=0.3] at (\j+0.5,\i+0.5) {$ $};
        \node[scale=1.5,fill=red,,opacity=0.3] at (\j+12.5,\i+0.5) {$ $};
        }
        }
              \draw[step=1.0,black, thin] (0,0) grid (6,6);
              \draw[step=1.0,black, thin] (12,0) grid (18,6);
   \node at (3,-1){$T_3$} ;      
   \node at (15,-1){$T'_3=T_3\cup\{(i_c,j_c)\}$} ;
   \node at (1.5,0.5) {$?$};\node at (1.5,2.5) {$?$};\node at (1.5,3.5) {$?$};\node at (1.5,5.5) {$?$};
   \node at (13.5,0.5) {$?$};\node at (13.5,2.5) {$?$};\node at (13.5,3.5) {$?$};\node at (13.5,5.5) {$?$};
   \node at (3.5,0.5) {$?$};\node at (3.5,2.5) {$?$}; \node at (3.5,3.5) {$?$};\node at (3.5,5.5) {$?$};
   \node at (15.5,0.5) {$?$};\node at (15.5,2.5) {$?$}; \node at (15.5,3.5) {$?$};\node at (15.5,5.5) {$?$};
          \end{tikzpicture}}
 
 Finally, applying $(f_2,\Id)$, one obtains the tableaux
 
  \centerline{
  \begin{tikzpicture}[scale=0.4]
          ;
      \node[scale=1.5,color=red] at (15.5,1.5) {$\mathbf\times$};
      \node[scale=1.5,color=blue] at (15.5,4.5) {$\mathbf\times$};
      \node[scale=1.5,color=blue] at (3.5,4.5) {$\mathbf\times$};
      \node[scale=1.5,color=blue] at (13.5,1.5) {$\mathbf\times$};
      \node[scale=1.5,color=blue] at (1.5,1.5) {$\mathbf\times$};
       \node[scale=1.5,fill=white] at (1.5,2.5) {$ $};
       \node[scale=1.5,fill=white] at (5.5,2.5) {$ $};
       \node[scale=1.5,fill=white] at (1.5,4.5) {$ $};
       \node[scale=1.5,fill=white] at (5.5,4.5) {$ $};
              \node[scale=1.5,fill=white] at (13.5,2.5) {$ $};
       \node[scale=1.5,fill=white] at (17.5,2.5) {$ $};
       \node[scale=1.5,fill=white] at (13.5,4.5) {$ $};
       \node[scale=1.5,fill=white] at (17.5,4.5) {$ $};
               \foreach \j in {0,...,2} { \foreach \i in{0,...,2} {
        \node[scale=1.5,fill=red,,opacity=0.3] at (\j+0.5,\i+0.5) {$ $};
        \node[scale=1.5,fill=red,,opacity=0.3] at (\j+12.5,\i+0.5) {$ $};
        }
        }
           \foreach \j in {3,...,5} { \foreach \i in{3,...,5} {
        \node[scale=1.5,fill=red,,opacity=0.3] at (\j+0.5,\i+0.5) {$ $};
        \node[scale=1.5,fill=red,,opacity=0.3] at (\j+12.5,\i+0.5) {$ $};
        }
        }
              \draw[step=1.0,black, thin] (0,0) grid (6,6);
              \draw[step=1.0,black, thin] (12,0) grid (18,6);
   \node at (3,-1){$T_4$} ;      
   \node at (15,-1){$T'_4=T_4\cup\{(i_c,j_c)\}$} ;
          \end{tikzpicture}}
          
which are not Nerode equivalent. Indeed, a permutation of the columns sends these tableaux respectively to

  \centerline{
  \begin{tikzpicture}[scale=0.4]
          ;
      \node[scale=1.5,color=red] at (13.5,1.5) {$\mathbf\times$};
      \node[scale=1.5,color=blue] at (13.5,4.5) {$\mathbf\times$};
      \node[scale=1.5,color=blue] at (1.5,4.5) {$\mathbf\times$};
      \node[scale=1.5,color=blue] at (15.5,1.5) {$\mathbf\times$};
      \node[scale=1.5,color=blue] at (3.5,1.5) {$\mathbf\times$};
      \node[scale=1.5,color=blue] at (1.5,4.5) {$\mathbf\times$};
      \node[scale=1.5,color=blue] at (13.5,4.5) {$\mathbf\times$};
               \foreach \j in {0,...,2} { \foreach \i in{0,...,2} {
        \node[scale=1.5,fill=red,,opacity=0.3] at (\j+0.5,\i+0.5) {$ $};
        \node[scale=1.5,fill=red,,opacity=0.3] at (\j+12.5,\i+0.5) {$ $};
        }
        }
           \foreach \j in {3,...,5} { \foreach \i in{3,...,5} {
        \node[scale=1.5,fill=red,,opacity=0.3] at (\j+0.5,\i+0.5) {$ $};
        \node[scale=1.5,fill=red,,opacity=0.3] at (\j+12.5,\i+0.5) {$ $};
        }
        }
              \draw[step=1.0,black, thin] (0,0) grid (6,6);
              \draw[step=1.0,black, thin] (12,0) grid (18,6);
   \node at (3,-1){$T_5$} ;      
   \node at (15,-1){$T'_5$} ;
          \end{tikzpicture}}
The tableau $T_5$ is non final while the tableau $T'_5$ is final.
\end{example}
\begin{proposition}\label{prop-sateqlocsat}
  A non-empty tableau  $T$ is saturated if and only if it is locally saturated. Furthermore the locally saturated empty tableau  is saturated if and only if $(0,0)$ is not in the final zone.
\end{proposition}
\begin{proof}
 Let $T$ be a non-empty locally saturated tableau and suppose that $T\neq \Sat(T)$. From Lemma \ref{lem-interval}, $T$ is Nerode equivalent to any $T'$ such that $T\subset T'\subset \Sat(T)$. In particular, if $(i,j)\in \Sat(T)\setminus T$ then $T$ is Nerode equivalent to $T\cup\{(i,j)\}$. But this contradicts Lemma \ref{lem-T2Tgen}. Hence, $T=\Sat(T)$.\\
 By construction, if $(0,0)\in F$ the tableau $\emptyset$ is Nerode equivalent to $\{(0,0)\}$ so it is not saturated. Conversely, if $(0,0)\not \in F$ then the tableau $\emptyset$ is the only final state with no cross in the final zone. So the empty tableau is the only one in its Nerode class and it is saturated.
\end{proof}
\section{Monster witnesses\label{sec-monster}}
Let us denote by $\SA_\bullet^{F_1,F_2}$ the set of the accessible saturated tableaux. From Section \ref{section-saturation}, the set of saturated tableaux is always a representative set of the Nerode classes of $\mathrm{M}_\bullet^{F_1,F_2}$.
Hence,  the tableaux of $\SA_\bullet^{F_1,F_2}$ are in one to one correspondence with the states of $\Min(\mathrm M_\bullet^{F_1,F_2})$ and  this implies 
\[
\sc_{\ostar}(m,n)=\max\{\#\SA_\bullet^{F_1,F_2}\mid F_1\subset \IntEnt m,\ F_2\subset\IntEnt n\}.
\]
To accurately calculate the state complexity, we have to compute the value of  $\#\SA_\bullet^{F_1,F_2}$. This value depends both on the number of elements in $F_1$
 and $F_2$ and of some properties of the final zone as shown in Table \ref{tab-SA}.
 \begin{table}[H]
 \[
     \begin{array}{|c|c|c|}\hline
          & F\subset (\{0\}\times\IntEnt n)\cup (\IntEnt m\times \{0\}) & F\not\subset (\{0\}\times\IntEnt n)\cup (\IntEnt m\times \{0\})\\\hline
       (0,0)\in F   & \#\SA_{\bullet}^{F_1,F_2}\leq\#\SV_{\bullet}^{F_1,F_2}=\#\LSV_{\bullet}^{F_1,F_2}-1& \#\SA_{\bullet}^{F_1,F_2}=\#\SV_{\bullet}^{F_1,F_2}=\#\LSV_{\bullet}^{F_1,F_2}-1\\\hline
       (0,0)\not\in F& \#\SA_{\bullet}^{F_1,F_2}\leq\#\SV_{\bullet}^{F_1,F_2}=\#\LSV_{\bullet}^{F_1,F_2}& \#\SA_{\bullet}^{F_1,F_2}=\#\SV_{\bullet}^{F_1,F_2}=\#\LSV_{\bullet}^{F_1,F_2}\\\hline
     \end{array}\]
     \caption{Relations between the number of accessible saturated states $\#\SA_\bullet^{F_1,F_2}$, saturated valid states  $\#\SV_\bullet^{F_1,F_2}$, and local-saturated valid states $\#\LSV_{\bullet}^{F_1,F_2}$.\label{tab-SA}}
     \end{table}
    The condition on the columns of Table \ref{tab-SA} means that the final zone is either contained in the union of row $0$ and column $0$ or not. 
    By construction, this condition does not affect the number of  elements of $\LSV_\bullet^{F_1,F_2}$,
thus nor the number of elements of  $\#\SV^{F_1,F_2}_\bullet$.
From Table \ref{tab-SA}, we deduce that to maximize $\#\SA_\bullet^{F_1,F_2}$, it is sufficient to maximize $\#\SV_\bullet^{F_1,F_2}$ when $ F\not\subset (\{0\}\times\IntEnt n)\cup (\IntEnt m\times \{0\})$. In other words,

\begin{theorem}\label{th-SV2sc}
We have
 \[
\begin{array}{ll}\sc_{\ostar}(m,n)=\max\left\{\#\SV_\bullet^{F_1,F_2}\right|&F_1\subsetneq \IntEnt m,\ F_2\subsetneq\IntEnt n,\ F_1, F_2\neq\emptyset,\\ &\hspace{-1cm}\left.{\color{white}\mathrm{V}^{F_1,F_2}}(F_1\times\IntEnt n)\bullet (\IntEnt m\times F_2)\not\subset(\{0\}\times\IntEnt n)\cup (\IntEnt m\times \{0\})\right\}.
\end{array}
\]
\end{theorem}
 
As a preliminary result, we give an expression of each  $\#\SV^{F_1,F_2}_\bullet$ depending on the fact that the cell $(0,0)$ belongs or not to the final zone. According to Table \ref{op-bool}, without loss of generality, we consider  $\bullet\in\{\cap,\cup,\oplus\}$, the other cases being recovered by replacing $F_1$ or $F_2$ by its complementary.
\subsection{Counting saturated valid tableaux}
\subsubsection{Type A $(\bullet=\cap$)} From Proposition \ref{prop-sateqlocsat} and Section \ref{2x2case}, any valid tableau is saturated. So  valid tableaux are the only ones we have to count. We consider two cases.
\begin{enumerate}
    \item Suppose $(0,0)\in F$. There are $2^{mn-1}$ local saturated valid tableaux containing $(0,0)$ because any tableau containing $(0,0)$ is valid (see Fig. \ref{fig-valcap00} for an illustration).
     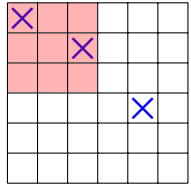
\begin{figure}[H]
      \centerline{
  \begin{tikzpicture}[scale=0.4]
          ;
  =
      \node[scale=1.5,color=blue] at (0.5,5.5) {$\mathbf\times$};
      \node[scale=1.5,color=blue] at (2.5,4.5) {$\mathbf\times$};
      \node[scale=1.5,color=blue] at (4.5,2.5) {$\mathbf\times$};
               \foreach \j in {0,...,2} { \foreach \i in{0,...,2} {
        \node[scale=1.5,fill=red,,opacity=0.3] at (\j+0.5,\i+3.5) {$ $};
        }
        }
              \draw[step=1.0,black, thin] (0,0) grid (6,6);
          \end{tikzpicture}}
          \caption{Example of valid tableau containing $(0,0)$ for the type A\label{fig-valcap00}.}
          \end{figure}
    
    Furthermore, the set of the valid tableaux that do not contain $(0,0)$ is the set of  tableaux having no cross in $F_1\times F_2$ (see an example in Fig. \ref{fig-valcapnot00}). There are $2^{mn-\#F_1\#F_2}$ such tableaux.
     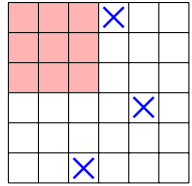
\begin{figure}[H]
      \centerline{
  \begin{tikzpicture}[scale=0.4]
          ;
  =
      \node[scale=1.5,color=blue] at (3.5,5.5) {$\mathbf\times$};
      \node[scale=1.5,color=blue] at (2.5,0.5) {$\mathbf\times$};
      \node[scale=1.5,color=blue] at (4.5,2.5) {$\mathbf\times$};
               \foreach \j in {0,...,2} { \foreach \i in{0,...,2} {
        \node[scale=1.5,fill=red,,opacity=0.3] at (\j+0.5,\i+3.5) {$ $};
        }
        }
              \draw[step=1.0,black, thin] (0,0) grid (6,6);
          \end{tikzpicture}}
          \caption{Example of valid tableau that does not contain $(0,0)$ for the type A when $(0,0)\in F$\label{fig-valcapnot00}}
          \end{figure}

    We deduce
    \begin{equation}\label{SVcap(00)}
        \#\LSV^{F_1,F_2}_\cap=2^{mn-1}+2^{mn-\#F_1\#F_2}\mbox{ and }\#\SV^{F_1,F_2}_\cap=2^{mn-1}+2^{mn-\#F_1\#F_2}-1.
    \end{equation}
   \item Suppose $(0,0)\not\in F$. For the same reason as in the previous case, there  are still $2^{mn-1}$ valid tableaux having a cross at $(0,0)$. But there are only  $2^{nm-\#F_1\#F_2-1}$ valid tableaux with no cross at $(0,0)$ (see an example in Fig. \ref{fig-valcapnot002})
     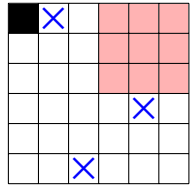
\begin{figure}[H]
      \centerline{
  \begin{tikzpicture}[scale=0.4]
          ;
  =
      \node[scale=1.5,color=blue] at (1.5,5.5) {$\mathbf\times$};
      \node[scale=1.5,color=blue] at (2.5,0.5) {$\mathbf\times$};
      \node[scale=1.5,color=blue] at (4.5,2.5) {$\mathbf\times$};
               \foreach \j in {0,...,2} { \foreach \i in{0,...,2} {
        \node[scale=1.5,fill=red,,opacity=0.3] at (\j+3.5,\i+3.5) {$ $};
        }
        }
       \node[scale=1.5,fill=black] at (0.5,5.5) {$ $};
              \draw[step=1.0,black, thin] (0,0) grid (6,6);
          \end{tikzpicture}}
          \caption{Example of valid tableau that does not contain $(0,0)$ for the type A when $(0,0)\not\in F$\label{fig-valcapnot002}.}
          \end{figure}
          We deduce that in this case we have
           \begin{equation}\label{SVcapnot(00)}
        \#\SV^{F_1,F_2}_\cap=\#\LSV^{F_1,F_2}_\cap=2^{mn-1}+2^{mn-\#F_1\#F_2-1}.
    \end{equation}
\end{enumerate}
\subsubsection{Type O ($\bullet=\cup$)}
 According to  Proposition \ref{prop-sateqlocsat} and Section \ref{2x2case},  for any non-empty saturated tableau  $T$,  if $(i,j)$ and $(i',j')$ belong to$T$ then $(i',j)$ and $(i,j')$ belong to $T$. In another words, for any saturated tableau $T$, there exist $A\subset\IntEnt m$ and $B\subset\IntEnt n$ such that $T=A\times B$. So we have to enumerate the pairs $(A,B)$ such that if $(A\times B)\cap F\neq\emptyset$ then $0\in A$ and $0\in B$.
We consider two cases
\begin{enumerate}
    \item Suppose $(0,0)\in F$. The configuration of the final zone is illustrated in Fig. \ref{fig-valcup00}.  Remark first that any tableau $T=A\times B$ containing  $(0,0)$ is valid. There are $2^{m+n-2}$ such tableaux. If $(0,0)\not\in A\times B$ then the condition of validity implies $A\subset \IntEnt m\setminus F_1$ and $B\subset \IntEnt n\setminus F_2$. Hence, there are $(2^{m-\#F_1}-1)(2^{n-\#F_2}-1)+1$ such tableaux.
    \begin{figure}[H]
      \centerline{
  \begin{tikzpicture}[scale=0.4]
          ;
  =
               \foreach \j in {0,...,5} { \foreach \i in{0,...,2} {
        \node[scale=1.5,fill=red,,opacity=0.3] at (\j+0.5,\i+3.5) {$ $};
        }
        }
\foreach \j in {0,...,2} { \foreach \i in{0,...,2} {
        \node[scale=1.5,fill=red,,opacity=0.3] at (\j+0.5,\i+0.5) {$ $};
        }
        }
              \draw[step=1.0,black, thin] (0,0) grid (6,6);     \end{tikzpicture}}
          \caption{Final zone for the type O when $(0,0)\in F$.\label{fig-valcup00}}
          \end{figure}
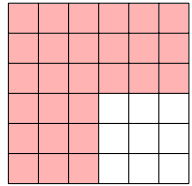
          Hence, we obtain
\[     \begin{array}{rcl}   \#\LSV^{F_1,F_2}_\cup&=&2^{m+n-2}+(2^{m-\#F_1}-1)(2^{n-\#F_2}-1)+1\\&=&2^{m+n-2}+2^{m+n-\#F_1-\#F_2}-2^{m-\#F_1}-2^{n-\#F_2}+2,\end{array}
 \]
    and then\begin{equation}\label{SVcup(00)}
        \#\SV^{F_1,F_2}_\cup=2^{m+n-2}+2^{m+n-\#F_1-\#F_2}-2^{m-\#F_1}-2^{n-\#F_2}+1.
    \end{equation}
    \item Suppose $(0,0)\not\in F$. The configuration of the final zone is illustrated in Fig. \ref{fig-valcupnot00}.
        \begin{figure}[h]
      \centerline{
  \begin{tikzpicture}[scale=0.4]
          ;
  =
               \foreach \j in {0,...,5} { \foreach \i in{0,...,2} {
        \node[scale=1.5,fill=red,,opacity=0.3] at (\j+0.5,\i+0.5) {$ $};
        }
        }
\foreach \j in {0,...,2} { \foreach \i in{0,...,2} {
        \node[scale=1.5,fill=red,,opacity=0.3] at (\j+3.5,\i+3.5) {$ $};
        }
        }
              \draw[step=1.0,black, thin] (0,0) grid (6,6);     \end{tikzpicture}}
          \caption{Final zone for the type O when $(0,0)\not\in F$\label{fig-valcupnot00}}
          \end{figure}
          The same reasoning as in the previous case shows that there are $2^{m+n-2}$ saturated valid tableaux having a cross at $(0,0)$.
          In order to count the tableaux that do not contain $(0,0)$, we apply the inclusion-exclusion principle. 
          Indeed, we sum the number of locally saturated tableaux which are subsets of  $(\IntEnt m\setminus \{0\}) \times \IntEnt n$, the number of locally saturated tableaux which are  subsets of  $\IntEnt m\times(\IntEnt n\setminus\{0\})$, and we substract to the number of locally saturated tableaux which are subset of $(\IntEnt m\setminus 0)\times(\IntEnt n\setminus\{0\})$ because they were counted twice.
          So we obtain $\left((2^{m-\#F_1-1}-1)((2^{n-\#F_2}-1)+1\right) + \left( (2^{m-\#F_1}-1)((2^{n-\#F_2-1}-1)+1\right) -((2^{m-\#F_1-1}-1)((2^{n-\#F_2-1}-1)+1)=\frac342^{m+n-\#F_1-\#F_2}-2^{m-\#F_1}-2^{m-\#F_2}+2
         $ 
          such tableaux. Hence,
            \begin{equation}\label{SVcupnot(00)}
        \#\SV^{F_1,F_2}_\cup=\#\LSV^{F_1,F_2}_\cup=2^{m+n-2}+\frac342^{m+n-\#F_1-\#F_2}-2^{m-\#F_1}-2^{n-\#F_2}+2.
    \end{equation}
\end{enumerate}
\subsubsection{Type X ($\bullet=\oplus$)}
 In this case the valid saturated tableaux are more complicated to enumerate. Indeed, from Proposition \ref{prop-sateqlocsat} and Section \ref{2x2case}, the locally saturated tableaux are the tableaux that avoid one of the following $2\times 2$ motives:
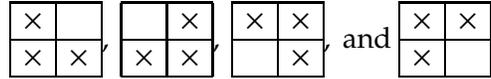
\[
\begin{array}{|c|c|}
\hline \times&\\\hline\times&\times\\\hline
\end{array},\ \begin{array}{|c|c|}
\hline &\times\\\hline\times&\times\\\hline
\end{array},\ \begin{array}{|c|c|}
\hline \times&\times\\\hline&\times\\\hline
\end{array},\mbox{ and }
\begin{array}{|c|c|}
\hline \times&\times\\\hline\times&\\\hline
\end{array}.
\]
Let us denote by $\alpha_{m,n}$ the number of $m\times n$ locally saturated tableaux and by $\alpha'_{m,n}$ the number of $m\times n$ locally saturated tableaux having a cross at $(0,0)$. In a previous paper \cite{CLMP15}, we give formulas allowing to compute these values. Nevertheless, the precise knowledge of these formulas is not useful for our purpose, we do not recall it here but we invite the reader interested by these enumeration results to read this paper.\\
Again, we have to consider two cases :
\begin{enumerate}
    \item Suppose $(0,0)\in F$. The configuration of the final zone is illustrated in Fig. \ref{fig-valxor00}.
     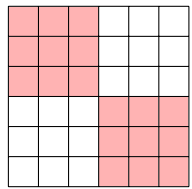
\begin{figure}[H]
      \centerline{
  \begin{tikzpicture}[scale=0.4]
          ;
  =
               \foreach \j in {3,...,5} { \foreach \i in{0,...,2} {
        \node[scale=1.5,fill=red,,opacity=0.3] at (\j+0.5,\i+0.5) {$ $};
        }
        }
\foreach \j in {0,...,2} { \foreach \i in{3,...,5} {
        \node[scale=1.5,fill=red,,opacity=0.3] at (\j+0.5,\i+0.5) {$ $};
        }
        }
              \draw[step=1.0,black, thin] (0,0) grid (6,6);     \end{tikzpicture}}
          \caption{Final zone for the type X when $(0,0)\in F$\label{fig-valxor00}.}
          \end{figure}
    As any locally saturated tableaux having a cross at $(0,0)$ is valid,  the number of such tableaux equals to $\alpha'_{m,n}$. Notice also that a locally saturated tableau $T$ having no cross at $(0,0)$ is valid if and only if T has crosses only in the non-final zone, i.e. in the two rectangles drawn in white in Figure \ref{fig-valxor00}.
    This gives $\alpha_{\#F_1,\#F_2}\alpha_{m-\#F_1,n-\#F_2}$ locally saturated valid tableau with no cross at $(0,0)$ and then
     \begin{equation}\label{SVxor(00)}
        \#\LSV^{F_1,F_2}_\oplus=\alpha'_{m,n}+\alpha_{\#F_1,\#F_2}\alpha_{m-\#F_1,n-\#F_2}\mbox{ and } \#\SV^{F_1,F_2}_\oplus=\alpha'_{m,n}+\alpha_{\#F_1,\#F_2}\alpha_{m-\#F_1,n-\#F_2}-1.
    \end{equation}
    \item Suppose $(0,0)\not\in F$. The configuration of the final zone is illustrated in Fig. \ref{fig-valxornot00}.
    \begin{figure}[H]
      \centerline{
  \begin{tikzpicture}[scale=0.4]
          ;
  =
               \foreach \j in {0,...,2} { \foreach \i in{0,...,2} {
        \node[scale=1.5,fill=red,,opacity=0.3] at (\j+0.5,\i+0.5) {$ $};
        }
        }
\foreach \j in {3,...,5} { \foreach \i in{3,...,5} {
        \node[scale=1.5,fill=red,,opacity=0.3] at (\j+0.5,\i+0.5) {$ $};
        }
        }
              \draw[step=1.0,black, thin] (0,0) grid (6,6);     \end{tikzpicture}}
          \caption{Final zone for the type X when $(0,0)\not\in F$.\label{fig-valxornot00}}
          \end{figure}
As in the previous case, the number of tableaux with a cross in $(0,0)$ is $\alpha'_{m,n}$. For the tableaux not containing a cross in $(0,0)$, we have to consider tableaux in the non-final zone containing $(0,0)$ ($\alpha_{\#F_1,\#F_2}  -\alpha'_{\#F_1,\#F_2}$ or $\alpha_{m-\#F_1,n-\#F_2}  -\alpha'_{m-\#F_1,n-\#F_2}$  depending on the fact that $(0,0)$ is in $F_1\times F_2$ or in $\overline{F_1}\times \overline{F_2}$) and  tableaux in the non-final zone not containing $(0,0)$, which gives us           
          
    \begin{equation}
        \label{SVxornot(00)}
        \#\SV^{F_1,F_2}_\oplus=\#\LSV^{F_1,F_2}_\oplus=\left\{\begin{array}{ll}\alpha'_{m,n}+(\alpha_{\#F_1,\#F_2}-\alpha'_{\#F_1,\#F_2})\alpha_{m-\#F_1, n-\#F_2}&\mbox{ if }(0,0)\in F_1\times F_2\\
        \alpha'_{m,n}+\alpha_{m, n}(\alpha_{m-\#F_1,n-\#F_2}-\alpha'_{m-\#F_1,n-\#F_2})&\mbox{ otherwise. }
        \end{array}\right.
    \end{equation}
\end{enumerate}
\subsection{Computing witnesses}
In this section, we compute  a final zone allowing us to obtain the tight bound for the complexity of each studied operation. This computation allows us to recover the tight bound for the star of intersection due to Jir\'askov\'a and Okhotin \cite{JO11} and that of the  star of union  due to Salomaa \textit{et al.} \cite{SSY07}. This also allows us to give an  expression of  the tight bound for the star of xor.
\subsubsection{Type A}
 Obviously, from (\ref{SVcap(00)}) and (\ref{SVcapnot(00)}) the maximal value of $\#\LSV_\cap^{F_1,F_2}$ is reached when $\#F_1=\#F_2=1$ in both cases. 
Notice that if $(0,0)\in F$ then $\#F_1=\#F_2=1$ implies $F=\{(0,0)\}$. But, in this case a tableau is accessible if and only if it contains at most one cross. The states $\emptyset$ and $\{(0,0)\}$ being Nerode equivalent, we obtain
\begin{equation}\label{cap00FF=1}
\#\SA_\cap^{\{0\},\{0\}}=mn. \end{equation}
If $(0,0)\in F$ and ($\#F_1>1$ or $\#F_2>1$) then, by Table \ref{tab-SA} and (\ref{SVcap(00)}), we have
\begin{equation}\label{capg00FF>1}
\#\SA^{F_1,F_2}_\cap\leq \#\LSV_\cap^{F_1,F_2}-1\leq \#\LSV_\cap^{\{0\},\{0\}}-1=2^{mn-1}+2^{mn-2}-1=\frac342^{mn}-1.
\end{equation}
Suppose now $(0,0)\not\in F$, from Table \ref{tab-SA} and the discussion above, the maximal values for $\#\SA^{F_1,F_2}_{\cap}$ is reached when $F_1=\{f_1\}$ and $F_2=\{f_2\}$ with $f_1\in\IntEnt m\setminus\{0\}$ and $f_2\in\IntEnt n\setminus\{0\}$. In this case we have
\begin{equation}\label{capgnot00FF>1}
\#\SA^{\{f_1\},\{f_2\}}_\cap= \#\LSV_\cap^{\{f_1\},\{f_2\}}=2^{mn-1}+2^{mn-2}=\frac342^{mn}.
\end{equation}
We summarize the results contained in (\ref{cap00FF=1}), (\ref{capg00FF>1}) and (\ref{capgnot00FF>1}) in the following theorem:
\begin{theorem}[Jir\'askov\'a and Okhotin \cite{JO11}]\label{th-scA}
  When $m,n>1$ we have
  \[
  \sc_{\ostar}(m,n)=\frac342^{mn}
  \]
 and $2$-monster $\mon_{m,n}^{\{f_1\},\{f_2\}}$  with $f_1\in\IntEnt m\setminus\{0\}$ and $f_2\in\IntEnt n\setminus\{0\}$ is a witness.
\end{theorem}
\begin{example}
We list in Table \ref{table-scAvalues} the first values of the state complexity. The valid saturated tableaux illustrating the case $m=n=2$ are pictured in Figure \ref{fig-exA}.
  \begin{table}[H]
     \[\begin{array}{|c|cccccc|}
     \hline
     m\setminus n& 2&3&4&5&6&7\\\hline
     2&12& 48& 192& 768& 3072& 12288\\
     3&48& 384& 3072& 24576& 196608& 1572864\\
     4&192& 3072& 49152& 786432& 12582912& 201326592\\
     5&768& 24576& 786432& 25165824& 805306368& 25769803776\\
     6& 3072& 196608& 12582912& 805306368& 51539607552& 3298534883328\\
     7&12288& 1572864& 201326592& 25769803776& 3298534883328&
    422212465065984\\
     8&49152& 12582912& 3221225472&
    824633720832& 211106232532992& 54043195528445952\\\hline
     \end{array}
     \]
     \caption{First values of $\sc_{\ostar}(m,n)$ for the type A.\label{table-scAvalues}}
     \end{table}
     \begin{figure}[h] \centerline{
     \begin{tikzpicture}[scale=0.4]
\node[fill=red,opacity=0.3,scale=1.5] at (1.5,0.5) {$ $}; 
    \draw[step=1.0,black, thin] (0,0) grid (2,2); 
     \end{tikzpicture}
      \begin{tikzpicture}[scale=0.4]
\node[fill=red,opacity=0.3,scale=1.5] at (1.5,0.5) {$ $}; 
\node at (0.5,1.5) {$\times$};
    \draw[step=1.0,black, thin] (0,0) grid (2,2); 
     \end{tikzpicture}
     \begin{tikzpicture}[scale=0.4]
\node[fill=red,opacity=0.3,scale=1.5] at (1.5,0.5) {$ $}; 
\node at (1.5,1.5) {$\times$};
    \draw[step=1.0,black, thin] (0,0) grid (2,2); 
     \end{tikzpicture}
     \begin{tikzpicture}[scale=0.4]
\node[fill=red,opacity=0.3,scale=1.5] at (1.5,0.5) {$ $}; 
\node at (0.5,0.5) {$\times$};
    \draw[step=1.0,black, thin] (0,0) grid (2,2); 
     \end{tikzpicture}
     \begin{tikzpicture}[scale=0.4]
\node[fill=red,opacity=0.3,scale=1.5] at (1.5,0.5) {$ $}; 
\node at (0.5,0.5) {$\times$};
\node at (0.5,1.5) {$\times$};
    \draw[step=1.0,black, thin] (0,0) grid (2,2); 
     \end{tikzpicture}
      \begin{tikzpicture}[scale=0.4]
\node[fill=red,opacity=0.3,scale=1.5] at (1.5,0.5) {$ $}; 
\node at (0.5,0.5) {$\times$};
\node at (1.5,1.5) {$\times$};
    \draw[step=1.0,black, thin] (0,0) grid (2,2); 
     \end{tikzpicture}
      \begin{tikzpicture}[scale=0.4]
\node[fill=red,opacity=0.3,scale=1.5] at (1.5,0.5) {$ $}; 
\node at (1.5,0.5) {$\times$};
\node at (0.5,1.5) {$\times$};
    \draw[step=1.0,black, thin] (0,0) grid (2,2); 
     \end{tikzpicture}
 \begin{tikzpicture}[scale=0.4]
\node[fill=red,opacity=0.3,scale=1.5] at (1.5,0.5) {$ $}; 
\node at (1.5,1.5) {$\times$};
\node at (0.5,0.5) {$\times$};
    \draw[step=1.0,black, thin] (0,0) grid (2,2); 
     \end{tikzpicture}
      \begin{tikzpicture}[scale=0.4]
\node[fill=red,opacity=0.3,scale=1.5] at (1.5,0.5) {$ $}; 
\node at (1.5,1.5) {$\times$};
\node at (0.5,0.5) {$\times$};
\node at (0.5,1.5) {$\times$};
    \draw[step=1.0,black, thin] (0,0) grid (2,2); 
     \end{tikzpicture}
      \begin{tikzpicture}[scale=0.4]
\node[fill=red,opacity=0.3,scale=1.5] at (1.5,0.5) {$ $}; 
\node at (1.5,0.5) {$\times$};
\node at (0.5,0.5) {$\times$};
\node at (0.5,1.5) {$\times$};
    \draw[step=1.0,black, thin] (0,0) grid (2,2); 
     \end{tikzpicture}
      \begin{tikzpicture}[scale=0.4]
\node[fill=red,opacity=0.3,scale=1.5] at (1.5,0.5) {$ $}; 
\node at (1.5,1.5) {$\times$};
\node at (1.5,0.5) {$\times$};
\node at (0.5,1.5) {$\times$};
    \draw[step=1.0,black, thin] (0,0) grid (2,2); 
     \end{tikzpicture}
        \begin{tikzpicture}[scale=0.4]
\node[fill=red,opacity=0.3,scale=1.5] at (1.5,0.5) {$ $}; 
\node at (1.5,1.5) {$\times$};
\node at (1.5,0.5) {$\times$};
\node at (0.5,1.5) {$\times$};
\node at (0.5,0.5) {$\times$};
    \draw[step=1.0,black, thin] (0,0) grid (2,2); 
     \end{tikzpicture}}
     \caption{The $12$ saturated valid $2\times 2$-tableaux for the type A, $F_1=F_2=\{1\}$\label{fig-exA}}
     \end{figure}
     \end{example}
%

\subsubsection{Type O}
We assume that $m,n\geq 2$.
 From (\ref{SVcup(00)}) and (\ref{SVcupnot(00)}) the maximal value of $\#\LSV_\cup^{F_1,F_2}$ is reached when $\#F_1=\#F_2=1$ in both cases. Let $F_1=\{f_1\}$ and $F_2=\{f_2\}$ be such that $(0,0)$ is in the final zone of $\mathrm{M}_{\cup}^{F_1,F_2}$. From (\ref{SVcup(00)}) and (\ref{SVcupnot(00)}), if $F'_1$ and $F'_2$ are such that $\#F'_1=\#F_1$ and $\#F'_2=\#F_2$ and $(0,0)$ is not in the final zone of $\mathrm{M}_{\cup}^{F'_1,F'_2}$ then $\#\LSV_\cup^{F_1,F_2}>\#\LSV_\cup^{F'_1,F'_2}$. 
 So, as the final zone must not be completely included in the $0$-row and the $0$-column and as $(0,0)$ is in the final zone of $\mathrm{M}_{\cup}^{F_1,F_2}$  one of the two state $f_1$ or $f_2$ is $0$ while the other is not.  
 
Hence, from Table \ref{tab-SA}, (\ref{SVcup(00)}) and (\ref{SVcupnot(00)})  we have
 \[
 \#\SA_\cup^{\{0\},\{f_2\}}=\#\SA_\cup^{\{f_1\},\{0\}}=2^{m+n-1}-2^{m-1}-2^{n-1}+1.
 \]
 We summarize the results above in the following theorem:
\begin{theorem}[ Salomaa \textit{et al.} \cite{SSY07}]\label{th-sctypeO}
  When $m,n>1$ we have
  \[
  \sc_{\ostar}(m,n)=2^{m+n-1}-2^{m-1}-2^{n-1}+1,
  \]
 and $2$-monster $\mon_{m,n}^{\{f_1\},\{0\}}$  or $\mon_{m,n}^{\{0\},\{f_2\}}$,  with $f_1\in\IntEnt m\setminus\{0\}$ and $f_2\in\IntEnt n\setminus\{0\}$, is a witness.
\end{theorem}
\begin{example}
We list in Table \ref{table-scOvalues} the first values of the state complexity. The valid saturated tableaux illustrating the case $m=n=2$ are pictured in Figure \ref{fig-exO}.
  \begin{table}[H]
     \[\begin{array}{|c|ccccccc|}
     \hline
     m\setminus n& 2&3&4&5&6&7&8\\\hline
     2&5& 11& 23& 47& 95& 191& 383\\
     3&11& 25& 53& 109& 221& 445& 893\\
     4&23& 53& 113& 233& 473& 953& 1913\\
     5&47& 109& 233& 481& 977& 1969& 3953\\
     6& 95& 221& 473& 977& 1985& 4001& 8033\\
     7&191&445& 953& 1969& 4001& 8065& 16193\\
     8&383& 893& 1913& 3953& 8033& 16193& 32513\\\hline
     \end{array}
     \]
     \caption{First values of $\sc_{\ostar}(m,n)$ for the type O.\label{table-scOvalues}}
     \end{table}
     \begin{figure}[h] \centerline{
     \begin{tikzpicture}[scale=0.4]
\node[fill=red,opacity=0.3,scale=1.5] at (1.5,0.5) {$ $}; 
\node[fill=red,opacity=0.3,scale=1.5] at (0.5,1.5) {$ $};
\node[fill=red,opacity=0.3,scale=1.5] at (0.5,0.5) {$ $}; 
\node at (0.5,1.5) {$\times$};
    \draw[step=1.0,black, thin] (0,0) grid (2,2); 
     \end{tikzpicture}
     \begin{tikzpicture}[scale=0.4]
\node[fill=red,opacity=0.3,scale=1.5] at (1.5,0.5) {$ $}; 
\node[fill=red,opacity=0.3,scale=1.5] at (0.5,1.5) {$ $};
\node[fill=red,opacity=0.3,scale=1.5] at (0.5,0.5) {$ $}; 
\node at (1.5,1.5) {$\times$};
    \draw[step=1.0,black, thin] (0,0) grid (2,2); 
     \end{tikzpicture}
     \begin{tikzpicture}[scale=0.4]
\node[fill=red,opacity=0.3,scale=1.5] at (1.5,0.5) {$ $}; 
\node[fill=red,opacity=0.3,scale=1.5] at (0.5,1.5) {$ $};
\node[fill=red,opacity=0.3,scale=1.5] at (0.5,0.5) {$ $}; 
\node at (0.5,1.5) {$\times$};
\node at (0.5,0.5) {$\times$};
    \draw[step=1.0,black, thin] (0,0) grid (2,2); 
     \end{tikzpicture}
     \begin{tikzpicture}[scale=0.4]
\node[fill=red,opacity=0.3,scale=1.5] at (1.5,0.5) {$ $}; 
\node[fill=red,opacity=0.3,scale=1.5] at (0.5,1.5) {$ $};
\node[fill=red,opacity=0.3,scale=1.5] at (0.5,0.5) {$ $}; 
\node at (0.5,1.5) {$\times$};
\node at (1.5,1.5) {$\times$};
    \draw[step=1.0,black, thin] (0,0) grid (2,2); 
     \end{tikzpicture}
     \begin{tikzpicture}[scale=0.4]
\node[fill=red,opacity=0.3,scale=1.5] at (1.5,0.5) {$ $}; 
\node[fill=red,opacity=0.3,scale=1.5] at (0.5,1.5) {$ $};
\node[fill=red,opacity=0.3,scale=1.5] at (0.5,0.5) {$ $}; 
\node at (0.5,1.5) {$\times$};
\node at (0.5,0.5) {$\times$};
\node at (1.5,0.5) {$\times$};
\node at (1.5,1.5) {$\times$};
    \draw[step=1.0,black, thin] (0,0) grid (2,2); 
     \end{tikzpicture}
     }
     \caption{The $5$ saturated valid $2\times 2$-tableaux for the type O, $F_1=\{1\}$, $F_2=\{0\}$.\label{fig-exO}}
     \end{figure}
     \end{example}
\subsubsection{Type X}\label{sec-witnTypeX}
We assume that $m,n\geq 2$.
From (\ref{SVxor(00)}) and (\ref{SVxornot(00)}) the maximal value of $\#\LSV_\oplus^{F_1,F_2}$ is reached when $(0,0)$ is in the final zone. We now show that the state complexity of star of xor is reached when the size of the final zone is minimal, \textit{i.e.}  final sets are both minimal ($\#F_1=1$ and $\#F_2=1$) or both maximal ($\#F_1=m-1$ and $\#F_2=n-1$). Equality (\ref{SVxor(00)}) implies that we have to show that $\alpha_{p,q}\alpha_{m-p,n-q}\leq \alpha_{1,1}\alpha_{m-1,n-1}$ for any $1\leq p\leq m-1$ and $1\leq q\leq n-1$. 
Rather than proving this inequality directly, which is quite difficult, we will instead examine the combinatorics of the objects that are counted.
 The number   $\alpha_{p,q}\alpha_{m-p,n-q}$ counts the $m\times n$ non final locally saturated tableaux having a final zone  $F=(F_1\times\IntEnt n)\oplus(\IntEnt m\times F_2)$ with $\#F_1=p$ and $\#F_2=q$. Without loss of generality, we assume that $F_1=\IntEnt p$ and $F_2=\IntEnt n\setminus \IntEnt{n-q}$ and so $F=(\IntEnt{m-p}\times\IntEnt q)\cup \left(\{m-p,m-p+1,\dots,m-1\}\times\{q,q+1,\dots,n-1\}\right)$. An illustration is give in Figure \ref{fig-zoneF}.
\begin{figure}[H]
      \centerline{
  \begin{tikzpicture}[scale=0.4]
               \foreach \j in {4,...,9}  { \foreach \i in{0,...,5} {
        \node[scale=1.5,fill=red,,opacity=0.3] at (\j+0.5,\i+0.5) {$ $};
        }
        }
\foreach \j in {0,...,3} { \foreach \i in{6,...,11} {
        \node[scale=1.5,fill=red,,opacity=0.3] at (\j+0.5,\i+0.5) {$ $};
        }
        }       
       \node at (-1,11.5) {$0$};
       \node at (-1,10.5) {$1$};
       \node at (-1,9.5) {$\vdots$};
        \node at (-1,0.5) {$m-1$};
        \node at (-2,6.5) {$m-p-1$};
        \node at (-1.5,5.5) {$m-p$};
        \node at (-2,4.5) {$m-p+1$};
        \node at (-1,3.5) {$\vdots$};
        \node at (0.5,12.5){$0$};
        \node at (1.5,12.5){$1$};
        \node at (4.5,12.5){$q$};
        \node at (5.7,12.5){$q+1$};
        \node at (7.3,12.5) {$\cdots$};
        \node at (2.2,12.5){$\cdots$};
        \node at (9.5,12.5){$n-1$};
                \draw[black, thin] (0,0) rectangle  (10,12); 
              \draw[black, thin] (0,0) rectangle  (4,12);
              \draw[black, thin] (0,0) rectangle  (10,6);
              \end{tikzpicture}
              }
     \caption{Illustration of the zone $F$ (in red in the picture).\label{fig-zoneF}}
     \end{figure}
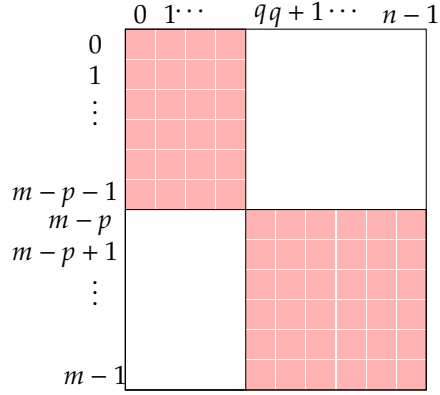
 Let $\LST_{m,n}(p,q)$  be the set of these tableaux. In other words, the set $\LST_{m,n}(p,q)$ is the set of locally saturated $m\times n$-tableaux  $T$ satisfying $T\cap F=\emptyset$. 

Our proof is constructive and consists in exhibiting a map $\phi:\LST_{m,n}(p,q)\longrightarrow \LST_{m,n}(1,1)$ and proving it is an injection (see an illustration in Fig.\ref{fig-mapphi}). 

     \begin{figure}[H]
      \centerline{
  \begin{tikzpicture}[scale=0.4]
          ;
  =
               \foreach \j in {3,...,5} { \foreach \i in{0,...,2} {
        \node[scale=1.5,fill=red,,opacity=0.3] at (\j+0.5,\i+0.5) {$ $};
        }
        }
\foreach \j in {0,...,2} { \foreach \i in{3,...,5} {
        \node[scale=1.5,fill=red,,opacity=0.3] at (\j+0.5,\i+0.5) {$ $};
        }
        }
       \node at (1.5,4.5) {\Large$\emptyset$};
       \node at (4.5,1.5) {\Large$\emptyset$};
       \node at (1.5,1.5) {$T_1$};
       \node at (4.5,4.5) {$T_2$};
        \node at (7.5,3) {\Large$\longrightarrow$};
        \node at (7.5,4) {\Large$\phi$};
              \draw[step=3.0,black, thin] (0,0) grid (6,6);     \end{tikzpicture}
              \begin{tikzpicture}[scale=0.4]
          ;
  =
                \foreach \i in{1,...,5} {
        \node[scale=1.5,fill=red,,opacity=0.3] at (0.5,\i+0.5) {$ $};
        \node[scale=1.5,fill=red,,opacity=0.3] at (\i+0.5,0.5) {$ $};
        }        
       \node at (0.5,3.5) {\Large$\emptyset$};
       \node at (3.5,0.5) {\Large$\emptyset$};
       \node at (0.5,0.5) {$T_1'$};
       \node at (3.5,3.5) {$T_2'$};
              \draw[black, thin] (0,0) rectangle  (6,6); 
              \draw[black, thin] (0,1) rectangle  (1,6);
              \draw[black, thin] (1,0) rectangle  (6,1); \end{tikzpicture}}
          \caption{Illustration of the map $\phi$ \label{fig-mapphi}}
          \end{figure}
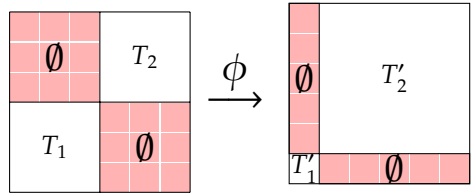
          Before describing $\phi$, we need to introduce some tools on locally saturated tableaux.\\
          \begin{definition}\label{def-rowcol}
          For any tableau $T$,  the set of  indices of the crosses belonging to the $i^{th}$ row (resp. $j^{th}$ column) is denoted by
          \[
          \row_i(T)=\{j\mid (i,j)\in T\}\mbox{ (resp. } \col_j(T)=\{i\mid (i,j)\in T\}\mbox{ ).}
          \]
          Thus defined, the $i^{th}$ row (resp. $j^{th}$ column) of $T$ is the set $\{i\}\times\row_i(T)$ (resp. $\col_j(T)\times\{j\}$).
          \end{definition}
          The following proposition is a reformulation of a result proved in \cite{CLMP15} and initialy stated in terms of words over an alphabet whose letters are indexed by subsets of $\IntEnt n$.
          \begin{proposition}\label{prop-caractsat}
            The three following assertion are equivalent
            \begin{enumerate}
                \item The tableau $T$ is locally saturated.
                \item For any $i_1,i_2\in\IntEnt m$ either $\row_{i_1}(T)=\row_{i_2}(T)$ or $\row_{i_1}(T)\cap\row_{i_2}(T)=\emptyset$.
                \item For any $j_1,j_2\in\IntEnt n$ either $\col_{j_1}(T)=\col_{j_2}(T)$ or $\col_{j_1}(T)\cap\col_{j_2}(T)=\emptyset$.
            \end{enumerate}
          \end{proposition}
          \begin{definition}\label{def-Supp}
                    Let $T$ be a $m\times n$-tableau and $(i,j)\in\IntEnt m\times\IntEnt n$. We define $$\displaystyle\SuppR_i(T)=T\setminus \left(\{i\}\times\row_{i}(T)\right)\text{ and } \SuppC_j(T)=T\setminus \left(\col_j(T)\times \{j\}\right).$$
                    as the tableau $T$ in which  row $i$ (resp. column $j$) has been emptied of its crosses.
          \end{definition}
          \begin{definition}\label{def-Merge}
                    Let $T$ be a locally saturated $m\times n$-tableau and $i_1,i_2\in\IntEnt m$. We define the tableau $\MergeR_{i_1,i_2}(T)$ as the unique  $m\times n$-tableau $T'$ such that for any $i\in\IntEnt m$, $$\row_i(T')=\left\{\begin{array}{ll}\row_{i_1}(T)\cup \row_{i_2}(T)&\text{ if }\mathrm{row}_{i}(T)\in\{\mathrm{row}_{i_1}(T),\row_{i_2}(T)\}\\
                    \row_{i}(T)&\text{ otherwise}.\end{array}\right.$$
                    Symmetrically,  we define the tableau $\MergeC_{j_1,j_2}(T)$ as the unique $m\times n$-tableau $T'$ such that for any $j\in\IntEnt n$,
                    $$\col_j(T')=\left\{\begin{array}{ll}\col_{j_1}(T)\cup \col_{j_2}(T)&\text{ if }\col_{j}(T)\in\{\col_{j_1}(T),\col_{j_2}(T)\}\\
                    \col_{j}(T)&\text{ otherwise}.\end{array}\right.$$
                    
          \end{definition}
           \begin{example}    
          The following picture illustrates the action of $\MergeR_{0,1}$ on a tableau:\\
            
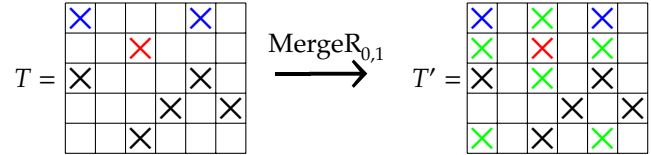
\centerline{
  \begin{tikzpicture}[scale=0.4]
   \node at (-1,2.5) {$T=$};
      \node[scale=1.5] at (2.5,0.5) {$\times$};
      \node[scale=1.5,color=red] at (2.5,3.5) {$\times$};
      \node[scale=1.5] at (3.5,1.5) {$\times$};
      \node[scale=1.5] at (5.5,1.5) {$\times$};
      \node[scale=1.5] at (0.5,2.5) {$\times$};
      \node[scale=1.5] at (4.5,2.5) {$\times$};
      \node[scale=1.5,color=blue] at (0.5,4.5) {$\times$};
      \node[scale=1.5,color=blue] at (4.5,4.5) {$\times$};
              \draw[step=1.0,black, thin] (0,0) grid (6,5);
       \node at (8.5,2.5){\Huge$\longrightarrow$};
   \node at (8.7,3.5){$\MergeR_{0,1}$};
          \end{tikzpicture} 
          \begin{tikzpicture}[scale=0.4]
          \node at (-1,2.5) {$T'=$};
      \node[scale=1.5] at (2.5,0.5) {$\times$};
      \node[scale=1.5,color=red] at (2.5,3.5) {$\times$};
      \node[scale=1.5] at (3.5,1.5) {$\times$};
      \node[scale=1.5] at (5.5,1.5) {$\times$};
      \node[scale=1.5] at (0.5,2.5) {$\times$};
      \node[scale=1.5] at (4.5,2.5) {$\times$};
      \node[scale=1.5,color=green] at (2.5,2.5) {$\times$};
      \node[scale=1.5,color=green] at (2.5,4.5) {$\times$};
      \node[scale=1.5,color=green] at (0.5,3.5) {$\times$};
      \node[scale=1.5,color=green] at (4.5,3.5) {$\times$};
      \node[scale=1.5,color=green] at (0.5,0.5) {$\times$};
      \node[scale=1.5,color=green] at (4.5,0.5) {$\times$};
      \node[scale=1.5,color=blue] at (0.5,4.5) {$\times$};
      \node[scale=1.5,color=blue] at (4.5,4.5) {$\times$};
              \draw[step=1.0,black, thin] (0,0) grid (6,5);  
          \end{tikzpicture}}
          In this example $\row_0(T)=\{0,4\}=\row_2(T)$, 
          $\row_1(T)=\{2\}=\row_4(T)$, and
          \[\row_0(T')=\row_1(T')=\row_2(T')=\row_4(T')=\{0,2,4\}.
          \]
          \end{example}
          \begin{lemma}\label{invarsat}
            Let $T$ be a locally saturated $m\times n$-tableau. The following assertions hold:
            \begin{enumerate}
            \item For any $i_1,i_2\in\IntEnt m$,  the tableau $\MergeR_{i_1,i_2}(T)$ is locally saturated.
            \item For any $j_1, j_2\in\IntEnt n$,  the tableau $\MergeC_{j_1,j_2}(T)$ is locally saturated.
            \item For any $i\in \IntEnt m$, the tableau $\SuppR_i(T)$ is locally saturated.
            \item For any $j\in \IntEnt n$, the tableau $\SuppC_j(T)$ is locally saturated.
            \end{enumerate}
          \end{lemma}
          \begin{proof}
           Assertions 1 and 3 are obtained by using assertion 2 of Proposition \ref{prop-caractsat}. Assertions 2 and 4 are obtained by using assertion 3 of Proposition \ref{prop-caractsat}.
          \end{proof}
         
    Let $\phi:\LST_{m,n}(p,q)\rightarrow\LST_{m,n}(1,1)$ defined by
     \begin{enumerate}
         \item\label{phim1} If $\left(\{m-1\}\times \row_{m-1}(T)\right)\cup
         \left(\col_0(T)\times\{0\}\right)\subset \{(m-1,0)\}$ then we set $\phi(T)=T$.
         
         The set of  images by $\phi$ of the tableaux $T$ that satisfy the condition of this case are those that  have no  crosses in either  the final zone of $T$ or the final zone of $\phi(T)$.
         \item If $\left(\{m-1\}\times\row_{m-1}(T)\right)\cup
         \left(\col_0(T)\times\{0\}\right)\not\subset \{(m-1,0)\}$ then we have to consider the two following cases:
     \begin{enumerate}
         \item\label{phim2a}  If $(m-1,0)\not\in T$ then we have to consider two cases
         \begin{enumerate}
             \item \label{phim2ai} If there exists $(i,j)\in \IntEnt{m-p-1}\times\{q,\dots,n-1\}$ such that $(i,j)\not \in T$ then
             we choose $(i,j)$ minimal for the lexicographic order. We  set \begin{equation}\label{eq-phi2ai}\phi(T)=\SuppR_{m-1}\left(\MergeR_{i,m-1}\left(\SuppC_0\left(\MergeC_{0,j}(T)\right)\right)\right).\end{equation}
             From Lemma \ref{invarsat}, the tableau $\phi(T)$ is locally saturated and the use of the functions $\SuppR_{m-1}$ and $\SuppC_0$  implies that $\phi(T)$ belongs to $\LST_{m,n}(1,1)$ (see an example in Fig. \ref{fig-mapphiex1}).
             
             The set of  images by $\phi$ of the tableaux $T$ that satisfy the condition of this case are those that do not have a cross  at $(m-1,0)$, have at least a cross in the final zone of $T$ and  that do not contain $\IntEnt{m-p-1}\times\{q,\dots,n-1\}$  as a subset.
               
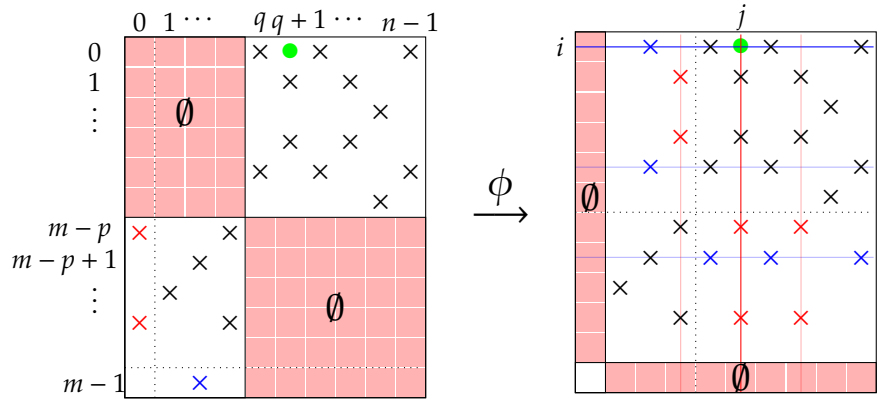
\begin{figure}[H]
      \centerline{
  \begin{tikzpicture}[scale=0.4]
          ;
  =
               \foreach \j in {4,...,9} { \foreach \i in{0,...,5} {
        \node[scale=1.5,fill=red,,opacity=0.3] at (\j+0.5,\i+0.5) {$ $};
        }
        }
\foreach \j in {0,...,3} { \foreach \i in{6,...,11} {
        \node[scale=1.5,fill=red,,opacity=0.3] at (\j+0.5,\i+0.5) {$ $};
        }
        }        
       \node at (-1,11.5) {$0$};
       \node at (-1,10.5) {$1$};
       \node at (-1,9.5) {$\vdots$};
       \node at (-1.5,5.5) {$m-p$};
       \node at (-2,4.5) {$m-p+1$};
       \node at (-1,3.5) {$\vdots$};
        \node at (-1,0.5) {$m-1$};
        \node at (0.5,12.5){$0$};
        \node at (1.5,12.5){$1$};
        \node at (4.5,12.5) {$q$};
        \node at (5.8,12.5) {$q+1$};
        \node at (7.5,12.5) {$\cdots$};
        \node at (2.5,12.5){$\cdots$};
        \node at (9.5,12.5){$n-1$};
       \node at (2,9.5) {\Large$\emptyset$};
       \node at (7,3) {\Large$\emptyset$};
       \node at (2.5,0.5) {$\color{blue}\times$};
       \node at (0.5,2.5) {$\color{red}\times$};
       \node at (3.5,2.5) {$\times$};
       \node at (1.5,3.5) {$\times$};
       \node at (2.5,4.5) {$\times$};
       \node at (0.5,5.5) {$\color{red}\times$};
       \node at (3.5,5.5) {$\times$};
       \node at (8.5,6.5) {$\times$};
        \node at (8.5,9.5) {$\times$};
       \node at (9.5,7.5) {$\times$};
       \node at (6.5,7.5) {$\times$};
       \node at (9.5,11.5) {$\times$};
       \node at (6.5,11.5) {$\times$};
       \node at (5.5,8.5) {$\times$};
       \node at (7.5,8.5) {$\times$};
       \node at (5.5,10.5) {$\times$};
       \node at (7.5,10.5) {$\times$};
             \node at (4.5,7.5) {$\times$};
       \node at (4.5,11.5) {$\times$};
        \node at (12.5,6) {\Large$\longrightarrow$};
        \node at (12.5,7) {\Large$\phi$};
                \draw[black, thin] (0,0) rectangle  (10,12); 
              \draw[black, thin] (0,0) rectangle  (4,12);
              \draw[black, thin] (0,0) rectangle  (10,6);
                \draw[black,dotted] (1,0) -- (1,12);
              \draw[black,dotted] (0,1) -- (10,1);
               \node[color=green] at (5.5,11.5) {\Large$\bullet$};
              \end{tikzpicture}
              \begin{tikzpicture}[scale=0.4]
          ;
  =
    \node[color=green] at (5.5,11.5) {\Large$\bullet$};
               \foreach \j in {1,...,9} {
        \node[scale=1.5,fill=red,,opacity=0.3] at (\j+0.5,0.5) {$ $};
        }
 \foreach \i in{1,...,11} {
        \node[scale=1.5,fill=red,,opacity=0.3] at (0.5,\i+0.5) {$ $};
        }
         \draw[color=blue] (0,11.5)--(10,11.5);
         \draw[color=blue,opacity=0.3] (0,7.5)--(10,7.5);
         \draw[color=blue,opacity=0.3] (0,4.5)--(10,4.5);
         \draw[color=red] (5.5,0)--(5.5,12);
         \draw[color=red,opacity=0.3] (7.5,0)--(7.5,12);
           \draw[color=red,opacity=0.3] (3.5,0)--(3.5,12);
         \node[scale=1] at (5.5,12.5) {$j$};
         \node[scale=1] at (-0.5,11.5) {$i$};
       \node at (0.5,6.5) {\Large$\emptyset$};
       \node at (5.5,0.5) {\Large$\emptyset$};
       \node at (2.5,11.5) {$\color{blue}\times$};
       \node at (2.5,7.5) {$\color{blue}\times$};
       \node at (5.5,2.5) {$\color{red}\times$};
       \node at (7.5,2.5) {$\color{red}\times$};
       \node at (3.5,2.5) {$\times$};
       \node at (1.5,3.5) {$\times$};
       \node at (2.5,4.5) {$\times$};
       \node at (5.5,5.5) {$\color{red}\times$};
       \node at (7.5,5.5) {$\color{red}\times$};
       \node at (3.5,5.5) {$\times$};
       \node at (8.5,6.5) {$\times$};
        \node at (8.5,9.5) {$\times$};
       \node at (9.5,7.5) {$\times$};
       \node at (6.5,7.5) {$\times$};
       \node at (9.5,11.5) {$\times$};
       \node at (6.5,11.5) {$\times$};
       \node at (5.5,8.5) {$\times$};
       \node at (7.5,8.5) {$\times$};
       \node at (5.5,10.5) {$\times$};
       \node at (7.5,10.5) {$\times$};
        \node[color=blue] at (9.5,4.5) {$\times$}; 
        \node[color=blue] at (6.5,4.5) {$\times$}; 
       \node[color=blue] at (4.5,4.5) {$\times$}; 
       \node[color=red] at (3.5,10.5) {$\times$};
       \node[color=red] at (3.5,8.5) {$\times$};
       \node at (4.5,7.5) {$\times$};
       \node at (4.5,11.5) {$\times$};
                \draw[black, thin] (0,0) rectangle  (10,12); 
              \draw[black, thin] (0,0) rectangle  (1,12);
              \draw[black, thin] (0,0) rectangle  (10,1);
              \draw[black,dotted] (4,0) -- (4,12);
              \draw[black,dotted] (0,6) -- (10,6);
              \end{tikzpicture}
             }
             The green dot corresponds to the point $(i,j)$. The blue lines symbolize the operation $\MergeR_{i,m-1}$. The red lines symbolize the operation $\MergeC_{0,j}$.
          \caption{Computation of $\phi$ in the case \ref{phim2a}i \label{fig-mapphiex1}.}
          \end{figure}
             \item\label{phim2aii} If $\IntEnt{m-p-1}\times\{q,\dots,n-1\}\subset T$ then the tableau $\widetilde T=T\setminus \IntEnt{m-p-1}\times\{q,\cdots,n-1\}$ is locally saturated. From Lemma  \ref{invarsat}, the tableau \begin{equation}\label{eq-phi2aii}\widetilde T'=\SuppR_{m-1}\left(\MergeR_{0,m-1}\left(\SuppC_0\left(\MergeC_{0,q}\left(\widetilde T \right)\right)\right)\right)\end{equation} is also locally saturated and belongs to $\LST_{m,n}(1,1)$. We set $\phi(T)=\widetilde T'\cup \{(m-1,0)\}$.
             Since $\widetilde T'\in \LST_{m,n}(1,1)$, we have also  $\phi(T)\in \LST_{m,n}(1,1)$ (see Figure \ref{fig-mapphiex2} for an example).\\
             
              The set of  images by $\phi$ of the tableaux $T$ that satisfy the condition of this case are those that  have a cross  in $(m-1,0)$, have at least a cross in the final zone of $T$ and  that  have no cross in the $\IntEnt{m-p-1}\times\{q,\dots,n-1\}$ zone.
               
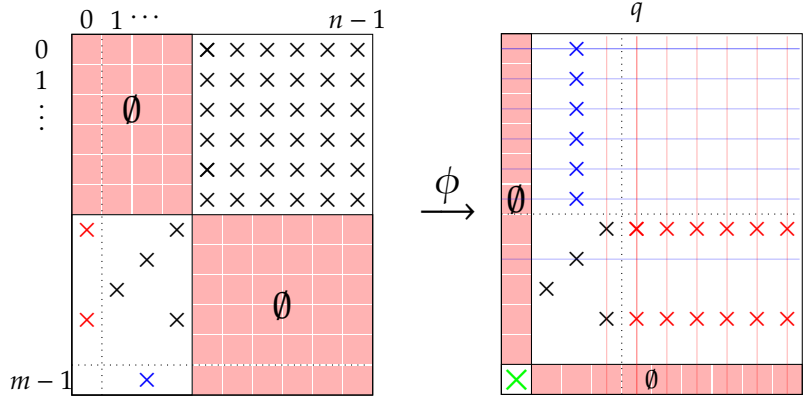
\begin{figure}[H]
      \centerline{
  \begin{tikzpicture}[scale=0.4]
               \foreach \j in {4,...,9} { \foreach \i in{0,...,5} {
        \node[scale=1.5,fill=red,,opacity=0.3] at (\j+0.5,\i+0.5) {$ $};
        }
        }
\foreach \j in {0,...,3} { \foreach \i in{6,...,11} {
        \node[scale=1.5,fill=red,,opacity=0.3] at (\j+0.5,\i+0.5) {$ $};
        }
        }
  \foreach \j in {4,...,9} { \foreach \i in{6,...,11} {
        \node at (\j+0.5,\i+0.5) {$\times$};
        }
        }             
       \node at (-1,11.5) {$0$};
       \node at (-1,10.5) {$1$};
       \node at (-1,9.5) {$\vdots$};
        \node at (-1,0.5) {$m-1$};
        \node at (0.5,12.5){$0$};
        \node at (1.5,12.5){$1$};
        \node at (2.5,12.5){$\cdots$};
        \node at (9.5,12.5){$n-1$};
       \node at (2,9.5) {\Large$\emptyset$};
       \node at (7,3) {\Large$\emptyset$};
       \node at (2.5,0.5) {$\color{blue}\times$};
       \node at (0.5,2.5) {$\color{red}\times$};
       \node at (3.5,2.5) {$\times$};
       \node at (1.5,3.5) {$\times$};
       \node at (2.5,4.5) {$\times$};
       \node at (0.5,5.5) {$\color{red}\times$};
       \node at (3.5,5.5) {$\times$};
             \node at (4.5,7.5) {$\times$};
       \node at (4.5,11.5) {$\times$};
        \node at (12.5,6) {\Large$\longrightarrow$};
        \node at (12.5,7) {\Large$\phi$};
                \draw[black, thin] (0,0) rectangle  (10,12); 
              \draw[black, thin] (0,0) rectangle  (4,12);
              \draw[black, thin] (0,0) rectangle  (10,6);
                \draw[black,dotted] (1,0) -- (1,12);
              \draw[black,dotted] (0,1) -- (10,1);
              \end{tikzpicture}
              \begin{tikzpicture}[scale=0.4]
          ;
  =
               \foreach \j in {1,...,9} {
        \node[scale=1.5,fill=red,,opacity=0.3] at (\j+0.5,0.5) {$ $};
        }
 \foreach \i in{1,...,11} {
        \node[scale=1.5,fill=red,,opacity=0.3] at (0.5,\i+0.5) {$ $};
        }
       \node at (0.5,6.5) {\Large$\emptyset$};
       \node at (3.5,2.5) {$\times$};
       \node at (1.5,3.5) {$\times$};
       \node at (2.5,4.5) {$\times$};
       \node at (3.5,5.5) {$\times$};
                \draw[black, thin] (0,0) rectangle  (10,12); 
              \draw[black, thin] (0,0) rectangle  (1,12);
              \draw[black, thin] (0,0) rectangle  (10,1);
              \draw[black,dotted] (4,0) -- (4,12);
              \draw[black,dotted] (0,6) -- (10,6);
                \foreach \j in{6,...,11}{ \node[color=blue] at (2.5,\j+0.5) {$\times$};
                \draw[color=blue,opacity=0.3] (0,\j+0.5) -- (10,\j+0.5);}
                \foreach \i in{4,...,9} { \node[color=red] at (\i+0.5,2.5) {$\times$};
                 \node[color=red] at (\i+0.5,5.5) {$\times$};
                 \draw[color=red,opacity=0.3] (\i+0.5,0) -- (\i+.5,12);}
                 \node[color=red] at (4.5,5.5) {$\times$};
                 \node[color=green] at (0.5,0.5) {\Large$\times$};
                 \draw[color=blue,opacity=0.3] (0,11.5) -- (10,11.5);
                 \draw[color=blue,opacity=0.3] (0,4.5) -- (10,4.5);
                 \draw[color=red,opacity=0.3] (4.5,0) -- (4.5,12);
                 \draw[color=red,opacity=0.3] (3.5,0) -- (3.5,12);
                 \node at (4.5,12.7) {$q$};
                 \node at (5,0.5) {$\emptyset$};
              \end{tikzpicture}
             }
             The green cross corresponds to the added cross $(m-1,0)$. The blue lines symbolize the operation $\MergeR_{0,m-1}$. The red lines symbolize the operation $\MergeC_{0,q}$.
          \caption{Computation of $\phi$ in the case \ref{phim2a}ii \label{fig-mapphiex2}.}
          \end{figure}
         \end{enumerate}
         \item\label{phim2b} If $(m-1,0)\in T$ then we have to consider two cases
         \begin{enumerate}
             \item  \label{phim2bi} If there exists $(i,j)\in \left(\IntEnt{m-p-1}\times\{q,\dots,n-1\}\right)\cap T$ then we consider such a pair $(i,j)$ minimal for the lexicographic order.  From Lemma \ref{invarsat}, the tableau \begin{equation}\label{eq-phi2bi}\widetilde T=\SuppR_{m-1}\left(\MergeR_{i,m-1}\left(\SuppC_0\left(\MergeC_{0,j}(T)\right)\right)\right)\end{equation} belongs to $\LST_{m,n}(1,1)$.
              Hence, we set $$\phi(T)=\widetilde T\cup\{(m-1,0)\}\in \LST_{m,n}(1,1).$$ 
              See an example in Figure \ref{fig-mapphiex3}.

  The set of  images by $\phi$ of the tableaux $T$ that satisfy the condition of this case are those that  have a cross  at $(m-1,0)$, have at least a cross in the final zone of $T$ and  at least a cross  in the $\IntEnt{m-p-1}\times\{q,\dots,n-1\}$ zone.
               
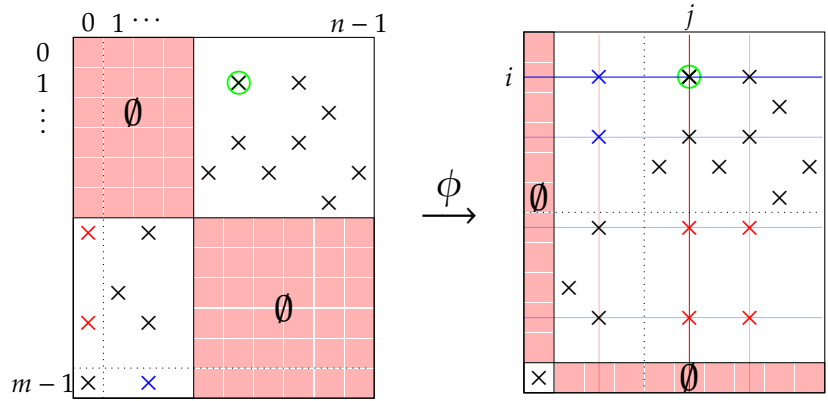
\begin{figure}[H]
      \centerline{
  \begin{tikzpicture}[scale=0.4]
          ;
  =
               \foreach \j in {4,...,9} { \foreach \i in{0,...,5} {
        \node[scale=1.5,fill=red,,opacity=0.3] at (\j+0.5,\i+0.5) {$ $};
        }
        }
\foreach \j in {0,...,3} { \foreach \i in{6,...,11} {
        \node[scale=1.5,fill=red,,opacity=0.3] at (\j+0.5,\i+0.5) {$ $};
        }
        }        
       \node at (-1,11.5) {$0$};
       \node at (-1,10.5) {$1$};
       \node at (-1,9.5) {$\vdots$};
        \node at (-1,0.5) {$m-1$};
        \node at (0.5,12.5){$0$};
        \node at (1.5,12.5){$1$};
        \node at (2.5,12.5){$\cdots$};
        \node at (9.5,12.5){$n-1$};
       \node at (2,9.5) {\Large$\emptyset$};
       \node at (7,3) {\Large$\emptyset$};
       \node at (2.5,0.5) {$\color{blue}\times$};
       \node at (0.5,2.5) {$\color{red}\times$};
      \node at (0.5,5.5) {$\color{red}\times$};
       \node at (2.5,2.5) {$\times$};
       \node at (1.5,3.5) {$\times$};
     \node at (0.5,0.5) {$\times$};
       \node at (2.5,5.5) {$\times$};
       \node at (8.5,6.5) {$\times$};
        \node at (8.5,9.5) {$\times$};
       \node at (9.5,7.5) {$\times$};
       \node at (6.5,7.5) {$\times$};
       \node at (5.5,8.5) {$\times$};
       \node at (7.5,8.5) {$\times$};
       \node[color=green] at (5.5,10.5) {\Large$\otimes$};
       \node at (5.5,10.5) {$\times$};
       \node at (7.5,10.5) {$\times$};
             \node at (4.5,7.5) {$\times$};
        \node at (12.5,6) {\Large$\longrightarrow$};
        \node at (12.5,7) {\Large$\phi$};
                \draw[black, thin] (0,0) rectangle  (10,12); 
              \draw[black, thin] (0,0) rectangle  (4,12);
              \draw[black, thin] (0,0) rectangle  (10,6);
                \draw[black,dotted] (1,0) -- (1,12);
              \draw[black,dotted] (0,1) -- (10,1);
              \end{tikzpicture}
              \begin{tikzpicture}[scale=0.4]
               \foreach \j in {1,...,9} {
        \node[scale=1.5,fill=red,,opacity=0.3] at (\j+0.5,0.5) {$ $};
        }
 \foreach \i in{1,...,11} {
        \node[scale=1.5,fill=red,,opacity=0.3] at (0.5,\i+0.5) {$ $};
        }
       \node at (0.5,0.5) {$\times$};
           \node[color=green] at (5.5,10.5) {\Large$\otimes$};
       \node at (5.5,10.5) {$\times$};
         \draw[color=blue] (0,10.5)--(10,10.5);
         \draw[color=blue,opacity=0.3] (0,8.5)--(10,8.5);
         \draw[color=blue,opacity=0.3] (0,5.5)--(10,5.5);
         \draw[color=blue,opacity=0.3] (0,2.5)--(10,2.5);
         \draw[color=red] (5.5,0)--(5.5,12);
         \draw[color=red,opacity=0.3] (7.5,0)--(7.5,12);
           \draw[color=red,opacity=0.3] (2.5,0)--(2.5,12);
           \node[color=blue] at (2.5,8.5) {$\times$};
           \node[color=blue] at (2.5,10.5) {$\times$};
           \node[color=red] at (5.5,5.5) {$\times$};
           \node[color=red] at (5.5,2.5) {$\times$};
            \node[color=red] at (7.5,5.5) {$\times$};
           \node[color=red] at (7.5,2.5) {$\times$};
         \node[scale=1] at (5.5,12.5) {$j$};
         \node[scale=1] at (-0.5,10.5) {$i$};
       \node at (0.5,6.5) {\Large$\emptyset$};
       \node at (5.5,0.5) {\Large$\emptyset$};
       \node at (2.5,2.5) {$\times$};
       \node at (1.5,3.5) {$\times$};
       \node at (2.5,5.5) {$\times$};
       \node at (8.5,6.5) {$\times$};
        \node at (8.5,9.5) {$\times$};
       \node at (9.5,7.5) {$\times$};
       \node at (6.5,7.5) {$\times$};
       \node at (5.5,8.5) {$\times$};
       \node at (7.5,8.5) {$\times$};
       \node at (5.5,10.5) {$\times$};
       \node at (7.5,10.5) {$\times$};
       \node at (4.5,7.5) {$\times$};
%
                \draw[black, thin] (0,0) rectangle  (10,12); 
              \draw[black, thin] (0,0) rectangle  (1,12);
              \draw[black, thin] (0,0) rectangle  (10,1);
              \draw[black,dotted] (4,0) -- (4,12);
              \draw[black,dotted] (0,6) -- (10,6);
              \end{tikzpicture}
             }
             The green circle corresponds to the point $(i,j)$. The blue lines symbolize the operation $\MergeR_{i,m-1}$. The red lines symbolize the operation $\MergeC_{0,j}$.
          \caption{Computation of $\phi$ in the case \ref{phim2b}i \label{fig-mapphiex3}.}
          \end{figure}
             \item \label{phim2bii} If $\left(\IntEnt{m-p-1}\times\{q,\cdots,n-1\}\right)\cap T=\emptyset$ then the tableau $\widetilde T=T\cup \IntEnt{m-p-1}\times\{q,\dots,n-1\}$ is locally saturated. We set
             \begin{equation}\label{eq-phi2bii}\phi(T)=\SuppR_{m-1}\left(\MergeR_{m-p,m-1}\left(\SuppC_0\left(\MergeC_{0,n-1}\left(\widetilde T \}\right)\right)\right)\right).\end{equation}
             From Lemma \ref{invarsat}, the tableau $\phi(T)$ belongs to $\LST_{m,n}(1,1)$ (see Fig. \ref{fig-mapphiex4} for an example).\\
             
             The set of  images by $\phi$ of the tableaux $T$ that satisfy the condition of this case are those that  have no cross  in $(m-1,0)$, have at least a cross in the final zone of $T$ and  contains the $\IntEnt{m-p-1}\times\{q,\dots,n-1\}$ zone as a subset.
                           \begin{figure}[H]
      \centerline{
  \begin{tikzpicture}[scale=0.4]
          ;
               \foreach \j in {4,...,9} { \foreach \i in{0,...,5} {
        \node[scale=1.5,fill=red,,opacity=0.3] at (\j+0.5,\i+0.5) {$ $};
        }
        }
\foreach \j in {0,...,3} { \foreach \i in{6,...,11} {
        \node[scale=1.5,fill=red,,opacity=0.3] at (\j+0.5,\i+0.5) {$ $};
        }
        }        
       \node at (-1,11.5) {$0$};
       \node at (-1,10.5) {$1$};
       \node at (-1,9.5) {$\vdots$};
        \node at (-1,0.5) {$m-1$};
        \node at (0.5,12.5){$0$};
        \node at (1.5,12.5){$1$};
        \node at (2.5,12.5){$\cdots$};
        \node at (9.5,12.5){$n-1$};
       \node at (2,9.5) {\Large$\emptyset$};
       \node at (7,3) {\Large$\emptyset$};
       \node at (2.5,0.5) {$\color{blue}\times$};
       \node at (0.5,2.5) {$\color{red}\times$};
      \node at (0.5,5.5) {$\color{red}\times$};
       \node at (2.5,2.5) {$\times$};
       \node at (1.5,3.5) {$\times$};
     \node at (0.5,0.5) {$\times$};
       \node at (2.5,5.5) {$\times$};
        \node at (12.5,6) {\Large$\longrightarrow$};
        \node at (12.5,7) {\Large$\phi$};
                \draw[black, thin] (0,0) rectangle  (10,12); 
              \draw[black, thin] (0,0) rectangle  (4,12);
              \draw[black, thin] (0,0) rectangle  (10,6);
                \draw[black,dotted] (1,0) -- (1,12);
              \draw[black,dotted] (0,1) -- (10,1);
              \end{tikzpicture}
              \begin{tikzpicture}[scale=0.4]
               \foreach \j in {1,...,9} {
        \node[scale=1.5,fill=red,,opacity=0.3] at (\j+0.5,0.5) {$ $};
        }
 \foreach \i in{1,...,11} {
        \node[scale=1.5,fill=red,,opacity=0.3] at (0.5,\i+0.5) {$ $};
        }
        \foreach \j in {4,...,9} {
            \foreach \i in {6,...,11} {
            \node at (0.5+\j,0.5+\i) {$\times$};
            }
        }
        \foreach \j in {4,...,9} {
        \node[color=red] at (\j+0.5,2.5) {$\times$};
        \node[color=red] at (0.5+\j,5.5) {$\times$};
        }
        \foreach \i in {6,...,11} {
        \node[color=blue] at (2.5,0.5+\i) {$\times$};
        }
           \node[color=blue] at (2.5,8.5) {$\times$};
           \node[color=blue] at (2.5,10.5) {$\times$};
           \node[color=red] at (5.5,5.5) {$\times$};
           \node[color=red] at (5.5,2.5) {$\times$};
            \node[color=red] at (7.5,5.5) {$\times$};
           \node[color=red] at (7.5,2.5) {$\times$};
       \node at (0.5,6.5) {\Large$\emptyset$};
       \node at (5.5,0.5) {\Large$\emptyset$};
       \node at (2.5,2.5) {$\times$};
       \node at (1.5,3.5) {$\times$};
       \node at (2.5,5.5) {$\times$};
       \node at (4.5,7.5) {$\times$};   
                \draw[black, thin] (0,0) rectangle  (10,12); 
              \draw[black, thin] (0,0) rectangle  (1,12);
              \draw[black, thin] (0,0) rectangle  (10,1);
              \draw[black,dotted] (4,0) -- (4,12);
              \draw[black,dotted] (0,6) -- (10,6);
              \end{tikzpicture}
             }
          \caption{Computation of $\phi$ in the case \ref{phim2b}ii \label{fig-mapphiex4}.}
          \end{figure}
         \end{enumerate}
     \end{enumerate}
     \end{enumerate}
     Figure \ref{table-zoneABC} shows that  sets of images by $\phi$ for each of the previous cases are pairwise disjoint. 
     
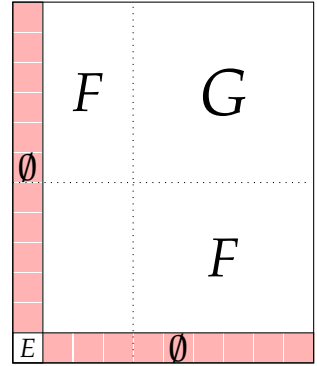
\begin{figure}[H]
     \begin{minipage}{0.6\linewidth}
     $
     \begin{array}{|c|c|c|c|}\hline
     &E\cap\phi(T)&F\cap\phi(T)&G\cap\phi(T)\\\hline
     \mbox{case \ref{phim1}}&?&\emptyset&?\\\hline
     \mbox{case \ref{phim2a}i}&\emptyset&\neq\emptyset&\neq G\\\hline
     \mbox{case \ref{phim2a}ii}&E&\neq\emptyset&\emptyset\\\hline
     \mbox{case \ref{phim2b}i}&E&\neq\emptyset&\neq\emptyset\\\hline
     \mbox{case \ref{phim2b}ii}&\emptyset&\neq\emptyset& G\\\hline
     \end{array}
     $\\ \\
     \end{minipage}
    \begin{minipage}{0.4\linewidth}
            \begin{tikzpicture}[scale=0.4]
       
               \foreach \j in {1,...,9} {
        \node[scale=1.5,fill=red,,opacity=0.3] at (\j+0.5,0.5) {$ $};
        }
 \foreach \i in{1,...,11} {
        \node[scale=1.5,fill=red,,opacity=0.3] at (0.5,\i+0.5) {$ $};
        }
   \node at (7,9){\Huge $G$};
   \node at (2.5,9){\huge $F$};
   \node at (7,3.5){\huge $F$};
    \node at (0.5,0.5) {$E$};
       \node at (0.5,6.5) {\Large$\emptyset$};
       \node at (5.5,0.5) {\Large$\emptyset$};
                \draw[black, thin] (0,0) rectangle  (10,12); 
              \draw[black, thin] (0,0) rectangle  (1,12);
              \draw[black, thin] (0,0) rectangle  (10,1);
              \draw[black,dotted] (4,0) -- (4,12);
              \draw[black,dotted] (0,6) -- (10,6);
             
              \end{tikzpicture}

    \end{minipage}
      $E=\{(m-1,0)\}$,\\ $F=\left(\IntEnt{m-p}\times\IntEnt q\right)\cup\left(\{m-p,\dots,m-1\}\times\{q,\dots,n-1\}\right)$,\\ 
      $G=\IntEnt{m-p}\times \{q,\dots,n-1\}.$
      \caption{Configurations of three pairwise disjoint zones $E$, $F$, and $G$  in $\phi(T)$ with respect to cases \ref{phim1}, \ref{phim2a}i, \ref{phim2a}ii, 
     \ref{phim2b}i, and \ref{phim2b}ii. \label{table-zoneABC}}
     \end{figure}
     
     So it remains to check that there is a $\psi$ function that allows us to go back from $\phi(T)$ to $T$ in each case. Let $T'$ be any tableau in the image of $\phi$. Recall that $F$ is the final zone of $T$.
    
     The definition of $\psi$ is as follows:
     \begin{enumerate}
         \item We set $\psi(T')=T'$. 
         \item \begin{enumerate}
             \item \begin{enumerate}
                 \item Let $(i,j)$ be the minimal element for the lexicographic order in $\left(\IntEnt{m-p}\times \{q,\cdots,n-1\}\right)\setminus T'$.  We set \begin{equation}\label{eq-psi2ai}\psi(T')=\MergeC_{0,j}(\MergeR_{i,m-1}(T'))\setminus F.\end{equation}
                 \item We set
                 \begin{equation}\label{eq-psi2aii} \psi(T')=\left(\MergeC_{0,q}(\MergeR_{0,m-1}(T'))\cup \left(\IntEnt{m-p}\times\{q,\dots,n-1\}\right)\right)\setminus(F\cup \{(m-1,0)\}).\end{equation}
             \end{enumerate}
             \item \begin{enumerate}
                 \item Let $(i,j)$ be the minimal element for the lexicographic order in $\left(\IntEnt{m-p}\times \{q,\cdots,n-1\}\right)\cap T'$. 
                 We set \begin{equation}\label{eq-psi2bi}\psi(T')=\MergeC_{0,j}(\MergeR_{i,m-1}(T'))\setminus F.\end{equation}
                 \item We set
                 \begin{equation}\label{eq-psi2bii}\psi(T')=\left(\{(0,m-1)\}\cup (\MergeC_{0,q}\left(\MergeR_{0,m-1}(T')\right)\right)\setminus \left(F\cup \left(\IntEnt{m-p}\times\{q,\dots,n-1\}\right) \right)\end{equation}
             \end{enumerate}
         \end{enumerate}
     \end{enumerate}
     Consider also the zone 
     $D=\{(m-1,0)\}\cup\left((\IntEnt {m-1}\times(\IntEnt n\setminus \{0\})\setminus F\right)$ (see Figure \ref{fig-zoneD}).
     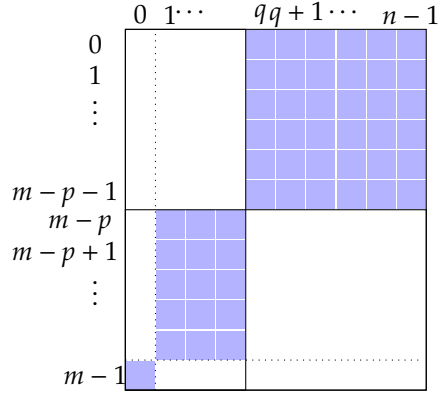
\begin{figure}[H]
      \centerline{
  \begin{tikzpicture}[scale=0.4]
               \foreach \j in {1,...,3} { \foreach \i in{1,...,5} {
        \node[scale=1.5,fill=blue,,opacity=0.3] at (\j+0.5,\i+0.5) {$ $};
        }
        }
\foreach \j in {4,...,9} { \foreach \i in{6,...,11} {
        \node[scale=1.5,fill=blue,,opacity=0.3] at (\j+0.5,\i+0.5) {$ $};
        }
        }
        \node[scale=1.5,fill=blue,opacity=0.3] at (0.5,0.5) {$ $};
       \node at (-1,11.5) {$0$};
       \node at (-1,10.5) {$1$};
       \node at (-1,9.5) {$\vdots$};
        \node at (-1,0.5) {$m-1$};
        \node at (-2,6.5) {$m-p-1$};
        \node at (-1.5,5.5) {$m-p$};
        \node at (-2,4.5) {$m-p+1$};
        \node at (-1,3.5) {$\vdots$};
        \node at (0.5,12.5){$0$};
        \node at (1.5,12.5){$1$};
        \node at (4.5,12.5){$q$};
        \node at (5.7,12.5){$q+1$};
        \node at (7.3,12.5) {$\cdots$};
        \node at (2.2,12.5){$\cdots$};
        \node at (9.5,12.5){$n-1$};
                \draw[black, thin] (0,0) rectangle  (10,12); 
              \draw[black, thin] (0,0) rectangle  (4,12);
              \draw[black, thin] (0,0) rectangle  (10,6);
                \draw[black,dotted] (1,0) -- (1,12);
              \draw[black,dotted] (0,1) -- (10,1);
              \end{tikzpicture}
              }
     \caption{Illustration of the zone $D$ (in blue in the picture).\label{fig-zoneD}}
     \end{figure}
     \begin{lemma}\label{lem-LSTD}
      For any $T\in\LST_{m,n}(p,q)$ we have
      \begin{equation}\label{eq-D}
      T\cap D=\psi(\phi(T))\cap D.
      \end{equation}
     \end{lemma}
     \begin{proof}
      In what follows, we check that Equality (\ref{eq-D}) holds for all  cases of the definition of $\phi$:
      \begin{enumerate}
          \item In that case, both $\phi$ and $\psi$ equal the identity. Hence, the result is trivial.
          \item \begin{enumerate}
              \item \begin{enumerate}
                  \item Since $T\cap F=\emptyset$, Formula (\ref{eq-phi2ai}) and  definitions of $\MergeC$ and $\MergeR$ imply that $T\cap D=\phi(T)\cap D$. Since $\phi(T)\cap ((\IntEnt{m}\times\{0\})\cup(\{m-1\}\times \IntEnt{n}))=\emptyset$, Equality \ref{eq-psi2ai}) and  definition of $\MergeC$ and $\MergeR$ imply that $\phi(T)\cap D=\psi(\phi(T))\cap D$. Hence, $T\cap D=\psi(\phi(T))\cap D$.
                  \item For any tableau $T$ satisfying  condition \ref{phim2aii}, we have $T\cap D=\IntEnt{m-p}\times \{q,\dots,n-1\}$. But from Formula (\ref{eq-psi2aii}), we have also 
                  $\phi(\psi(T))\cap D=\IntEnt{m-p}\times \{q,\dots,n-1\}$. Hence, $T\cap D=\psi(\phi(T))\cap D$.
              \end{enumerate}
              \item \begin{enumerate}
                  \item Since $T\cap F=\emptyset$, Formula (\ref{eq-phi2bi}) and  definition of $\MergeC$ and $\MergeR$ imply that $T\cap D=\phi(T)\cap D$. Since $\phi(T)\cap ((\IntEnt{m}\times\{0\})\cup(\{m-1\}\times \IntEnt{n}))=\emptyset$, Equality \ref{eq-psi2bi}) and  definition of $\MergeC$ and $\MergeR$ imply that $\phi(T)\cap D=\psi(\phi(T))\cap D$. Hence, $T\cap D=\psi(\phi(T))\cap D$.
                  \item For any tableau $T$ satisfying  condition \ref{phim2bii}, we have $T\cap D=\emptyset$. But from Formula (\ref{eq-psi2bii}), we have also,
                  $\phi(\psi(T))\cap D=\emptyset$. Hence, $T\cap D=\psi(\phi(T))\cap D$.
              \end{enumerate}
          \end{enumerate}
      \end{enumerate}
     \end{proof}
     We shall now prove that $\psi$ is the reciprocal function of $\phi$.
     \begin{proposition}\label{prop-phiinj}
     For any $T\in \LST_{m,n}(p,q)$, we have $\psi(\phi(T))=T$.
     \end{proposition}
     \begin{proof}
      From the definition of $\psi$, it is easy to check that $\psi(\phi(T))\cap F=\emptyset=T\cap F$. So using Lemma \ref{lem-LSTD}, it suffices to prove that 
      \[T\cap R=\psi(\phi(T))\cap R,\]
      where
      \[
      R=\IntEnt m\times\IntEnt n\setminus(D\cup F)=R_1\cup R_2.
      \]
      with $R_1=\{m-p,\dots,m-2\}\times \{0\}$ and $R_2=\{m-1\}\times\{1,\dots,p-1\}$. Hence, it suffices to prove that \begin{equation}\label{eq-R1}\psi(\phi(T))\cap R_1 =T\cap R_1,\end{equation} the other equality ($\psi(\phi(T))\cap R_2 =T\cap R_2$) being obtained symmetrically.
      In what follows, we check Equality (\ref{eq-R1}) for every cases occurring in the definition of $\phi$:
    
      \begin{enumerate}
      \item If $T$ matches with  case \ref{phim1} of  definition of $\phi$ then both $\phi$ and $\psi$ equal the identity. So the result is obvious.
      \item 
      \begin{enumerate}
      \item \begin{enumerate}
          \item Let $T\in \LST_{m,n}(p,q)$ be any tableau satisfying the condition of case \ref{phim2a}i of the definition of $\phi$. Let $(i,j)$ be the smallest pair of $\left(\IntEnt{m-p}\times \{q,\cdots,n-1\}\right)\setminus T$ for the lexicographic order.  From Equality (\ref{eq-phi2ai}) and the definition of $\MergeC$, we have $(i',0)\in T\cap R_1$ if and only if $(i',j)\in \phi(T)$ for any $i'\in \{m-p,\ldots, m-2\}$. Furthermore from Equality (\ref{eq-psi2ai}) and the definition of $\MergeC$, we have also for any $i'\in \{m-p,\ldots, m-2\}$, $(i',j)\in\phi(T)$ if and only if  $(i',0)\in\psi(\phi(T))\cap R_1$ because $\phi(T)\cap (\IntEnt{m-p}\times\{0\})=\emptyset$.  Hence, $(i',0)\in T\cap R_1$ if and only if $(i',0)\in\psi(\phi(T))\cap R_1$. This proves Equality (\ref{eq-R1}).
                 \item Let $T\in\LST_{m,n}(p,q)$ be any tableau satisfying the condition of case \ref{phim2a}ii of the definition of $\phi$.  From Equality (\ref{eq-phi2aii}) and the definition of $\MergeC$, for any $m-p-1<i<m-1$,  we have $(i,0)\in T\cap R_1$ if and only if $(i,p)\in \phi(T)$. Furthermore, since $\phi(T)\cap (\IntEnt{m-p}\times\{0\})=\emptyset$, equality (\ref{eq-psi2aii}) and the definition of $\MergeC$ imply that for any  $i\in \{m-p,\ldots m-2\}$,  we have $(i,p)\in \phi(T)$ if and only if $(i,0)\in \psi(\phi(T))\cap R_1$.  Hence,  we have $(i,0)\in T\cap R_1$ if and only if $(i,0)\in\psi(\phi(T))\cap R_1$. This proves Equality (\ref{eq-R1}).
            \end{enumerate}
        \item \begin{enumerate}
            \item Let $T\in\LST_{m,n}(p,q)$ be any tableau satisfying the condition of case \ref{phim2b}i of the definition of $\phi$. Let $(i,j)$ be the smallest pair of $\left(\IntEnt{m-p}\times \{q,\cdots,n-1\}\right)\cap T$ for the lexicographic order. From Equality (\ref{eq-phi2bi}) and the definition of  $\MergeC$, for any $i'\in\{m-p,\ldots, m-2\}$,  we have $(i',0)\in T\cap R_1$ if and only if $(i',j)\in \phi(T)$. 
            Furthermore, since $\phi(T)\cap (\IntEnt{m-p}\times\{0\})=\{(m-1,0\})$, Equality (\ref{eq-psi2bi}) and the definition of $\MergeC$ imply that for any  $i'\in \{m-p,\ldots, m-2\}$,  we have $(i',j)\in \phi(T)$ if and only if $(i',0)\in \psi(\phi(T))\cap R_1$.  Hence, $(i',0)\in T\cap R_1$ if and only if $(i',0)\in\psi(\phi(T))\cap R_1$. This proves Equality (\ref{eq-R1}).
            \item Let $T\in\LST_{m,n}(p,q)$ be any tableau satisfying the condition of case \ref{phim2b}ii of the definition of $\phi$.  From Equality (\ref{eq-phi2bii}) and the definition of  $\MergeC$, for any $i\in\{m-p,\ldots, m-2\}$,  we have $(i,0)\in T\cap R_1$ if and only if $(i,q)\in \phi(T)$. 
            Furthermore, since $\phi(T)\cap (\IntEnt{m-p}\times\{0\})=\emptyset$, Equality (\ref{eq-psi2bii}) and the definition of $\MergeC$ imply that for any  $i\in\{m-p,\ldots, m-2\}$,  we have $(i,q)\in \phi(T)$ if and only if $(i,0)\in \psi(\phi(T))\cap R_1$.  Hence,  we have $(i,0)\in T\cap R_1$ if and only if $(i,0)\in\psi(\phi(T))\cap R_1$. This proves Equality (\ref{eq-R1}).
        \end{enumerate}
        \end{enumerate}
    \end{enumerate}
    Hence, Equality (\ref{eq-R1}) holds in any case and this implies the proposition.
      \end{proof}
     Proposition \ref{prop-phiinj} implies that $\phi$ is an injection and, since $\alpha_{1,1}=2$, we obtain
     \[
     \max\{\alpha_{p,q}\alpha_{m-p,n-q}\mid p\in\{1,\dots,m-1\}, q\in\{1,\dots,n-1\}\}=2\alpha_{m-1,n-1}.
     \]
     We deduce the following theorem:
     \begin{theorem}\label{th-sctypeX}
       When $m,n>1$ we have
       \[
       \sc_{\ostar}(m,n)=\alpha'_{m,n}+2\alpha_{m-1,n-1}-1
       \]
       and  the $2$-monsters $\mon_{m,n}^{\{f_1\},\{0\}}$ and $\mon_{m,n}^{\{0\},\{f_2\}}$, for $f_1\in\IntEnt m\setminus\{0\}$ and $f_2\in\IntEnt n\setminus \{0\}$ are witnesses.
     \end{theorem}
     \begin{example}
  We use formulas of \cite{CLMP15} to compute the first values that are listed in Table \ref{table-scXvalues}.
  The valid saturated tableaux illustrating the case $m=n=2$ are pictured in Figure \ref{fig-exX}.
     \begin{table}[H]
     \[\begin{array}{|c|ccccccc|}
     \hline
     m\setminus n& 2&3&4&5&6&7&8\\\hline
     2&8&20&50&128&338&920&2570\\
     3&20& 66& 212& 690& 2300& 7866& 27572\\
     4&50& 212& 848& 3368& 13520& 55232& 230168\\
     5&128& 690& 3368& 15930& 75008& 355890& 1711208\\
     6&338& 2300& 13520& 75008& 407528& 2206880& 12020360\\
     7&920& 7866& 55232& 355890& 2206880& 13482546& 82181312\\
     8&2570& 27572& 230168& 1711208& 12020360& 82181312& 555813728\\\hline
     \end{array}
     \]
     \caption{First values of $\sc_{\ostar}(m,n)$ for the type X.\label{table-scXvalues}}
     \end{table}
          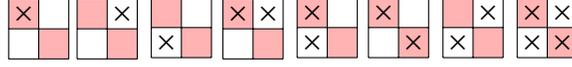
\begin{figure}[h] \centerline{
     \begin{tikzpicture}[scale=0.4]
\node[fill=red,opacity=0.3,scale=1.5] at (1.5,0.5) {$ $}; 
\node[fill=red,opacity=0.3,scale=1.5] at (0.5,1.5) {$ $};
\node at (0.5,1.5) {$\times$};
    \draw[step=1.0,black, thin] (0,0) grid (2,2); 
     \end{tikzpicture}
      \begin{tikzpicture}[scale=0.4]
\node[fill=red,opacity=0.3,scale=1.5] at (1.5,0.5) {$ $}; 
\node[fill=red,opacity=0.3,scale=1.5] at (0.5,1.5) {$ $};
\node at (1.5,1.5) {$\times$};
    \draw[step=1.0,black, thin] (0,0) grid (2,2); 
     \end{tikzpicture}
          \begin{tikzpicture}[scale=0.4]
\node[fill=red,opacity=0.3,scale=1.5] at (1.5,0.5) {$ $}; 
\node[fill=red,opacity=0.3,scale=1.5] at (0.5,1.5) {$ $};
\node at (0.5,0.5) {$\times$};
    \draw[step=1.0,black, thin] (0,0) grid (2,2); 
     \end{tikzpicture}
          \begin{tikzpicture}[scale=0.4]
\node[fill=red,opacity=0.3,scale=1.5] at (1.5,0.5) {$ $}; 
\node[fill=red,opacity=0.3,scale=1.5] at (0.5,1.5) {$ $};
\node at (1.5,1.5) {$\times$};
\node at (0.5,1.5) {$\times$};
    \draw[step=1.0,black, thin] (0,0) grid (2,2); 
     \end{tikzpicture}
       \begin{tikzpicture}[scale=0.4]
\node[fill=red,opacity=0.3,scale=1.5] at (1.5,0.5) {$ $}; 
\node[fill=red,opacity=0.3,scale=1.5] at (0.5,1.5) {$ $};
\node at (0.5,0.5) {$\times$};
\node at (0.5,1.5) {$\times$};
    \draw[step=1.0,black, thin] (0,0) grid (2,2); 
     \end{tikzpicture}
       \begin{tikzpicture}[scale=0.4]
\node[fill=red,opacity=0.3,scale=1.5] at (1.5,0.5) {$ $}; 
\node[fill=red,opacity=0.3,scale=1.5] at (0.5,1.5) {$ $};
\node at (1.5,0.5) {$\times$};
\node at (0.5,1.5) {$\times$};
    \draw[step=1.0,black, thin] (0,0) grid (2,2); 
     \end{tikzpicture}
       \begin{tikzpicture}[scale=0.4]
\node[fill=red,opacity=0.3,scale=1.5] at (1.5,0.5) {$ $}; 
\node[fill=red,opacity=0.3,scale=1.5] at (0.5,1.5) {$ $};
\node at (1.5,1.5) {$\times$};
\node at (0.5,0.5) {$\times$};
    \draw[step=1.0,black, thin] (0,0) grid (2,2); 
     \end{tikzpicture}
       \begin{tikzpicture}[scale=0.4]
\node[fill=red,opacity=0.3,scale=1.5] at (1.5,0.5) {$ $}; 
\node[fill=red,opacity=0.3,scale=1.5] at (0.5,1.5) {$ $};
\node at (1.5,1.5) {$\times$};
\node at (0.5,1.5) {$\times$};
\node at (0.5,0.5) {$\times$};
\node at (1.5,0.5) {$\times$};
    \draw[step=1.0,black, thin] (0,0) grid (2,2); 
     \end{tikzpicture}
     }
     \caption{The $8$ saturated valid $2\times 2$-tableaux for the type X, $F_1=\{1\}$, $F_2=\{0\}$.\label{fig-exX}}
     \end{figure}
     \end{example}
     \subsection{Small automata ($m=1$ or $n=1$)}
     Let us consider only the case where $m=1$, the other case being obtained symmetrically. Since the notion of saturation is related to $2\times 2$-patterns, it becomes irrelevant when we assume $m=1$. If $n=1$ then the state complexity is obviously $1$. We consider now $n>1$. We just have to investigate the valid tableaux with respect to the final zone of the resulting automaton. A fast computation shows that the number of valid tableaux is maximal when the final zone is $F=\{\mathrm f\}$ with $\mathrm f>0$.  In that case, a tableau is valid if and only if it contains $(0,0)$ or it  contains neither $(0,0)$ nor $(0,\mathrm f)$. Hence, there is $2^{n-1}+2^{n-2}=\frac342^n$ valid tableaux. We first check that all valid tableaux are accessible. We proceed by induction and consider the order $<$ of Definition \ref{def-order}. The tableaux $\{\}$ and $\{(0,0)\}$ are obviously accessible. Let $T\neq\emptyset$ and $T\neq\{(0,0)\}$. We consider three cases
     \begin{enumerate}
         \item If neither $(0,0)$ nor $(0,\mathrm f)$ are in $T$ then we choose $(0,j)\in T$ and we set $T'=T\cup \{(0,0)\}\setminus\{(0,j)\}$. Since $T'<T$, by induction hypothesis, the tableau $T'$ is accessible and so $T=\delta^{(\Id,(0,j))}T'$ is accessible too.
         \item  If $(0,0)\in T$ and $(0,\mathrm  f)\not\in T$ then we choose $(0,j)\in T$ and we set $T'=T\cup\{(0,\mathrm f)\}\setminus \{(0,j)\}$. The tableau $T'$ is valid and since $T'<T$, the induction hypothesis show that $T'$ is accessible. Hence, the tableau $T=\delta^{(\Id,(j,\mathrm f))}$ is accessible too.
         \item If both $(0,0)$ and $(0,\mathrm  f)$ are in $T$ then we set $T'=T\setminus\{(0,\mathrm  f)\}$. Since $T'<T$, by induction hypothesis, the tableau $T'$ is accessible and so $T=\delta^{(\Id,(0,\mathrm f))}T'$ is accessible too.
     \end{enumerate}
     This proves that all valid tableaux are accessible. Now we show that they are pairwise non-equivalent. Let $T\neq T'$ be two valid tableaux. Without restriction, we assume that there exists $j$ such that $(0,j)\in T$ and $(0,j)\not\in T'$. Let $g\in\IntEnt n^{\IntEnt n}$ defined by $g(j')=0$ if $j'\neq j$ and $g(j)=\mathrm f$. The tableau $\delta^{(\Id,g)}T$ is final while the tableau $\delta^{(\Id,g)}T'$ is non final. This shows that valid tableaux are pairwise non-equivalent.\\ 
     It remains to find $F_1$ and $F_2$ such that $F_1\bullet F_2=\{\mathrm f\}$ whatever the Boolean operation $\bullet$ is.  For $\bullet=\oplus$ or $\bullet=\cup$, we set $F_1=\emptyset$ and $F_2=\{\mathrm  f\}$ while for $\bullet=\cap$, we set $F_1=\{0\}$ and $F_2=\{\mathrm f\}$.\\
     It follows that the state complexity is $\frac342^n$ and a witness is $\mon_{1,n}^{F_1,F_2}$ with $F_1$ and $F_2$ as above.
    \section{Conclusion}
    The interest of this discussion does not only lie in the fact that a new state complexity is calculated but especially because it illustrates methods and strategies that can be implemented to compute state complexities. The first lesson  is that it is often easier to compute state complexities of families of operators than of isolated operations. Indeed,   identifying common points between different operations leads us to develop theoretical tools that are relevant in the context. We have already illustrated this approach in previous papers. The most telling example is the family of $1$-uniform operations that led us to develop the theory of modifiers and monsters \cite{CHLP20} (also called OLPA in \cite{Dav18}). This rather general theory does not directly give formulas but offers a  mathematical framework to calculate them. Among other examples that we have studied, let us mention the family of multiple catenation operators \cite{CLP19}, the family of friendly operators  \cite{CHL20b},
and the quasi-Boolean operations \cite{CHL20a}.

We are specifically interested in classes of operators, called $1$-uniforms, that can be encoded using modifiers. For this range of problems, a three-step strategy  emerged and this is what the paper illustrates. For many 1-uniform operations, we know algorithms that encode them on automata and behind many of these algorithms are in fact modifiers. Modifiers explain how  states and transitions of the input automata are transformed  and make combinatorial objects naturally appear as  states of the output automaton (in our case the combinatorial objects are tableaux). Transitions of the output automaton are also described combinatorially through an action of the transformation monoid on the combinatorial objects. This is the first step in the strategy and is illustrated in the preliminary paragraph of Section \ref{sec-sc}. It allows to (re)encode the problem initially stated in terms of language theory into a combinatorial language through  algebraic tools (modifiers). The second step is to study the combinatorics of the  objects and to deduce properties of the output automaton. In the present paper, the combinatorial notion of validity is linked to the notion of accessibility (see Section \ref{sec-valid}) and the notion of local saturation is linked to the Nerode equivalence (see Sections \ref{section-saturation}, \ref{2x2case}, and \ref{sec-localsat}). Finally, the third step is to find witnesses.  We choose witnesses among the monsters by using combinatorics as a guide. Section \ref{sec-monster}  illustrates this step.

One of the main interest of this method is that even when it does not allow to obtain a closed formula for the state complexity, it allows to express it as the number of elements of a set having a combinatorial description. Let's take the example of the shuffle product: although the value of its state complexity is not known (only conjectured \cite{BJLRS16}), it can be expressed  in terms of tableaux of set partitions \cite{CLP20}. Regarding the complexities studied in this paper, the one of  star of  union  ( more generally type O operations, see Theorem \ref{th-sctypeO}) and  the one of star of  intersection  (more generally type A operations, see Theorem \ref{th-scA})  admit a closed form. We can also consider formula of Theorem \ref{th-sctypeX} as a quasi-closed form for the complexity of star of  symmetric difference and, more generally, of type X operations. Indeed, although the asymptotic properties of the $\alpha_{n,k}$ and $\alpha'_{n,k}$ numbers are not well known, we can compute these numbers far enough and relate them to other combinatorial numbers (see \cite{CLMP15}).

Throughout this paper, we can see that the main difficulty comes from the combinatorics of  tableaux that constitutes the bulk of the calculations. More generally, a deep understanding of the theory of modifiers  requires to study the actions of the transformation monoid of a finite set on combinatorial objects as  vectors, tableaux, sets, functions \textit{etc.} 

At this stage, we have not developed tools to determine \emph{alphabetic simplicity} that is the minimum size of the alphabet for the state complexity to be reached. Notice that, for an alphabet of a given size, the higher the alphabetic simplicity, the smaller the state complexity. The alphabetic simplicity is not necessarily a constant but depends on state complexities of the input languages. Although in his thesis, Edwin Hamel-de-le Court \cite{Ham20} has shown that, alphabetic simplicity for star of symmetrical difference is bounded by $17$, we  prefer to defer this study to a future work in which we will develop different tools and concepts related to the notion of alphabetic simplicity. 

We have investigated the computation of state complexity with respect to complete deterministic automata. Changing, even slightly,  characteristics of the automata encoding  input and/or output languages could have significant consequences on  state complexities. For instance, the state complexity of the shuffle product with respect to (not necessarily complete) deterministic automata is well known \cite{CKY02} whereas we have only a conjecture for complete deterministic automata \cite{BJLRS16,CLP20}. 
We can also define the state complexity with respect to other types of automata (non-deterministic, 2-ways etc). We think that there could exist, for each case, tools similar to  modifiers. In order to define them in a very general framework, it will probably be necessary to use notions from category theory. This is an ambitious scientific project which is only at the beginning and which will require the accumulation of many examples before leading to a consistent theory.

        

     \bibliography{biblio.bib}
\appendix
\section{Notations}
\small
\begin{longtable}{lll}
$\Acc(A)$&:& Restriction of the automaton $A$ to its accessible. See Section \ref{sec-LAAT}.\\ 
$\Acc_{\bullet}^{F_1,F_2}$&:& Set of accessible states of $\mathrm{M}_{\bullet}^{F_1,F_2}$.\\
$\SA^{F_1,F_2}_\bullet$&:&\begin{minipage}[t]{0.78\linewidth} Set of accessible saturated tableaux with respect to the final sets $F_1$ and $F_2$ and the boolean operation $\bullet$.\end{minipage}\\ \\
$\bullet$&:& A Boolean operation or a Boolean function. See Section \ref{sec-LAAT} and Table \ref{op-bool}.\\ \\
$\col_j(T)$&:&\begin{minipage}[t]{0.78\linewidth} Indices of the rows in which there are crosses in the column $i$ of the tableau $T$. See Definition \ref{def-rowcol}.\end{minipage}\\
$\mathrm{M}^{F_1,F_2}_\bullet$&:&  The DFA $\StBool_\bullet(\mon^{F_1,F_2}_{m,n})$.\\ 
$\widehat{\mathrm{M}}^{F_1,F_2}_\bullet$&:&\begin{minipage}[t]{0.78\linewidth} The DFA with the same alphabet, the same states, the same initial state and the same final states as $\mathrm M^{F_1,F_2}_\bullet$ but with the transition function $d$ defined by $d^{(f,g)}(T)=T\cdot (f,g)$. See Section \ref{section-saturation}.\end{minipage}\\ \\
$E_{i,j}=\{(i,j)\}$&:&\begin{minipage}[t]{0.78\linewidth} Tableau with only one cross at position $(i,j)$. See Section \ref{sec-valid}.\end{minipage}\\ \\
$\ostar$&:& Kleene star of a Boolean operation. See formula (\ref{eq-ostar}).\\ \\
$T\rightarrow T'$&:& One step of the local saturation. See Definition \ref{def-arrow}.\\
$\displaystyle\mathop\rightarrow^*$&:& Transitive closure of $\rightarrow$. See Definition \ref{def-arrow}.\\
$\LSV^{F_1,F_2}_\bullet$&:&\begin{minipage}[t]{0.78\linewidth} Set of  local saturated tableaux with respect to the final sates $F_1$ and $F_2$ and the boolean operation $\bullet$. See Definition \ref{def-arrow}.\end{minipage}\\
$\LST_{m,n}(p,q)$&:&\begin{minipage}[t]{0.78\linewidth} Set of the locally saturated $m\times n$ tableaux $T$, in the type X, satisfying $T\cap F=\emptyset$ for $F=(\IntEnt{m-p}\times\IntEnt q)\cup \left(\{m-p,m-p+1,\dots,m-1\}\times\{q,q+1,\dots,n-1\}\right)$. See Section \ref{sec-witnTypeX}.\end{minipage}
\\ \\
$\MergeC_{j_1,j_2}(T)$&:&\begin{minipage}[t]{0.78\linewidth} Tableau $T$ in which  columns $j_1$ and $j_2$ have been merged with respect to the operation of saturation. See Definition\ref{def-Merge}.\end{minipage}\\
$\MergeR_{i_1,i_2}(T)$&:&\begin{minipage}[t]{0.78\linewidth} Tableau $T$ in which rows $i_1$ and $i_2$ have been merged with respect to the operation of saturation. See Definition\ref{def-Merge}.\end{minipage}\\ \\
$\SuppC_j(T)$&:&\begin{minipage}[t]{0.78\linewidth} The tableau $T$ in which all crosses have been removed from column $j$. See Definition \ref{def-Supp}.\end{minipage}\\
$\SuppR_i(T)$&:&\begin{minipage}[t]{0.78\linewidth} The tableau $T$ in which all crosses have been removed from row $i$. See Definition \ref{def-Supp}.\end{minipage}
\\ \\ 
$\mathfrak{Bool}_\bullet$&:& Modifier of the Boolean operation $\bullet$. See Exemple \ref{ex-xor}.\\
$\mathfrak{St}$&:& Modifier of the Kleene star operation. See Example \ref{ex-star}.\\
$\StBool_\bullet$&:& Modifier of the Kleene star of the boolean operation $\bullet$. See Formula (\ref{eq-StBoolMod}).\\
$\mon^{F_1,\dots,F_k}_{n_1,\dots,n_k}$&:&\begin{minipage}[t]{0.78\linewidth} $k$-tuple of automata over a big alphabet encoding all the possible transitions. See Definition \ref{def-mon}.\end{minipage}\\ \\
 $T<T'$&:& Partial order on tableaux. See Definition \ref{def-order}.\\ \\
 $F_I$&:& Restricted final zone associated to $I$. See Definition \ref{def-restriction}.\\
 $\Ind_I(T)$&:& Inducted of $T$. See Definition \ref{def-restriction}.\\
$\Red(T)$&:& Reduced of $T$. See Definition \ref{def-restriction}.\\
$T|_I$&:& Restriction of the tableau $T$ to $I$. See Definition \ref{def-restriction}.\\
$\row_i(T)$&:&\begin{minipage}[t]{0.78\linewidth} Indices of the columns in which there are crosses in the row $i$ of the tableau $T$. See Definition \ref{def-rowcol}.\end{minipage}\\ \\
 $\Sat(T)$&:&\begin{minipage}[t]{0.78\linewidth} Unique saturated tableau in the Nerode class of $T$. See Definition \ref{def-sat}.\end{minipage}\\
 $\Sat(A)$&:& Automaton labelled by saturated tableaux. See Definition \ref{def-sat}.\\
 $\SV_\bullet^{F_1,F_2}$&:&\begin{minipage}[t]{0.78\linewidth} Set of saturated valid tableaux with respect to the final sets $F_1$ and $F_2$ and the boolean operation $\bullet$. See Formule (\ref{eq-SV}).\end{minipage}\\ \\
$\sc_\otimes$&:& State complexity of the operation $\otimes$. See Formula (\ref{eq-statecomp}).\\
$\Val^{F_1,F_2}_\bullet$&:& Set of valid tableaux labelling states of $\mathrm{M}^{F_1,F_2}_\bullet$.
See Section \ref{sec-valid}.\\
$\Val(\mathrm{M}^{F_1,F_2}_\bullet)$&:& Restriction of the DFA $\mathrm{M}_{\bullet}^{F_1,F_2}$ to its valid states. See Section \ref{sec-valid}.\\
\end{longtable}
\end{document}